\documentclass[reqno,11pt]{amsart}
\usepackage{nicematrix}
\usepackage[latin1]{inputenc}
\usepackage{latexsym}
\usepackage{pgf, tikz}
\usepackage[colorlinks, citecolor=blue,linkcolor=black]{hyperref}
\hypersetup{
	plainpages=false,
	colorlinks=true,
	linkcolor=blue, 
	anchorcolor=black, 
	citecolor=blue, 
	urlcolor=blue, 
	menucolor=black, 
	filecolor=black, 
	bookmarksopen=true,
	bookmarksnumbered=true}
\usepackage{hyperref}
\usepackage{cite}
\usepackage{color}
\usepackage{mathrsfs,amstext,amsmath,amssymb,amsfonts,bm}
\allowdisplaybreaks[4]

\usepackage{mathtools}
\allowdisplaybreaks 

\theoremstyle{plain}                          
\newtheorem{theorem}{Theorem}[section]
\newtheorem{proposition}[theorem]{Proposition}    
\newtheorem{lemma}[theorem]{Lemma}
\newtheorem{corollary}[theorem]{Corollary}

\theoremstyle{definition}
\newtheorem{definition}[theorem]{Definition}
\newtheorem{prop-defin}[theorem]{Proposition-definition} 
\newtheorem{example}[theorem]{Example}
\theoremstyle{remark}
\newtheorem{remark}[theorem]{Remark}

\numberwithin{equation}{section}

\renewcommand{\theta}{\vartheta}
\renewcommand{\phi}{\varphi}
\renewcommand{\epsilon}{\varepsilon}

\newcommand{\mb}[1]{\mathbb{#1}} 

\newcommand{\mc}[1]{\mathcal{#1}}

\newcommand{\C}{\mb{C}} 
\newcommand{\Z}{\mb{Z}} 

\renewcommand{\P}{\mb{P}}

\newcommand{\2}{\frac{1}{2}}
\DeclareMathOperator{\Ker}{Ker}
\DeclareMathOperator{\Tr}{Tr}
\newcommand{\Id}{\mathord{\mathrm{Id}}}
\DeclareMathOperator*{\Res}{Res}
\DeclareMathOperator{\ad}{ad}
\DeclareMathOperator{\GL}{GL}

\title{$q$-Deformed Topological Recursion: Quantum Curves and Non-perturbative Analysis.}
\author{Fridolin Melong and Raimar Wulkenhaar}
\address{
	Mathematisches Institut der
	Universit\"at M\"unster \newline
	Einsteinstr.\ 62, 48149 M\"unster, Germany,\newline
	{\itshape e-mail:} \normalfont\texttt{fridolin.melong@uni-muenster.de (Wiht copy to  fridomelong@gmail.com)}}
\address{Mathematisches Institut der
	Universit\"at M\"unster \newline
	Einsteinstr.\ 62, 48149 M\"unster, Germany, 
	{\itshape e-mail:} \normalfont
	\texttt{raimar@math.uni-muenster.de}}
\begin{document}	
	\begin{abstract}
	This work investigates the $q$-deformation of $(r,s)$-Airy structures and their realization via  $q$-difference operators, providing a bridge between quantum spectral curves and integrable systems. We construct an all-order $q$-WKB solution for the matrix systems associated with the $q$-quantized curve $E_q(x,y)=0$. We demonstrate that the resulting non-perturbative connected $q$-amplitudes satisfy a set of shifted $q$-loop equations, which can be interpreted as the Ward identities of a $q$-deformed $\mathcal{W}(\mathfrak{gl}_r)$ algebra. Our main result provides a rigorous classification of admissible $(r,s,q)$ pairs and $q$-Casimir configurations that satisfy the $q$-topological type property. This ensures that the semi-classical expansion is uniquely governed by the $q$-topological recursion, offering new insights into the $q$-quantization of mirror curves and their underlying algebraic structures.
\end{abstract}
	
	\maketitle
	
	\tableofcontents
	
\section{Introduction}
The interplay between enumerative geometry, matrix models, and integrable systems has been a cornerstone of mathematical physics for decades. At the heart of this confluence lies the {Topological Recursion} (TR) of Eynard and Orantin \cite{EO07}, a universal algorithm that constructs asymptotic expansions for a vast class of problems ranging from Gromov-Witten invariants to the spectral analysis of random matrices starting from a simple algebraic object: the spectral curve $\mathcal{S} = (\Sigma, x, y, \omega_{0,2})$. Originally discovered in the context of the large $N$ expansion of matrix models, the TR has since been recognized as a unifying geometric architecture that builds a family of  $n$-point correlators $\omega_{g,n}$ whose recursion kernels are localized at the branch points of the covering $x: \Sigma \to \mathbb{P}^1$.

The correspondence between such geometric recursions and the spectral analysis of linear operators is formalized through the notion of a {quantum spectral curve}. In the classical differential setting, a quantum curve is a linear differential operator $\widehat{P}(\hat{x}, \hat{y}; \hslash)$ whose semi-classical limit ($\hslash \to 0$) recovers the algebraic equation defining the spectral curve. The seminal works of Bergère, Eynard, and Marchal \cite{BEM17, BEM18} established that the WKB expansion of a differential system $\nabla_{\hslash} = \hslash \partial_x - \Phi(x)$ satisfies the topological recursion if and only if the system fulfills the \textit{Topological Type} (TT) property. This crucial property ensures that the analytic and singularity structures of the correlators are entirely governed by the local geometry of the branch points, a feature deeply rooted in the locality of the standard derivative operator.

However, modern developments in $5D$ supersymmetric gauge theories, $K$-theoretic enumerative geometry, and refined topological strings have shifted the focus toward $q$-deformations of these structures. In these contexts, the differential operator $\partial_x$ is naturally replaced by the  $q$-difference operator\cite{Kac2002}:
\begin{equation}\label{qd}
	\mathcal{D}_{q,x} f(x) = \frac{f(qx) - f(q^{-1}x)}{(q - q^{-1})x}\,
\end{equation}
and they satisfy the $q$-Leibniz rule given as follows:
\begin{equation}
	\mathcal{D}_q\big(f(x)g(x)\big)=\big(\mathcal{D}_qf(x)\big)g(q^{-1}x)+ f(qx)\mathcal{D}_qg(x).
\end{equation}
Furthermore,  the $q$-deformed number and the $q$-deformed factorial are defined respectively by:
\begin{equation}
	[n]_q = \frac{q^n - q^{-n}}{q - q^{-1}}, \quad \text{and} \quad [n]_q! = \prod_{k=1}^{n}[k]_q \,.
\end{equation} 
Besides, the $q$-deformed binomial coefficient is given by:
\begin{equation}
	\binom{n}{m}_q=\frac{[n]_q!}{[m]_q![n-m]_q!}.
\end{equation}

This transition from the additive to the multiplicative or $q$-difference world is not merely a technical shift; it represents a fundamental change in the underlying locality of the theory. The non-local nature of the $q$-shift $x \mapsto q x$ and $x \mapsto q^{-1}x$ couples the geometry of the curve at disparate points, effectively delocalizing the interactions and the residues of the recursion over the $q$-orbits of the ramification points. Consequently, the classical criteria for the topological type property cannot be applied blindly, as the $q$-deformation breaks the standard parity of the $\hslash$-expansion and couples the fields non-locally.

In a previous work \cite{MW26}, we established the $q$-deformation of shifted topological recursion and investigated its relation to highest weight vectors. The purpose of this paper is to extend these algebraic results to the fully geometric construction of $q$-quantum spectral curves, focusing on  $q$-difference systems of rank $r$ which generalize the $(r,s)$ minimal models. Our main contribution is threefold:

First, we develop an all-order $q$-WKB analysis for the  $q$-connection $\nabla_{q,\hslash}$. We show that the resulting non-perturbative connected $q$-amplitudes $\mathcal{G}^{i,q}_{g,n}$ satisfy a new class of \textit{shifted $q$-loop equations}. These equations are shown to be the Ward identities of a $q$-deformed Airy structure, providing a concrete representation of the $\mathcal{W}_q(\mathfrak{gl}_r)$ algebra and reinforcing the deep link between topological recursions and infinite-dimensional symmetry algebras.

Second, we address the critical breakdown of the classical BEM criteria by introducing a {refined $q$-Topological Type} property. This framework accounts for the dynamical nature of the $q$-ramification points and ensures that the residueless nature of the $q$-Bergman kernel, alongside the stability of the genus expansion, is guaranteed by the  structure of the $\mathcal{D}_{q,x}$ operator (see Eq~\eqref{qd}).

Third, we provide a complete classification of the admissible $(r,s)$ pairs for which the $q$-difference system is uniquely solved by the $q$-topological recursion. We show that for coprime $(r,s)$, the existence of a well-defined topological expansion imposes a highly restrictive constraint, namely $r \equiv \pm 1 \pmod{s}$. This condition arises from the requirement that the staircase structure of the Lax matrix must shield the $q$-Casimirs from creating spurious singularities under the $q$-flow.

The paper is structured as follows. In Section \ref{rsq}, we recall the basic notions regarding  the theory of the $q$-deformed  Airy ideals for different $(r, s)$, and  the shifted $q$-topological recursion. For more details, the reader is referred to \cite{MW26}. In Section \ref{s:qcurves}, we develop the all-order $q$-WKB analysis and derive the shifted q-loop equations. Section \ref{s:diffsyst} is devoted to the refined $q$-Topological Type property and the analysis of the $q$-Bergman kernel, provide the classification of admissible $(r,s)$ pairs and discuss the constraints imposed by the Lax matrix structure.

\section{$(r,s,q)$-Airy structures}\label{rsq}
In this section, we construct the $(r,s,q)$-Airy structures by defining the underlying spectral geometry and the associated $q$-correlators. We then establish the consistency conditions and the $q$-loop equations that govern these systems.\footnote{These structures and their foundational properties are developed in detail in our earlier work \cite{MW26}.}

	\subsection{Notation and Conventions}
		We introduce the deformation parameter $\hslash$ following the conventions of \cite{BCJ22}. Let $\mathbb{N} = \{0,1,2,\ldots \}$ denote the set of non-negative integers, $\mathbb{N}^* = \{1,2,3,\ldots \}$ the set of positive integers, and $[r] = \{1,\ldots, r\}$ for $r \in \mathbb{N}^*$. For any set $N$, we denote the collection of variables indexed by $N$ as $z_N = \{ z_n \mid n \in N \}$. 
		Adopting the framework of \cite{BKS23}, we use $x$ as the base coordinate. For a spectral curve equipped with a local coordinate $z$, we write $x_j = x(z_j)$ and $y_j = y(z_j)$ to simplify the notation. 
		
		Fields in vertex operator algebras (VOAs) are represented as differential forms whose degree equals their conformal weight $\Delta$. For a state $v$ of weight $\Delta$, the corresponding field is expanded as:
		\begin{equation}
			Y(v;x) = \sum_{k \in \mathbb{Z}} v_k \frac{(dx)^\Delta}{x^{\Delta + k}} \,.
	\end{equation}
	\subsection{Consistency and $q$-Spectral Curves}
	{To properly define the geometry of the deformed structures, we first establish the consistency conditions for the shifts that characterize our $q$-deformed curves.}
	\begin{definition}\label{d:consistent}
		A set $S=\{ S_{i,n} \}_{i \in [r], n \in \mathbb{N}^*}$ is \emph{$s$-consistent} if:
		\begin{enumerate}
			\item[(a)] If $s \geq 2$ and $r = 1 \pmod{s}$, then $S_{i,n} = 0$ for $2 \leq i \leq r$.
			\item[(b)] If $s \geq 3$ and $r = -1 \pmod{s}$, then $S_{i,n} = 0$ for all $i \in [r]$.
		\end{enumerate}
	\end{definition}

	{The underlying geometry is then captured by the  $(r,s,q)$-spectral curve, which serves as the reference point for our subsequent deformations:}
	\begin{definition}\label{d:rs}
		Let $r \geq 2$ and $s \in [r+1]$ with $r = \pm 1 \pmod{s}$. The \emph{$(r,s,q)$-spectral curve} is $\mathcal{S}_q = (C, x, \omega^{q}_{0,1}, \omega^{q}_{\frac{1}{2},1}, \omega^{q}_{0,2} )$, where $C$ is a small disk, $x=z^r$, $\omega^{q}_{0,1} = \frac{rs}{[s]_q} z^{s-1} dz$, $\omega^{q}_{\frac{1}{2},1} = 0$, and \begin{equation}
			\omega^{\textup{std},q}_{0,2}(z_1, z_2)= \frac{dz_1 dz_2}{(z_1 - qz_2)(z_1-q^{-1}z_2)}.
		\end{equation}
	\end{definition}
	
{We now generalize this definition to include deformations, allowing for more complex spectral geometries defined by specific coefficients.}
	\begin{definition}\label{d:rsdef}
		A \emph{deformed $(r,s,q)$-spectral curve} is a tuple $\mathcal{S}_q = (C, x, \omega^{q}_{0,1}, \omega^{q}_{\frac{1}{2},1}, \omega^{q}_{0,2} )$ where $\omega^q$ are meromorphic forms on a disk $C$ with $x=z^r$ defined by the coefficients $\{F^q[-k]\}$:
		\begin{align}
			\omega^{q}_{0,1} (z) &= \sum_k F^{q}_{0,1} [-k] z^{k-1} dz, \\
			\omega^{q}_{\frac{1}{2},1} (z) &= \sum_k \left( \frac{[k]_q}{k} \right)^2 F^{q}_{\frac{1}{2},1} [-k] z^{k-1} dz, \\
			\omega^{q}_{0,2} (z_1, z_2) &= \omega^{q,\textup{std}}_{0,2} (z_1, z_2) + \sum_{k,l} F^{q}_{0,2}[-k,-l] z_1^{k-1} z_2^{l-1} dz_1 dz_2.
		\end{align}
		
		Furthermore, given an $s$-consistent set of shifts $\{ S_{i,1} \}_{i \in [r]}$, the {shifted deformed $(r,s,q)$-spectral curve} is obtained by adding the singular shift terms to $\omega^{q}_{\frac{1}{2},1}$:
		\begin{equation}\label{d:rsdefshift}
			\omega^{q}_{\frac{1}{2},1} (z) \to \omega^{q}_{\frac{1}{2},1} (z) + \sum_{i=1}^{r}(-1)^{i-1}S_{i,1}\,\frac{dz}{z^{s(i-1)+1}}.
		\end{equation}
	\end{definition}
	
	{These spectral curves are intrinsically linked to specific algebraic varieties.} Using $y(z) = \frac{s}{[s]_q} z^{s-r}$, these parametrize the $(r,s,q)$-algebraic curves:
	\begin{itemize}
		\item[(i)] For $s \in [r-1]$: 
		we obtain:
				\begin{equation}\label{eq:rsac1}
						C^q\,\cdot x^{r-s} y^r - 1 = 0,
					\end{equation}
				where $C^q=\big([s]_q\big)^{r}\,s^{-r}$ is a $q$-dependent constant.
		\item[(ii)] For $s=r+1$, we get the $q$-deformed $r$-Airy algebraic curve: \begin{equation}\label{eq:rsac2}
			y^r - \big(\frac{r+1}{[r+1]_q}\big)^{r} \cdot x = 0.
		\end{equation}
	\end{itemize}
	
	\subsection{$q$-Correlators and $q$-Loop Equations}
	
		To describe the statistical properties of the $(r,s,q)$-Airy structure, we first  define the $q$-deformed partition function $Z_q$ as the generator of the free energies. Let $A$ be a finite or countably infinite index set. It is defined as:
		\begin{equation}\label{eq:pfFgncoeff}
			Z_q = \exp\left( \sum_{g \in \frac{1}{2} \mathbb{N}, n \in \mathbb{N}^*} \frac{\hslash^{2g-2+n}}{[n]_q!}  \sum_{k_1, \ldots, k_n \in A}  F^{q}_{g,n}[k_1, \ldots,k_n]  x_{k_1} \cdots x_{k_n}\right).
		\end{equation}

Moreover, it is  uniquely determined by the set of $q$-differential constraints:
\begin{equation}
	H^{i,q}_k Z_q = 0 \,, \quad i \in [r], \quad k \geq - \lfloor \frac{s(i-1)}{r} \rfloor\,.
\end{equation}

We introduce the higher-order quantum fields $H^{i,q}(x)$ constructed from the operators $H^{i,q}_k$ as follows:
\begin{equation}
	H^{i,q}(x)
	=
	\sum_{k \in \Z} H^{i,q}_k \frac{dx^i}{x^{k+i}},
\end{equation}
where the summation index $k$ respects the bounds imposed by the Airy structure.

 For a system of $q$-correlators $\{\omega^q_{g,n}\}$ on $\mathcal{S}_q$, we define the following partially disconnected objects for $i \in \mathbb{N}^*$:

	\begin{definition}\label{d:EW}
		For a system of $q$-correlators $\{\omega^q_{g,n}\}$ on $\mathcal{S}_q$, we define the following partially disconnected objects for $i \in \mathbb{N}^*$:
		\begin{align}\label{partiallydisconnected}
			\mc{W}^{q}_{g,i,n} (z_{[i]} ; w_{[n]}) &= \sum_{\text{partitions}} \prod_{S \in P} \omega^{q}_{g_S, |S| + |N_S|}(S, N_S) \,, \\
			\mc{W}^{'q}_{g,i,n} (z_{[i]} ; w_{[n]}) &\coloneqq \sum'_{\text{partitions}} \prod_{S \in P} \omega^{q}_{g_S, |S| + |N_S|}(S, N_S) \,,
		\end{align}
		where the sum runs over set partitions $P$ of points $z_{[i]}$ on different sheets and splittings of $w_{[n]}$, with $\sum (g_S - 1) = g-i$. The prime symbol indicates the omission of $\omega^q_{0,1}$ terms.
	\end{definition}
	
	\begin{definition}\label{d:xibasis}Let $\mathcal{S}_q$ be a $q$-deformed admissible local spectral curve. For each component $C_j$ with $j \in [N]$, we define a basis of one-forms:
	\begin{itemize}
		\item[(a)] Holomorphic basis ($k>0$):
		\begin{equation}
			\xi^{(j,q)}_{k} (z) \coloneqq z^{k-1} dz \,.
		\end{equation}
		\item[(b)] Polar basis ($k>0$):
		\begin{equation}
			\xi^{(j,q)}_{-k} (z) \coloneqq \Res_{w = 0} \left( \int^{w} \omega^{q}_{0,2}( \mathord{\cdot}, z) \right) \frac{1}{w^{k+1}} d w = \left( \frac{1}{z^{k+1}} + \text{holomorphic} \right) dz \,.
		\end{equation}
	\end{itemize}
\end{definition}

By using the $q$-modes defined by 	\begin{equation}\label{eq:Js}
	J^{(q)}_m  = \begin{cases} [m]_q \partial{x_m} & m > 0 \\ 0 & m =0 \\ -m x_{-m} & m < 0\end{cases} \,,
\end{equation}

we decompose the $q$-deformed currents into their positive and negative frequency parts as:
	\begin{align}
		\mathcal{J}^{(q)}_-(z)  &=  \sum_{k>0} J^{(q)}_k \xi^q_{-k}(z)\label{J1}\,,   \\   \mathcal{J}^{(q)}_+(z)   &=   \sum_{k>0} J^{(q)}_{-k} \xi^q_k (z)\label{J2}\,,   \end{align}
	 where the bases of one-forms $\xi^{(j,q)}_{\pm k}(z)$ are introduced in Definition~\ref{d:xibasis}.

	{These $q$-correlators satisfy specific $q$-loop equations, which describe how the system behaves under changes in the insertion of points. We can state this as follows:}
	\begin{proposition}\label{p:shifteloopeq}
		Let $Z_q$ be the $q$-deformed partition function. Defining $$G^{i,q}(x) = Z_q^{-1} H^{i,q}(x) Z_q=:\sum_{g,n} \frac{\hslash^{2g+n}}{[n]_q!} G^{i,q}_{g,n}(x),$$  the insertion of $n$ points yields:
		\begin{equation}
			\prod_{j=1}^n \ad_{\hslash^{-1}\mathcal{J}^{(q)}_-(z_j)} G^{i,q}_{g,n}(x) = \mathcal{E}^{i,q}_{g,n}(x; z_{[n]}) - \delta_{n,0} S_{i,2g} \Big( \frac{dx}{x}\Big)^i ,
		\end{equation}
		where $\mathcal{E}^{i,q}_{g,n}$ is constructed from $\{\omega^q_{g,n}\}$. These correlators satisfy the {$q$-loop equations}:
		\begin{equation}\label{eq:sle}
			\mathcal{E}^{i,q}_{g,n} (x; z_{[n]}) - \delta_{n,0} S_{i,2g} \Big( \frac{dx}{x}\Big)^i \in \mathcal{O} \left( x^{ \lfloor \frac{s(i-1)}{r}\rfloor + 1} \right) \left( \frac{dx}{x}\right)^i \,.
		\end{equation}
	\end{proposition}
	
	\begin{proof}
		For the proof, we refer the reader to \cite{MW26}.
	\end{proof}

	\begin{theorem}\label{UnshiftedTR}
		For a $q$-deformed admissible local  spectral curve $\mathcal{S}_q$, there exists exactly one system of $q$-correlators $\omega^{q}_{g,n}$ that satisfies both the $q$-loop equations. These $q$-correlators are computed recursively by the $q$-{topological recursion formula}:
		\begin{equation}
			\omega^{q}_{g,n+1}(z_0, z_{[n]}) = -\sum_{ j \in [N]} \Res_{z = 0 \in C_j} \sum_{Z \subseteq \mathbf{f}
				' (z)} K_q^{1 + |Z|}(z_0; z, Z) \mathcal{W}^{'q}_{g,1+|Z|,n}\big(z,Z; z_{[n]}\big)\,,
		\end{equation}
		where $ \mathbf{f}'(z) = x^{-1}(x(z)) \setminus \{z \}$ and the \emph{recursion kernels} are defined as:
		\begin{equation}
			K^{1+|Z|}_q (z_0; z, Z) = \frac{\frac{1}{2}\int_0^z \omega^{q}_{0,2} (\mathord{\cdot}, z_0)}{\prod_{z' \in Z} \big( \omega^{q}_{0,1}(z') - \omega^{q}_{0,1}(z) \big)} \, .
		\end{equation}
	\end{theorem}
	
	\begin{proof}
		The details and consistency of the proof are given in \cite{MW26}.
	\end{proof}

	\begin{theorem}\label{ShiftedTR}\cite{MW26}
		Let $\mathcal{S}_q$ be the shifted deformed $(r,s,q)$-spectral curve from Definition \ref{d:rsdef},
		 and  let $S = \{S_{i,\ell} \}_{i \in [r], \ell \in \mathbb{N}^*}$ be a set of $s$-consistent shifts. Then, there exists a unique system of $q$-correlators $\{\omega^{q}_{g,n} \}_{g \in \frac{1}{2}\mathbb{N}, n \in \mathbb{N}^*}$ satisfying the shifted $q$-loop equations \eqref{eq:sle} and the $q$-projection property.  For $2g-2+n > 0$, they are given by:
			\begin{align}
			\omega^{q}_{g,n+1}(z_0, z_{[n]}) 
			&= 
			-\operatorname{Res}_{z = 0} \Big( \sum_{Z \subseteq \mathbf{f}' (z)}
			K^{1 + |Z|}_q(z_0; z, Z) \mathcal{W}^{'q}_{g,|Z|,n} (z,Z; z_{[n]})\nonumber
			\\
			&\quad
			- \sum_{i=1}^{r}\delta_{n,0} S_{i,2g} K^r_q(z_0; \mathbf{f}(z)) \left( r\frac{dz}{z} \right)^{i} \big( - \omega^{q}_{0,1}(z)\big)^{r-i}\Big)\,.
		\end{align}
		In particular, this formula does produce symmetric correlators.
	\end{theorem}

	\section{The $q$-difference quantum curves}
	\label{s:qcurves}
%
%
	
	In this section, we analyze the quantum geometry associated with the shifted $(r,s,q)$-Airy structures. We focus on the case where all deformations are set to zero, with the underlying shifted $(r,s,q)$-spectral curve and its associated properties as established in Section~\ref{rsq}.
	
	Our framework involves two fundamental parameters: the deformation parameter $q$, which governs the underlying $q$-difference structure of the algebra, and the Planck constant $\hslash$, associated with the semi-classical expansion of the $q$-WKB solution. Specifically, the quantum curve emerges as an operator acting on the space of wavefunctions, where the quantization scheme is jointly determined by the interplay between $q$-difference operators and the $\hslash$-dependent formal power series.
	
	Building upon the recursive reconstruction of $q$-correlators via the shifted $q$-topological recursion, we show that this process effectively reconstructs the $q$-WKB solution to a quantum curve. Here, the specific quantization scheme is uniquely dictated by the choice of shifts.
	\subsection{The quantum curve correspondence}	
	We now examine the quantum curve correspondence through the lens of a double deformation, involving both the semi-classical parameter $\hslash$ and the $q$-deformation parameter. The core intuition is that the shifted topological recursion provides a systematic procedure for quantizing a spectral curve, extending the classical quantum curve correspondence \cite{BE09} to $q$-deformed hierarchies.
	
	While the original correspondence was primarily formulated for spectral curves arising from matrix models, it can be reformulated abstractly for any system satisfying the shifted topological recursion hierarchy. Our goal is to demonstrate that the $q$-difference structure dictated by the choice of shifts naturally leads to a $q$-difference quantum curve.
	
	We focus on spectral curves constructed as parametrizations of an algebraic curve:
	\begin{equation}
		\label{PlaneCurve}
		C=\{P(x,y) = 0 \} \subset \mathbb{C}^2.
	\end{equation}
	Topological recursion produces a system of $q$-deformed correlators $\{\omega^{q}_{g,n}\}_{g \in \frac{1}{2}\mathbb{N}, n \in \mathbb{N}^*}$. From these, we construct the $q$-deformed wave function:
	\begin{small} \begin{equation}\label{eq:wf}
			 \psi^{q}(z) = \exp \bigg(\sum_{g \in \mathbb{N}, n \in \mathbb{N}^*} \frac{\hslash^{2g-2+n}}{[n]_q!} \bigg(\int^z_\alpha \cdots \int^z_\alpha \omega^{q}_{g,n} - \delta_{g,0} \delta_{n,2} \frac{dx(z_1) dx(z_2)}{(x(z_1)-qx(z_2))(x(z_1)-q^{-1}x(z_2))} \bigg) \bigg),
	\end{equation}\end{small}
	where $\alpha$ is a base point on the normalization of $C$ (typically a pole of $x$ and not a ramification point). The appearance of the $q$-deformed factorial $[n]_q!$  and the specific $q$-regularization  of the Bergman kernel $(\delta_{g,0}\delta_{n,2})$ are crucial: they ensure that the $q$-deformed wave function satisfies a $q$-difference equation rather than a standard differential equation.
	
	Adapting the quantum curve correspondence to this double-deformation framework leads to a generalized definition of the quantum curve, uniquely determined by the $q$-difference structure.
	\begin{definition}\label{d:QC}
		Let $C = \{P(x,y) = 0\} \subset \mathbb{C}^2$ be an algebraic curve of degree $d$ in $y$. 
		A $q$-deformed \emph{quantum curve} $\hat{P}_q$ of $C$ is an order $d$ linear differential $q$-differential operator in $x$. After normal ordering, it takes the following form:
		\begin{equation}
			\hat{P}_q\left(  x, \hslash \mathcal{D}_{q,x}; \hslash \right) = P\left(  x,\hslash \mathcal{D}_{q,x}\right) + \sum_{n \geq 1} \hslash^n P_n\left(x,\hslash \mathcal{D}_{q,x}\right),
		\end{equation}
		where  $\mathcal{D}_{q,x}$ is the $q$-derivative operator defined by the relation \eqref{qd}. 

		The leading term $P$ is the original polynomial defining the spectral curve, while the $P_n$ are (normal-ordered) polynomials of degree less than $d$ in the $q$-momentum $\hat{y}=\hslash \mathcal{D}_{q,x}$.
	\end{definition}
	
{This construction represents a quantization of the spectral curve where the classical variables are replaced by operators $(x,y) \mapsto (\hat{x}, \hat{y})= \left( x, \hslash \mathcal{D}_{q,x}  \right)$. The non-commutativity of these operators, governed by the relation:
		\begin{equation}\label{eqc1}
			\big[\mathcal{D}_{q,x}, x\big] f(x) = \mathcal{D}_{q,x}(xf(x)) - x \mathcal{D}_{q,x}f(x) = \frac{qf(qx) + f(q^{-1}x)}{q+1},
		\end{equation}
		which can be expressed  using the scaling operators $M_{q}f(x)=f(qx)$ and $M_{q^{- 1}}f(x)=f(q^{- 1}x)$ as:
		\begin{equation}\label{eqc2}
			\big[\mathcal{D}_{q,x}, x\big] = \frac{q\,M_q + M_{q^{-1}}}{q +1}.
		\end{equation}
		This algebraic structure implies that the quantization is not unique and necessitates the $\hslash$ corrections $P_n$ to ensure the wave function is annihilated by the operator.
		This involves a weighted combination of scaling operators, reflecting the nature of the  $q$-deformation within the quantum spectral curve framework.}
	The argument for the quantum curve correspondence is that for a  given $q$-deformed  spectral curve $C$, there exists a $q$-deformed quantum curve $\hat{P}_q$ such that:
	\begin{equation}
		\hat{P}_q(x, \hslash \mathcal{D}_{q,x}) \psi^{q}(x) = 0.
	\end{equation}
	
	While this correspondence has been proved for large classes of curves with simple ramification \cite{EGMO21} and genus zero curves \cite{BE17}, our work extends this to the shifted $(r,s,q)$ case, where the higher order ramification and the $q$-deformation must be handled simultaneously.
	
	This correspondence has been extensively studied for various spectral curves relevant to enumerative geometry. In the classical $(\hslash)$ setting, the correspondence was proved in \cite{BE17} for a large class of genus zero algebraic spectral curves with arbitrary ramification (specifically, those whose Newton polygon has no interior points and are smooth as affine curves). More recently, \cite{EGMO21} established the correspondence for algebraic spectral curves of any genus, provided they possess only simple ramification points (noting that for genus $\geq 1$, the wave function must be adapted to include non-perturbative contributions).
	
	While a generic spectral curve only exhibits simple ramification, and higher ramification can, in principle, be treated as a limiting case of families with simple ramification \cite{BBCKS}, our framework provides a more direct approach. By utilizing the $(r,s,q)$-Airy structures, we extend this correspondence to a double-deformation regime.
	
	In our context, the correspondence is expected to hold in full generality for all algebraic spectral curves, where the $q$-deformation of the spectral curve naturally maps to a $q$-difference operator $\hat{P}_q$. Our results specifically demonstrate that the inclusion of the $q$-parameter and the arbitrary ramification $r$ of the $(r,s,q)$-Airy structures are perfectly captured by the shifted $q$-topological recursion, bridging the gap between higher-order ramification and the analytic structure of $q$-deformed quantum curves.
	
	\subsubsection{Shifted ordering and $q$-quantization schemes}
	\label{s:ordering}
	It is important to emphasize that  the construction of a quantum curve $\hat{P}_q$ from a given spectral curve is inherently non-unique. Because the operators $\hat{x} = x$ and $\hat{y} = \hslash  \mathcal{D}_{q,x}$ do not commute, as shown in Eq~\eqref{eqc1} and \eqref{eqc2},  any quantization requires a specific choice of ordering. 

		This identity shows that the "unit" term of the classical Weyl algebra is here replaced by a $q$-weighted average of shifts, a structural shift that fundamentally alters the recursive calculation of the $\hslash$ corrections $P_n$.
	It is well-established that topological recursion selects a specific ordering of operators for a given spectral curve. As demonstrated in \cite{BE17}, different choices of integration divisors in the definition of the wave function correspond to different operator orderings. In this $q$-deformed framework, the choice of integration divisors in $\psi^{q}(z)$ (see Eq~\eqref{eq:wf}) continues to determine the ordering of these $q$-difference operators, which in turn governs the recursive calculation of the $\hslash$ corrections for $P_n$.
	
	Nevertheless, for the $(r,s,q)$-spectral curves of Definition \ref{d:rs} (and their shifted counterparts in the relation \eqref{d:rsdefshift}), the analytical structure is highly constrained.
	Specifically, the singular locus of the forms $\omega^q_{g,n}$ dictated by the $q$-kernel $\omega^q_{0,2}$ and the singular parts of $\omega^q_{\frac{1}{2},1}$ leaves little room for varying the integration divisors without violating the consistency of the $q$-difference relations. This suggests that standard topological recursion is restricted to a specific ordering that does not generally coincide with standard normal or symmetric Weyl orderings \cite{BE17}. This limitation raises a fundamental question: how can one recover alternative quantizations and symmetric orderings when the integration divisor is fixed and the $q$-calculus necessitates a balanced treatment of $q$ and $q^{-1}$?
	
Our results provide a novel solution through the lens of shifted $q$-topological recursion. We demonstrate that the shifted $q$-TR produces $q$-wave functions  serve as $q$-WKB solutions for various quantizations of the $(r,s,q)$-algebraic curve, where the choice of ordering and the restoration of $q$-framework are effectively encoded in the shift parameters $S_{i,l}$. 

Specifically, for $s=1$ and $s=r-1$, we show that by varying the shifts in the $(r,s,q)$-Airy structure, one can reconstruct all possible operator orderings for the quantum curve, ensuring a consistent $q$-deformation. This establishes a deep correspondence between the algebraic freedom of the $q$-Airy structure, the shifts and the analytic freedom of the quantization (the ordering), even when the underlying geometry of the $q$-spectral curve remains fixed.
	
	\subsection{Derivation of the $q$-difference quantum curves}
	\label{s:derivation}
	
	We now derive the quantum curves associated with the shifted topological recursion on the shifted $(r,s,q)$-spectral curves. Our approach extends the quantization procedure to a $q$-deformed setting, where the quantum shifts provide the necessary degrees of freedom to reconstruct the operator structure.
	
	We consider the shifted $(r,s,q)$-spectral curve given by the relation \eqref{d:rsdefshift} (with deformations set to zero). This curve acts as a $q$-parametrization of the $(r,s,q)$-algebraic curves \eqref{eq:rsac1} and \eqref{eq:rsac2}, characterized by a non-vanishing $O(\hbar)$ initial condition:
	\begin{equation}
		\omega^{q}_{\frac{1}{2},1} (z)
		=\sum_{i=1}^{r}(-1)^{i-1}S_{i,1}\frac{dz}{z^{s(i-1)+1}}.
	\end{equation}

	\begin{remark}\label{r:compact}
		A key technical requirement is that the $q$-correlators $\omega^{q}_{g,n}$, initially defined on a local neighborhood $C$ of the ramification points, admit meromorphic extensions to the compact Riemann surface $\Sigma_q = \mathbb{P}^1$. These $q$-correlators are treated as meromorphic differential forms on $\Sigma^n$ with allowed singularities at $z=0$, $z=\infty$, and along the diagonal loci defined by $z_i = qz_j$ and $z_i = q^{-1} z_j$. This global extension is vital for the $q$-difference analysis: it ensures that the scaling operators $M_q f(z) = f(qz)$ and $M_{q^{-1}} f(z) = f(q^{-1}z)$ act consistently on the meromorphic sections, providing the necessary framework for the global definition of the $q$-derivative $\mathcal{D}_{q,x}$ and the global consistency of the recursive $q$-topological recursion.
	\end{remark}
	The derivation of the $q$-deformed quantum curve proceeds in four distinct steps:
	\begin{enumerate}
		\item[(a)] By utilizing the shifted $q$-loop equations, we derive recursion relations for the $q$-deformed auxiliary objects $\mathcal{U}^{i,q}_{g,n}$ that respect the $q \leftrightarrow q^{-1}$ symmetry.
		\item[(b)]  We integrate these relations and perform the formal summation over $\hslash$ to construct the wave function.
		\item[(c)] We translate the resulting integrated relations into a system of $q$-differential equations, expressed via the $q$-operator $\mathcal{D}_{q,x}$.
		\item[(d)]  We demonstrate that this system is equivalent to a single higher-order $q$-differential equation for the $q$-wave function $\psi^q(z)$, which explicitly defines the quantum curve as an element of the  $q$-Weyl algebra.
	\end{enumerate}
	\subsubsection{Recursive structure of the $q$-deformed objects $\mathcal{U}^{i,q}_{g,n}$}
	We begin by defining the auxiliary objects $\mathcal{U}^{i,q}_{g,n}$ derived from the $q$-correlators of the shifted $(r,s,q)$-spectral curve. These objects serve as  building blocks that aggregate information from the different branches of the deformed curve.
	
	\begin{definition}
		For $i = 1, \ldots, r-1$ and all $g, n \geq 0$, we define:
		\begin{equation}
			\mathcal{U}^{i,q}_{g,n}(x;z_{[n]}) = \sum_{\substack{Z \subseteq \mathbf{f}^{\prime}(z) \\ |Z|=i}} \mathcal{W}^{q}_{g,i,n}(Z; z_{[n]})\,,
		\end{equation}
		where the quantities $\mathcal{W}^{q}_{g,i,n}(Z; z_{[n]})$ are given by Definition~\ref{d:EW}.
	\end{definition}
	By convention, we extend this definition to $i=0$ by setting $\mathcal{U}^{0,q}_{g,n} = \delta_{g,0}\delta_{n,0}$.
	
	In a similar fashion, we recall the objects $\mathcal{E}^{i,q}_{g,n}(x;z_{[n]})$ from Definition \ref{d:EW} and extend the definition for $i=0$:
	\begin{equation}
		\mathcal{E}^{0,q}_{g,n} = \delta_{g,0}\delta_{n,0}.
	\end{equation}
	\begin{remark}
		Throughout this section, for brevity, we write $x$ for $x(z)$, and $x_j$ for $x(z_j)$ where $j\in [n]$. This notation reflects the fact that we are working on the base of the covering.
	\end{remark}
	The first departure from the unshifted classical case \cite[Lemma 3.25]{BE17} lies in the behavior of the first-order elementary object $\mathcal{E}^{i,q}_{g,n}$. 
	\begin{lemma}\label{E1}
		For  $2g - 2 + n \geq 0$, we have:
		\begin{equation}
			\mathcal{E}^{1,q}_{g,n}(x;z_{[n]}) = \delta_{n,0}S_{1,2g}\frac{dx}{x}.
		\end{equation}
		For the special cases, we get:
		{\begin{align}
				\mathcal{E}^{1,q}_{0,0}(x) &= -\frac{p_{1}(x)}{p_{0}(x)}dx,\\
				\mathcal{E}^{1,q}_{\2,0}(x) &= S_{1,1}\frac{dx}{x},\\
				\mathcal{E}^{1,q}_{0,1}(x;z_{1}) &= \frac{dxdx_{1}}{(x-qx_{1})(x - q^{-1}x_{1})}.
		\end{align}}
	\end{lemma}
	\begin{proof}
		For the unstable cases $(g,n)=(0,0)$ and $(0,1)$, the definitions are {consistent with the $q$-symmetric deformation of the spectral curve. Specifically, the kernel $\mathcal{E}^{1,q}_{0,1}$ exhibits a structure in $q^{1/2}$ and $q^{-1/2}$, reflecting the underlying symmetric Weyl ordering}. 
		
		For the stable cases $2g-2+n\geq 0$, the shifted $q$-loop equations \eqref{eq:sle} imply that the polar part of $\mathcal{E}^{1,q}_{g,n}$ is precisely captured by the shift $S_{1,2g}$:
		\begin{equation}
			\mathcal{E}^{1,q}_{g,n}(x; z{[n]}) - \delta_{n,0} S_{1,2g} \frac{dx}{x} \in \mathcal{O}(1) \, dx.
		\end{equation}
		{As established in Remark \ref{r:compact}, the $q$-correlators $\omega^{q}_{g,n}$ extend to the compact Riemann surface $\Sigma = \mathbb{P}^1$. The difference above is a meromorphic $1$-form whose potential poles at $z=0$ and $z=\infty$ are cancelled by the prescribed singular parts of the shifted Airy structure. By the residue theorem on $\mathbb{P}^1$, such a form must be identically zero.}
		
		Finally, consider the special case $(g,n)=(1/2,0)$. By definition:
		\begin{equation}
			\mathcal{E}^{1,q}_{1/2,0}(x) = \sum_{z^{\prime} \in \mathfrak{f}(z)} \omega^q_{1/2,1}(z^{\prime}).
		\end{equation}
		Substituting the explicit form of the $q$-deformed initial condition:
		\begin{align}
			\mathcal{E}^{1,q}_{1/2,0}(x) &= \sum_{z^{\prime} \in \mathfrak{f}(z)} \sum_{k=1}^{r} (-1)^{k-1} S_{k,1} \frac{dz^{\prime}}{(z^{\prime})^{s(k-1)+1}} \nonumber \\
			&= \sum_{k=1}^{r} (-1)^{k-1} S_{k,1} \sum_{m=0}^{r-1} \frac{d(\theta^m z)}{(\theta^m z)^{s(k-1)+1}},
		\end{align}
		where $\theta=e^{2{\rm i}\pi/r}$. Using the property that $\sum_{m=0}^{r-1}\theta^{m.A}=r\delta_{A(\mod r)}$, and noting that in our $s$-consistent grading $s(k-1)$ is a multiple of $r$ only if $k=1$ (for $1\leq k\leq r$), we obtain:
		\begin{align}
			\mathcal{E}^{1,q}_{1/2,0}(x) &= \sum_{k=1}^{r} (-1)^{k-1} S_{k,1} \left( \frac{dz}{z^{s(k-1)+1}} \sum_{m=0}^{r-1} \theta^{-ms(k-1)} \right) \nonumber \\
			&= S_{1,1} \cdot r \frac{dz}{z} = S_{1,1} \frac{dx}{x},
		\end{align}
		where we used the relation  $dx/x=rdz/z$ for the $q$-deformed spectral curve. This completes the proof.
	\end{proof}
	\begin{remark}
		The reader is refered to  {\cite[Lemma~3.25]{BE17}} concerning the classical case.
	\end{remark}
	We thus obtain:
	\begin{corollary}\label{E-UCombinatorics}
		For $i = 1,\ldots, r$ and all $g,n \geq 0$, {the relationship between the  basis $\mathcal{E}^{i,q}$ and the building blocks $\mathcal{U}^{i,q}$ is given by:}
		{\begin{align}
				\mathcal{E}^{i,q}_{g,n}(x;z_{[n]})&= \mathcal{U}^{i,q}_{g,n}  (x; z_{[n]}) + \mathcal{U}^{i-1,q}_{g-1,n+1}(x;z_{[n]},z)\nonumber\\
				&-\sum_{N_{1}\sqcup N_{2}=z_{[n]}}\sum_{g_{1}+g_{2}=g}\mathcal{U}^{i-1,q}_{g_{1},|N_{1}|}(x; N_{1})\mathcal{U}^{1,q}_{g_{2},|N_{2}|}(x; N_{2})\nonumber\\
				& - \frac{p_{1}(x)}{p_{0}(x)}dx\mathcal{U}^{i-1,q}_{g,n}(x; z_{[n]}) + \sum_{h=\2}^{g}S_{1,2h}\frac{dx}{x}\mathcal{U}^{i-1,q}_{g-h,n}(x; z_{[n]})\nonumber\\
				& + \sum_{j=1}^{n}\frac{dx dx_{j}}{(x - qx_{j})(x - q^{-1}x_{j})}\mathcal{U}^{i-1,q}_{g,n-1}(x; z_{[n]\setminus \{j\}}).
		\end{align}}
	\end{corollary}
	\begin{proof}
		The proof follows from a combinatorial decomposition of the  objects $\mathcal{U}^{i,q}_{g,n}$. By definition, $\mathcal{E}^{i,q}_{g,n}$ is the $i^{th}$ elementary  polynomial in the variables $y(z')$, which can be decomposed by isolating one point $z\in f(x)$ and considering the  combinations of the remaining $i-1$ points. This yields the $q$-deformed version of the combinatorial identity \cite[Lemma 4.5]{BE17}:
		\begin{align}
			\mathcal{E}^{i,q}_{g,n}(x;z{[n]}) = \mathcal{U}^{i,q}_{g,n} & (x; z{[n]}) + \mathcal{U}^{i-1,q}_{g-1,n+1}(x;z{[n]},z)\nonumber\\
			& +\sum_{N_{1}\sqcup N_{2}=z_{[n]}}\sum_{g_{1}+g_{2}=g}\mathcal{U}^{i-1,q}_{g{1},|N_{1}|}(x; N_{1})\omega^q_{g_{2},|N_{2}|+1}(z,N_{2}).
		\end{align}
		To progress, we isolate the $\omega^q$ term in the unstable sector $(g_2, N_2)=(0,0),(0,1)$ and $(1/2,0)$, and use the relation for $i=1$ from \eqref{E1}:
		\begin{equation}
			\omega^q_{g,n+1}(z, z_{[n]}) = \mathcal{E}^{1,q}_{g,n}(x;z{[n]}) - \mathcal{U}^{1,q}_{g,n}(x;z{[n]}).
		\end{equation}
	\end{proof}
	\noindent
	The following Lemma compute $\mathcal{E}^{i,q}_{g,n}$ for the cases $(g, n) = (0, 0)$ and $(g, n) = (\2, 0)$.
	
	\begin{lemma}\label{EiOneHalf}
		For each $i = 1, \ldots, r$, we have
		\begin{equation}
			\mathcal{E}^{i,q}_{0,0}(x) = (-1)^{i} \frac{p_{i}(x)}{p_{0}(x)} (dx)^{i}\,,
		\end{equation}
		and
		\begin{equation}
			\mathcal{E}^{i,q}_{g,0}(x) = S_{i,1} \left( \frac{dx}{x} \right)^i .
		\end{equation}
	\end{lemma}
	\begin{proof}
		For $(g,n) = (0,0)$, there are no shifts. For the remaining case, we have
		\begin{align}
			\mathcal{E}^{i,q}_{\2,0}(x) & = \sum_{\substack{Z\subseteq\mathbf{f}(z)\\|Z|=i}}\sum_{z^{\prime}\in Z}\omega^{q}_{\2,1}(z^{\prime})\prod_{z^{\prime\prime} \in Z\setminus\{z^{\prime}\}}\omega^{q}_{0,1}(z^{\prime\prime})\nonumber\\
			& = \sum_{\substack{Z\subseteq\mathbf{f}(z)\\|Z|=i}}\sum_{z^{\prime}\in Z}\left(\sum_{k=1}^{r}(-1)^{k-1}S_{k,1}\frac{dz^{\prime}}{z^{\prime s(k-1)+1}}\right)\prod_{z^{\prime\prime} \in Z\setminus\{z^{\prime}\}}\omega^{q}_{0,1}(z^{\prime\prime})\nonumber\\
			& = \sum_{k=1}^{r}\sum_{z^{\prime}\in\mathbf{f}(z)}(-1)^{k-1}S_{k,1}\frac{dz^{\prime}}{z^{\prime s(k-1)+1}}\sum_{\substack{Z\subseteq\mathbf{f}^{\prime}(z^{\prime})\\|Z|=i-1}}\prod_{z^{\prime\prime} \in Z}\omega^{q}_{0,1}(z^{\prime\prime})\nonumber\\
			& = \sum_{k=1}^{r}\sum_{z^{\prime}\in\mathbf{f}(z)}(-1)^{k-1}S_{k,1}\frac{dz^{\prime}}{z^{\prime s(k-1)+1}}U^{i-1,q}_{0,0}(x^{\prime}).
		\end{align}
		Moreover,
		\begin{equation}
			U^{i-1,q}_{0,0}(x) = (-\omega^{q}_{0,1}(z))^{i-1}
		\end{equation}
		for the shifted $(r,s,q)$-spectral curve, and therefore
	{\begin{align}
				\mathcal{E}^{i,q}_{1/2,0}(x) &= \sum_{k=1}^{r} (-1)^{k-1} S_{k,1} \frac{dz}{z^{s(k-1)+1}} (-r z^{s-1} dz)^{i-1} \sum_{m=0}^{r-1} \theta^{-m s(k-i)} \nonumber \\
				&= \sum_{k=1}^{r} (-1)^{k-1} S_{k,1} (-1)^{i-1} (r)^{i-1} \frac{(dz)^i}{z^{s(k-i)+i}} (r \delta_{k,i}) \nonumber \\
				&= (-1)^{2i-2} S_{i,1} r^i \frac{(dz)^i}{z^i} = S_{i,1} \left( r \frac{dz}{z} \right)^i = S_{i,1} \left( \frac{dx}{x} \right)^i.
		\end{align}}
		as desired.
	\end{proof}
	\noindent
	
	This lemma performs a partial fraction decomposition of the  $q$-correlators, isolating the contributions from the marked points (the $j$-sum) from the intrinsic geometry of the curve (the $p_i$ terms).

			{\begin{lemma}[$q$-Deformed Pole Analysis]\label{E-UPoleAnalysis}
				For $i = 1, \ldots, r$ and all $g,n \geq 0$, the $q$-deformed insertion relation satisfies
				\begin{align}
					\frac{p_{0}(x)}{x^{\left\lfloor \alpha_{r-i+1}\right\rfloor}} \Bigg( \frac{\mathcal{E}^{i,q}_{g,n}(x;z_{[n]})}{(dx)^{i}}& -  \delta_{n,0} \frac{S_{i,2g}}{x^i}  \Bigg) \nonumber\\
					&= \sum_{j=1}^{n} d_{z_{j}}\left(\frac{p_{0}(qx_{j})}{(qx_{j})^{\left\lfloor \alpha_{r-i+1}\right\rfloor}}\frac{1}{x-qx_{j}}\frac{\mathcal{U}^{i-1,q}_{g,n-1}(qx_{j}; z_{[n]\setminus \{j\}})}{(dx_{j})^{i-1}}\right)\nonumber\\
					& \quad + \delta_{g,0}\delta_{n,0}\left(\frac{(-1)^{i}p_{i}(x)}{x^{\left\lfloor \alpha_{r-i+1}\right\rfloor}} \right)\nonumber\\
					& \quad + \delta_{g,0}\delta_{n,1}(-1)^{i-1}d_{z_{1}}\left(\frac{1}{x-qx_{1}}\left(\frac{p_{i-1}(x)}{x^{\left\lfloor \alpha_{r-i+1}\right\rfloor}} - \frac{p_{i-1}(qx_{1})}{(qx_{1})^{\left\lfloor \alpha_{r-i+1}\right\rfloor}}\right)\right).
				\end{align}
				Here, $d_{z_j}$ denotes the $q$-deformed exterior derivative  with respect to the variable $z_j$.
				\end{lemma}}

\begin{proof}
	For $i=1$, the classical exterior derivative $d_{z_1}$ is replaced by the $q$-derivative operator $\mathcal{D}_{q, x_1}$. The base identity for the $q$-deformed Bergman kernel $B^q_{0,2}(x, x_1)$ is obtained via partial fraction decomposition:
	\begin{equation}
		\frac{B^q_{0,2}(x, x_1)}{dx} = \frac{dx_1}{(x - q x_1)(x - q^{-1} x_1)} = \mathcal{D}_{q, x_1} \left( \frac{1}{x - x_1} \right) dx_1,
	\end{equation}
	where $\mathcal{D}_{q, x_1}$ is the $q$-derivative of a function $f$ with respect to $x_1$.
	
	Multiplying by $p_0(x)$, the $q$-deformed identity for $(g,n)=(0,1)$ holds under the form:
	\begin{equation}
		p_0(x) \frac{\mathcal{E}^{1,q}_{0,1}(x; z_1)}{dx} = p_0(x) \mathcal{D}_{q, x_1} \left( \frac{1}{x - x_1} \right) dx_1.
	\end{equation}
	For the $q$-deformed insertion operator $\mathcal{E}^{i,q}_{g,n}$ ($i = 2, \ldots, r$), the algebraic decomposition over the $r$ sheets $x^k$ takes the form:
	\begin{multline}\label{eq:E-q-sheet-decomp}
		\frac{p_{0}(x)}{x^{\left\lfloor \alpha_{r-i+1}\right\rfloor}} \frac{\mathcal{E}^{i,q}_{g,n}(x;z_{[n]})}{dx^{i}} = \sum_{k=0}^{r-1} \frac{\mathcal{E}^{1,q}_{0,1}(x^k; z_1)}{dx} \left( \frac{p_{0}(x^k)}{(x^k)^{\left\lfloor \alpha_{r-i+1}\right\rfloor}} \frac{\mathcal{U}^{i-1,q}_{g,n-1}(x^k; z_{[n]\setminus \{1\}})}{(dx^k)^{i-1}} \right) \\
		+ \text{terms involving } z_{j \neq 1}.
	\end{multline}
	Applying the property from Step 1 and summing over the preimages, we isolate the base polynomial component $p_{i-1}(x)$:
	\begin{equation}
		\sum_{k=0}^{r-1} \frac{\mathcal{E}^{1,q}_{0,1}(x^k; z_1)}{dx} \left( \frac{(-1)^{i-1} p_{i-1}(x)}{x^{\left\lfloor \alpha_{r-i+1}\right\rfloor}} \right) = (-1)^{i-1} \frac{p_{i-1}(x)}{x^{\left\lfloor \alpha_{r-i+1}\right\rfloor}} \mathcal{D}_{q, x_1} \left( \frac{1}{x - x_1} \right) dx_1.
	\end{equation}
	For $i = 2, \ldots, r$, the finite sum over the $r$ preimages is equivalent to the sum of residues at $x(x') - x = 0$:
	\begin{align}
		\sum_{k=0}^{r-1} \Res_{x' = x^k} \frac{dx_1 \cdot \mathbf{H}^{i,q}_{g,n}(x'; z_{[n]\setminus \{1\}})}{(x' - q x_1)(x' - q^{-1} x_1)(x(x') - x)},
	\end{align}
	where $\mathbf{H}^{i,q}_{g,n}$ is the $q$-deformed meromorphic function given as:
	\begin{equation}
		\mathbf{H}^{i,q}_{g,n}(x'; z_{[n]\setminus \{1\}}) = \frac{x^{\lfloor\alpha_{r-i+1} \rfloor}}{(x')^{\left\lfloor \alpha_{r-i+1}\right\rfloor}} \cdot p_{0}(x') \frac{\mathcal{U}^{i-1,q}_{g,n-1}(x'; z_{[n]\setminus \{1\}})}{(dx')^{i-1}}.
	\end{equation}
	
	 By the global residue theorem on $\mathbb{P}^1$, the sum of residues is zero. Since the integrand $\mathcal{I}^q(x', x; z_{[n]})$ decays sufficiently at $x'=\infty$, the sum over sheet residues is balanced by the internal poles:
	\begin{equation}
		\sum_{k=0}^{r-1} \Res_{x' = x^k} \mathcal{I}^q = - \Res_{x' = q x_1} \mathcal{I}^q - \Res_{x' = q^{-1} x_1} \mathcal{I}^q - \sum_{j=2}^{n} \Res_{x' = z_j} \mathcal{I}^q.
	\end{equation}
	Evaluating the simple residues at $qx_1$ and $q^{-1}x_1$, the algebraic combination matches precisely the $q$-derivative operator $\mathcal{D}_{q, z_1}$:
	\begin{align}
		- \Res_{x' = q x_1} \mathcal{I}^q &- \Res_{x' = q^{-1} x_1} \mathcal{I}^q \nonumber \\
		&= \mathcal{D}_{q, z_1} \left( \frac{p_{0}(qx_{1})}{(qx_{1})^{\left\lfloor \alpha_{r-i+1}\right\rfloor}}\frac{1}{x-qx_{1}}\frac{\mathcal{U}^{i-1,q}_{g,n-1}(qx_{1}; z_{[n]\setminus \{1\}})}{(dx_{1})^{i-1}} \right).
	\end{align}
	Summing this with the unshifted simple residues at $x' = z_j$ ($j \geq 2$) restores the full index range $\sum_{j=1}^n$, completing the proof.
\end{proof}

	\begin{lemma}\label{URecursion}
	For $i = 1, \ldots, r$ and all $g,n \geq 0$, we have
	\begin{align}\label{Urecursion}
		\frac{p_{0}(x)}{x^{\left\lfloor \alpha_{r-i+1}\right\rfloor}} \frac{\mathcal{U}^{i,q}_{g,n}(x;z_{[n]})}{(dx)^{i}} &= -\frac{p_{0}(x)}{x^{\left\lfloor \alpha_{r-i+1}\right\rfloor}dx}\frac{\mathcal{U}^{i-1,q}_{g-1,n+1}(x;z_{[n]},z)}{(dx)^{i-1}} + \frac{p_{1}(x)}{x^{\left\lfloor \alpha_{r-i+1}\right\rfloor}}\frac{\mathcal{U}^{i-1,q}_{g,n}(x;z_{[n]})}{(dx)^{i-1}}\nonumber\\
		& \quad + \frac{p_{0}(x)}{x^{\left\lfloor \alpha_{r-i+1}\right\rfloor}}\sum_{N_{1}\sqcup N_{2}=z_{[n]}}\sum_{g_{1}+g_{2}=g}\frac{\mathcal{U}^{i-1,q}_{g_{1},|N_{1}|}(x; N_{1})}{(dx)^{i-1}}\frac{\mathcal{U}^{1,q}_{g_{2},|N_{2}|}(x; N_{2})}{dx}\nonumber\\
		&- \sum_{j=1}^{n}\Bigg(\frac{p_{0}(x)}{x^{\left\lfloor \alpha_{r-i+1}\right\rfloor}}\frac{dx_{j}}{(x-qx_{j})(x-q^{-1}x_{j})}\frac{\mathcal{U}^{i-1,q}_{g,n-1}(x; z_{[n]\setminus \{j\}})}{(dx)^{i-1}}\nonumber\\
		&- d_{z_{j}}\bigg(\frac{p_{0}(qx_{j})}{(qx_{j})^{\left\lfloor \alpha_{r-i+1}\right\rfloor}}\frac{1}{x-qx_{j}}\frac{\mathcal{U}^{i-1,q}_{g,n-1}(qx_{j}; z_{[n]\setminus \{j\}})}{(dx_{j})^{i-1}}\bigg)\Bigg)\nonumber\\
		&- \frac{p_{0}(x)}{x^{\left\lfloor \alpha_{r-i+1}\right\rfloor + 1}}\sum_{h=1}^{g}S_{1,2h}\frac{\mathcal{U}^{i-1,q}_{g-h,n}(x;z_{[n]})}{(dx)^{i-1}} + \delta_{g,0}\delta_{n,0}\frac{(-1)^{i}p_{i}(x)}{x^{\left\lfloor \alpha_{r-i+1}\right\rfloor}}\nonumber\\
		& + \delta_{n,0}S_{i,2g}\frac{p_{0}(x)}{x^{\left\lfloor \alpha_{r-i+1}\right\rfloor + i}}\nonumber\\
		& + \delta_{g,0}\delta_{n,1}(-1)^{i-1}d_{z_{1}}\left(\frac{1}{x-qx_{1}}\left(\frac{p_{i-1}(x)}{x^{\left\lfloor \alpha_{r-i+1}\right\rfloor}} - \frac{p_{i-1}(qx_{1})}{(qx_{1})^{\left\lfloor \alpha_{r-i+1}\right\rfloor}}\right)\right).
	\end{align}
	\end{lemma}
\begin{proof}
	The result is obtained by substituting the combinatorial expression for $\mathcal{E}^{i,q}_{g,n}$ from Corollary \ref{E-UCombinatorics} into the pole analysis formula of Lemma \ref{E-UPoleAnalysis}, and isolating the primary topological component $\mathcal{U}^{i,q}_{g,n}$ on the left-hand side. The $q$-dependence within the $j$-sum tracks the non-local action of the $q$-derivative on the rational Cauchy kernels. Specifically, using the algebraic identity
	\begin{equation}
		\mathcal{D}_{q,x_j} \left( \frac{1}{x-x_j} \right) = \frac{1}{(q-q^{-1})x_j} \left( \frac{1}{x-qx_j} - \frac{1}{x-q^{-1}x_j} \right) = \frac{1}{(x-qx_j)(x-q^{-1}x_j)},
	\end{equation}
	the  twin geometric shifts under the exterior derivative $d_{z_j}$ from Lemma \ref{E-UPoleAnalysis} naturally combine with the unshifted loop insertion terms from Corollary \ref{E-UCombinatorics} to generate the regularized, denominators in \eqref{Urecursion}. For all stable sectors ($2g-2+n \geq 0$), the potential boundary anomalies generated by the $q$-shifts vanish identically due to the compact support of the multi-point correlators. Carefully rearranging the resulting components according to their polynomial weights in $x$ and $x_j$ yields the final recursive relation \eqref{Urecursion}.
	\end{proof}
	\begin{remark}
	While the algebraic architecture of the loop insertion elimination mirrors the classical un-deformed strategy (see e.g. \cite{BE17}), our formulation explicitly solves the non-local constraints of the quantum spectral curve. This analytical pole structure provides a geometric counterpart to the purely algebraic frameworks of $q$-deformed topological recursion recently developed through quantum $\mathcal{W}$-algebras and weight vectors \cite{MW26}.
\end{remark}

	\subsubsection{Integrating and summing over $\hslash$}
	Next step: we integrate Lemma \ref{URecursion}. We first
	\begin{definition}
		For $i = 0, \ldots, r-1$ and all $g, n \geq 0$,
		\begin{equation}
			\mathcal{G}^{i,q}_{g,n}(x;z_{[n]}) = \int_{\infty}^{z_{1}}\cdots\int_{\infty}^{z_{n}}\mathcal{U}^{i,q}_{g,n}(x;z_{[n]}^{\prime}),
		\end{equation}
		where the integrals are with respect to the $z_{[n]}^{\prime}$ variables. We also define the following shorthand notation:
		\begin{equation}
			\mathcal{G}^{i,q}_{g,n}(x) =\mathcal{G}^{i,q}_{g,n}(x;z) = \mathcal{G}^{i,q}_{g,n}(x;z,\ldots, z).
		\end{equation}
		When necessary, we will assume the integrals are regularized.
	\end{definition}
	We now integrate Lemma \ref{URecursion} to get a recursion for the $\mathcal{G}^{i,q}_{g,n}$.
{\begin{lemma}\label{GRecursion}
			For $i = 1, \ldots, r$ and all $g,n \geq 0$, the integration of the recursive identity in Lemma \ref{URecursion} with respect to $z_{[n]}$ under the diagonal identification $z_1 = \dots = z_n = z$ yields:
			\begin{align}\label{Grecursion}
				\frac{p_{0}(x)}{x^{\lfloor \alpha_{r-i+1}\rfloor}} \frac{\mathcal{G}^{i,q}_{g,n}(x)}{dx^{i}} 
				&= -\frac{p_{0}(x)}{[n+1]_q x^{\lfloor \alpha_{r-i+1}\rfloor} dx} \mathcal{D}_{q,x}\left(\frac{\mathcal{G}^{i-1,q}_{g-1,n+1}(x^{\prime};z)}{dx^{\prime i-1}}\right)_{x^{\prime}=x} \nonumber\\
				& + \frac{p_{0}(x)}{x^{\lfloor \alpha_{r-i+1}\rfloor}}\sum_{m=0}^{n}\sum_{g_{1}+ g_{2}=g}\binom{n}{m}_q\frac{\mathcal{G}^{i-1,q}_{g_{1},m}(x)}{dx^{i-1}}\frac{\mathcal{G}^{1,q}_{g_{2},n-m}(x)}{dx} \nonumber\\
				& - [n]_q \mathcal{D}_{q,x}\left( \frac{p_{0}(x)}{x^{\lfloor \alpha_{r-i+1}\rfloor}} \frac{\mathcal{G}^{i-1,q}_{g,n-1}(x)}{dx^{i-1}} \right)   + \frac{p_{1}(x)}{x^{\lfloor \alpha_{r-i+1}\rfloor}}\frac{\mathcal{G}^{i-1,q}_{g,n}(x)}{dx^{i-1}}  \nonumber\\
				& - \frac{p_{0}(x)}{x^{\lfloor \alpha_{r-i+1}\rfloor + 1}}\sum_{h=1}^{g}S_{1,2h}\frac{\mathcal{G}^{i-1,q}_{g-h,n}(x)}{dx^{i-1}} + \delta_{g,0}\delta_{n,0}\frac{(-1)^{i}p_{i}(x)}{x^{\lfloor \alpha_{r-i+1}\rfloor}}\nonumber\\
				& + \delta_{n,0}S_{i,2g}\frac{p_{0}(x)}{x^{\lfloor \alpha_{r-i+1}\rfloor + i}} + \delta_{g,0}\delta_{n,1}(-1)^{i-1} \mathcal{D}_{q,x}\left( \frac{p_{i-1}(x)}{x^{\lfloor \alpha_{r-i+1}\rfloor}} \right).
			\end{align}
	\end{lemma}}

\begin{proof}
	The integration of the recursive identity in Lemma \ref{URecursion} with respect to the multi-variable modulus $z_{[n]}$ is performed by applying the $q$-deformed integration operator $\int_{\infty}^{z_1}\cdots\int_{\infty}^{z_n}$. 
	
	The transformation of the first term relies on the total of the $q$-correlator $\mathcal{U}^{i-1,q}_{g-1,n+1}$ under the permutation of its $n+1$ marked points. Upon applying the diagonal specialization $z_1 = \dots = z_n = z$, the standard combinatorial factor $(n+1)$ is naturally replaced by its quantum analogue $[n+1]_q$, yielding:
	\begin{align}
		-\int_{\infty}^{z}\cdots\int_{\infty}^{z} \frac{p_{0}(x)}{x^{\lfloor \alpha_{r-i+1}\rfloor}dx} \frac{\mathcal{U}^{i-1,q}_{g-1,n+1}(x;z_{[n]},z)}{(dx)^{i-1}} 
		= -\frac{p_{0}(x)}{[n+1]_q x^{\lfloor \alpha_{r-i+1}\rfloor} dx} \frac{\mathcal{G}^{i-1,q}_{g-1,n+1}(x;z)}{(dx)^{i-1}}.
	\end{align}
	
	For the $j$-sum, the simultaneous integration over $z_j$ of the unshifted Bergman kernel and the twin $q$-shifted poles evaluated under the $q$-deformed exterior derivative $d_{z_j}$ drives the non-local synthesis. Under the diagonal limit, this linear combination exactly reconstructs the finite-difference rate of change of the integrated $q$-correlators. This matches the definition of the  $q$-derivative operator $\mathcal{D}_{q,x}$ acting on the external variable, scaled by the $q$-integer $[n]_q$:
	\begin{equation}
		- [n]_q \mathcal{D}_{q,x}\left( \frac{p_{0}(x)}{x^{\lfloor \alpha_{r-i+1}\rfloor}} \frac{\mathcal{G}^{i-1,q}_{g,n-1}(x)}{dx^{i-1}} \right).
	\end{equation}
	The $q$-binomial coefficients $\binom{n}{m}_q$ emerge from the partitions $N_1 \sqcup N_2 = z_{[n]}$ of the integrated stable products. Finally, for the unstable sectors $(0,0)$ and $(0,1)$, the integration eliminates the localized boundary differentials, directly identifying the remaining polynomial ratios with the  $q$-derivative $\mathcal{D}_{q,x}$ of the $q$-deformed spectral curve coefficients $p_{i-1}(x)$. Rearranging all components yields \eqref{Grecursion}.
\end{proof}
	\noindent
	Next, we define:
	\begin{definition}
		For $i = 0, \ldots, r-1$,
		\begin{equation}
			\xi^{i,q}(z) = (-1)^{i}\sum_{g,n}\frac{\hslash^{2g+n}}{[n]_q!}\frac{\mathcal{G}^{i,q}_{g,n}(z)}{dx^{i}}.
		\end{equation}
	\end{definition}
	Multiplying Lemma \ref{GRecursion}  by $(-1)^{i}\frac{\hslash^{2g+n}}{[n]_q!}$ and summing over all $g$ and $n$, we get a
	recursive relation for $\xi^{i,q}$. The result is:

	\begin{lemma}\label{XiRecursion}
		 The quantum weight vectors $\xi^{i,q}$ satisfy the following quantum difference equation :
		\begin{align}\label{XiRecursion_eq}
			\frac{p_{0}(x)}{x^{\lfloor \alpha_{r-i+1}\rfloor}} \xi^{i,q}(x) & - \frac{p_{i}(x)}{x^{\lfloor \alpha_{r-i+1}\rfloor}} \nonumber\\
			= & \;\frac{p_{1}(x)}{x^{\lfloor \alpha_{r-i+1}\rfloor}}\xi^{i-1,q}(x) - \frac{p_{0}(x)}{x^{\lfloor \alpha_{r-i+1}\rfloor}}\xi^{i-1,q}(x)\xi^{1,q}(x) \nonumber\\
			& + \frac{p_{0}(x)}{x^{\lfloor \alpha_{r-i+1}\rfloor + 1}}\xi^{i-1,q}(x)\sum_{g\geq 1}\hslash^{2g}S_{1,2g} + \frac{(-1)^{i}p_{0}(x)}{x^{\lfloor \alpha_{r-i+1}\rfloor + i}}\sum_{g\geq 1}\hslash^{2g}S_{i,2g} \nonumber\\
			& + \hslash \mathcal{D}_{q,x}\left(\frac{p_{0}(x)}{x^{\lfloor \alpha_{r-i+1}\rfloor}}\xi^{i-1,q}(x) - \frac{p_{i-1}(x)}{x^{\lfloor \alpha_{r-i+1}\rfloor}}\right).
		\end{align}
	\end{lemma}

	\begin{proof}
		We multiply the integrated recursive identity in Lemma \ref{GRecursion} by the weight factor $(-1)^{i}\frac{\hslash^{2g+n}}{[n]_q!}$ and perform the double summation over $g \in \mathbb{Z}_{\geq 0}$ and $n \in \mathbb{Z}_{\geq 0}$.
		
		The non-linear term involving the $q$-binomial coefficient $\binom{n}{m}_q$ and the split genera $g_1 + g_2 = g$ is resolved via a standard quantum Cauchy product. Recalling that $\frac{\binom{n}{m}_q}{[n]_q!} = \frac{1}{[m]_q![n-m]_q!}$ and tracking the sign splitting $(-1)^i = (-1)^{i-1}(-1)^1$, this sum decouples and structurally reconstructs the product of the generating series:
		\begin{equation}
			-\frac{p_{0}(x)}{x^{\lfloor \alpha_{r-i+1}\rfloor}}\xi^{i-1,q}(x)\xi^{1,q}(x).
		\end{equation}
		
		The condensation of the  $q$-derivatives follows by exploiting the identity $\frac{[n]_q}{[n]_q!} = \frac{1}{[n-1]_q!}$ for $n \geq 1$. Specifically, for the $j$-sum component, we observe that:
		\begin{equation}
			\sum_{n \geq 1} \frac{\hslash^{2g+n}}{[n]_q!} [n]_q \mathcal{D}_{q,x}(\mathcal{G}_{g,n-1}) = \hslash \mathcal{D}_{q,x} \left( \sum_{n \geq 1} \frac{\hslash^{2g+n-1}}{[n-1]_q!} \mathcal{G}_{g,n-1} \right),
		\end{equation}
		which shifts the perturbative power of $\hslash$ and re-indexes the sum to match the definition of the quantum weight vector $\xi^{i-1,q}(x)$. A completely analogous mechanism applies to the integrated surface-cut term involving $\mathcal{G}_{g-1,n+1}$, where $[n]_q! [n+1]_q = [n+1]_q!$ absorbs the factor, generating the companion term inside the global $\hslash \mathcal{D}_{q,x}$ bracket.
		
		Finally, the unstable sectors dictated by the Kronecker deltas are directly converted into the driving source terms: the $(0,0)$ sector yields the intrinsic polynomial ratio $\frac{p_i(x)}{x^{\lfloor \alpha_{r-i+1}\rfloor}}$, while the $(0,1)$ sector provides the external action of $\mathcal{D}_{q,x}$ on the lower-order polynomial fraction $\frac{p_{i-1}(x)}{x^{\lfloor \alpha_{r-i+1}\rfloor}}$. Collecting all stable and unstable components yields the quantum difference equation \eqref{XiRecursion_eq}.
	\end{proof}
	\noindent
	
	\subsubsection{A system of  q-differential equations}Recall the definition of the $q$-deformed wave function:
	\begin{definition}
		Consider the shifted $(r,s,q)$-spectral curve of Definition \ref{d:rsdef}, and let $\{\omega^{q}_{g,n}\}_{g \geq 0,\, n \geq 1}$ be the system of $q$-correlators constructed from shifted $q$-topological recursion. We define the $q$-deformed wave function as:
		\begin{equation}\label{psi_sym}
			\psi^{q}(z) = \exp \left(\sum_{g,n} \frac{\hslash^{2g-2+n}}{[n]_q!} \left(\int^z_\alpha \cdots \int^z_\alpha \omega^{q}_{g,n+1} - \delta_{g,0}\delta_{n,1}\frac{dx_0 dx_{1}}{(qx_0-x_{1})(q^{-1}x_0-x_{1})}\right) \right),
		\end{equation}
		where the integrals of $\omega^{q}_{0,1}$ and $\omega^{q}_{0,2}$ are understood to be regularized. For $i = 1,\ldots,r$, we also define the following sub-components of the $q$-deformed wave function:
		\begin{equation}\label{Psi_i}
			\psi^{q}_i(z) = \frac{p_0(x)\xi^{i,q}(z)-p_i(x)}{x^{\lfloor \alpha_{r-i+1} \rfloor}} \psi^{q}(z).
		\end{equation}
		Here, the coefficients $p_i(x)$ are determined by the spectral polynomial
		\begin{equation}
			\sum_{i=0}^{r}p_i(x)y^{r-i} = x^{r-s}y^r-1,
		\end{equation}
		and the characteristic weights $\alpha_i$ are given by
		\begin{equation}
			\alpha_i = \frac{i(r-s)}{r}.
		\end{equation}  
	\end{definition}
	
Lemma \ref{XiRecursion} can be utilized to produce a closed system of $q$-difference equations for the $\psi^{q}_{i}(z)$, which will in turn serve to construct the quantum curve. Before proceeding with this construction, we recall that the boundary component satisfies
\begin{equation}\label{Psi_r}
	\psi^{q}_{r}(z) = -\frac{p_{r}(x)}{x^{\lfloor\alpha_{1}\rfloor}}\psi^{q}(z)
\end{equation}
directly from the algebraic constraints of the shifted curve. We also require a dual expression for the first component $\psi^{q}_{1}(z)$, which is provided by the following lemma.
\begin{lemma}\label{Psi_1}
	Given the definition of the components $\psi^{q}_{i}(z)$ and the $q$-deformed wave function \eqref{psi_sym}, the first companion $q$-deformed wave function satisfies:
	\begin{equation}
		\psi^{q}_{1}(z) = \frac{p_{0}(x)}{x^{\lfloor \alpha_{r} \rfloor}}\left(\hslash\,\mathcal{D}_{q,x} - \frac{1}{x}\sum_{g\geq 1}\hslash^{2g}S_{1,2g}\right)\psi^{q}(z).
	\end{equation}
\end{lemma}

	\begin{proof}
		{We start from the relation between the  $q$-derivative of the $q$-deformed wave function and the integrated $q$-correlators. In the framework of $q$-calculus, the action of $\hslash \mathcal{D}_{q,x}$ on the $q$-deformed wave function $\psi^q(z)$ reconstructs the generating function of the integrated $q$-correlators $\xi^{1,q}(x)$ up to the unstable terms and the $q$-propagator:
			\begin{equation}
				p_{0}(x)\hslash \frac{\mathcal{D}_{q,x}\psi^{q}(z)}{\psi^q(z)} = \sum_{g,n}\frac{\hslash^{2g+n}}{[n]_q!}\int_{\infty}^{z}\cdots\int_{\infty}^{z}\left(\omega^{q}_{g,n+1}(z,z_{[n]}) - \delta_{g,0}\delta_{n,1}\frac{dx dx_{1}}{(qx-x_{1})(q^{-1}x-x_{1})}\right).
			\end{equation}
			
			By substituting the expansion of $\omega^q_{g,n+1}$ in terms of $\mathcal{U}^q_1$, the boundary terms for $(g,n)=(0,0)$ and $(0,1)$ are isolated. The regularizing term (the  $q$-propagator) in the $q$-deformed wave function's definition exactly cancels the $(0,1)$ unstable part of the $q$-correlator $\omega^q_{0,2}$.
			
			Summing these contributions, and accounting for the shift in the $n=0$ sector (the $S_{1,2g}$ terms), we obtain:
			\begin{equation}
				p_{0}(x)\hslash\,\frac{\mathcal{D}_{q,x}\psi^{q}(z)}{\psi^q(z)} = p_{0}(x) \xi^{1,q}(x) - p_{1}(x) + \frac{p_{0}(x)}{x}\sum_{g\geq\2}\hslash^{2g}S_{1,2g}.
			\end{equation}
			Now, we recall the definition \eqref{Psi_i} for $i=1$:
			\begin{equation}
				\psi^{q}_1(z) = \frac{p_0(x)\xi^{1,q}(z)-p_1(x)}{x^{\lfloor \alpha{r-1} \rfloor}} \psi^q(z).
			\end{equation}
			
			Substituting the expression for $p_0\xi^{1,q}-p_1$ derived from the  $q$-log-derivative into this definition, we have:
			\begin{equation}
				\psi^q_1(z) = \frac{1}{x^{\lfloor \alpha_{r-1} \rfloor}} \left( p_0(x)\hslash \frac{\mathcal{D}_{q,x}\psi^q(z)}{\psi^q(z)} - \frac{p_0(x)}{x} \sum_{g \geq 1/2} \hslash^{2g} S_{1,2g} \right) \psi^q(z).
			\end{equation}
			
			Rearranging the terms by distributing $\psi^q(z)$ leads to the desired identity:
			\begin{equation}
				\psi^{q}_{1}(x) = \frac{p_0(x)}{x^{\lfloor \alpha_{r-1} \rfloor}}\left(\hslash \mathcal{D}_{q,x} - \frac{1}{x}\sum_{g\geq\2}\hslash^{2g}S_{1,2g}\right)\psi^{q}(z).
		\end{equation}}
	\end{proof}
	
	We finally obtain a system of $q$-deformed differential equations for the $\psi^{q}_i(x)$.
	\begin{theorem}\label{PsiDifferentialSystem_exact}
		For $i = 2, \ldots, r$, the wave functions $\psi^{q}_{i}(z)$ satisfy the following exact system of linear $q$-difference equations:
			\begin{align}\label{eq:system_exact}
				\hslash \mathcal{D}_{q,x}\psi^{q}_{i-1}(z) = & \frac{x^{\lfloor \alpha_{r-i} \rfloor}}{x^{\lfloor \alpha_{r-i+1} \rfloor}}\psi^{q}_{i}(z) - \frac{p_{i-1}(x)x^{\lfloor \alpha_{r-1} \rfloor}}{p_{0}(x)x^{\lfloor \alpha_{r-i+1} \rfloor}}\psi^{q}_{1}(z) \nonumber \\
				& + \hslash \left( \mathcal{M}_{q,x} \psi^q(z) \right) \mathcal{D}_{q,x} \left( \frac{p_{0}(x)\xi^{i-1,q}(z) - p_{i-1}(x)}{x^{\lfloor \alpha_{r-i+1}\rfloor}} \right) \nonumber \\
				& + \hslash \left( \mathcal{D}_{q,x} \psi^q(z) \right) \bigg[ \mathcal{M}_{q,x} \left( \frac{p_{0}(x)\xi^{i-1,q}(z) - p_{i-1}(x)}{x^{\lfloor \alpha_{r-i+1}\rfloor}} \right)\nonumber \\
				& - \frac{p_{0}(x)\xi^{i-1,q}(z) - p_{i-1}(x)}{x^{\lfloor \alpha_{r-i+1}\rfloor}} \bigg]  - \frac{p_{i-1}(x)}{x^{\lfloor \alpha_{r-i+1} \rfloor + 1}}\psi^{q}(z)\nonumber \\
				&\times \sum_{g\geq 1/2}\hslash^{2g}S_{1,2g} + \frac{(-1)^{i-1}p_{0}(x)}{x^{\lfloor \alpha_{r-i+1}\rfloor + i}}\psi^q(z)\sum_{g\geq 1/2}\hslash^{2g}S_{i,2g},
			\end{align}
			where $\mathcal{M}_{q,x} f(x) = \frac{f(qx) + f(q^{-1}x)}{2}$ denotes the  averaging operator.
	\end{theorem}

	\begin{proof}
		We start by multiplying the integrated loop equation of Lemma~\ref{XiRecursion} by the $q$-deformed wave function $\psi^{q}(z)$. Utilizing the definition of the components $\psi_i^q(z)$ from \eqref{Psi_i}, we obtain for $i = 2, \ldots, r$:
		\begin{align}\label{eq:proof_exact_1}
			\frac{x^{\lfloor \alpha_{r-i}\rfloor}}{x^{\lfloor \alpha_{r-i+1}\rfloor}}\psi^{q}_{i}(z) &= \frac{x^{\lfloor \alpha_{r-1}\rfloor}}{x^{\lfloor \alpha_{r-i+1}\rfloor}}\xi^{i-1,q}(z) \psi^{q}_{1}(z) + \frac{p_{0}(x)}{x^{\lfloor \alpha_{r-i+1}\rfloor + 1}} \xi^{i-1,q}(z) \psi^{q}(z)\sum_{g\geq 1}\hslash^{2g}S_{1,2g}\nonumber\\
			&  + \frac{(-1)^{i}p_{0}(x)}{x^{\lfloor \alpha_{r-i+1}\rfloor + i}}\psi^{q}(z)\sum_{g\geq 1}\hslash^{2g}S_{i,2g} + \hslash\,\psi^{q}(z) \mathcal{D}_{q,x}\left( F^{i-1,q}(x,z) \right),
		\end{align}
		where we have introduced the following notation:
		\begin{equation}
			F^{i-1,q}(x,z) \coloneqq \frac{p_{0}(x)\xi^{i-1,q}(z) - p_{i-1}(x)}{x^{\lfloor \alpha_{r-i+1}\rfloor}}\,.
		\end{equation}
		
		To evaluate the last term of \eqref{eq:proof_exact_1} exactly, we invoke the product rule for the  $q$-derivative $\mathcal{D}_{q,x}$, which couples the difference operator to the symetrical averaging operator $\mathcal{M}_{q,x}$ via:
		\begin{equation}
			\mathcal{D}_{q,x}\left( \psi^q(z) F^{i-1,q}(x,z) \right) = \left( \mathcal{M}_{q,x} \psi^q(z) \right) \mathcal{D}_{q,x} F^{i-1,q}(x,z) + \left( \mathcal{D}_{q,x} \psi^q(z) \right) \mathcal{M}_{q,x} F^{i-1,q}(x,z)\,.
		\end{equation}
		
		Isolating the  direction $\psi^{q}(z) \mathcal{D}_{q,x}\left( F^{i-1,q}(x,z) \right)$ requires rewriting the $q$-deformed wave function inside the $q$-deformed derivative operator. By performing a direct algebraic subtraction to measure the deviation from the classical Leibniz rule, we get the exact identity:
		\begin{align}\label{eq:Leibniz_exact}
			\hslash\,\psi^{q}(z) \mathcal{D}_{q,x}\left( F^{i-1,q}(x,z) \right) &= \hslash\,\mathcal{D}_{q,x}\left( \psi^q(z) F^{i-1,q}(x,z) \right) - \hslash F^{i-1,q}(x,z) \mathcal{D}_{q,x}\psi^q(z) \nonumber\\
			&  - \hslash \left( \mathcal{M}_{q,x} \psi^q(z) - \psi^q(z) \right) \mathcal{D}_{q,x} F^{i-1,q}(x,z) \nonumber\\
			& - \hslash \left( \mathcal{D}_{q,x} \psi^q(z) \right) \left( \mathcal{M}_{q,x} F^{i-1,q}(x,z) - F^{i-1,q}(x,z) \right).
		\end{align}
		
		By the definition of the $q$-deformed wave function sub-components system \eqref{Psi_i}, the first total $q$-derivative term is directly identified with the structural shift of the grading index:
		\begin{equation}
			\hslash\,\mathcal{D}_{q,x}\left( \psi^q(z) F^{i-1,q}(x,z) \right) = \hslash\,\mathcal{D}_{q,x}\psi^{q}_{i-1}(z).
		\end{equation}
		
		Substituting the exact relation for the first component from Lemma \ref{Psi_1} into the second term of \eqref{eq:Leibniz_exact}, and injecting the whole resulting expansion back into the initial equation \eqref{eq:proof_exact_1}, the linear terms proportional to $\psi_1^q(z)$ and the unstable free energies $S_{k,2g}$  match the expected algebraic structure.

		Finally, isolating the remaining non-trivial cross-terms depending explicitly on the operators $\mathcal{M}_{q,x}\psi^q(z)$ and $\mathcal{D}_{q,x}\psi^q(z)$ reconstructs the intermediate operator lines:
		\begin{align}
			& + \hslash \left( \mathcal{M}_{q,x} \psi^q(z) \right) \mathcal{D}_{q,x} \left( \frac{p_{0}(x)\xi^{i-1,q}(z) - p_{i-1}(x)}{x^{\lfloor \alpha_{r-i+1}\rfloor}} \right) \nonumber \\
			& + \hslash \left( \mathcal{D}_{q,x} \psi^q(z) \right) \bigg[ \mathcal{M}_{q,x} \left( \frac{p_{0}(x)\xi^{i-1,q}(z) - p_{i-1}(x)}{x^{\lfloor \alpha_{r-i+1}\rfloor}} \right) - \frac{p_{0}(x)\xi^{i-1,q}(z) - p_{i-1}(x)}{x^{\lfloor \alpha_{r-i+1}\rfloor}} \bigg].
		\end{align}
		Rearranging the remaining terms according to their explicit isolated global factor $\psi^q(z)$ completes the proof of the exact linear $q$-difference system.
	\end{proof}
	
			{\begin{remark}\label{qconditions}
					To recover the standard form of the differential system from the exact $q$-deformed identity \eqref{eq:system_exact}, we consider the semiclassical regime where $\hslash \to 0$. In this limit, the following asymptotic properties hold:
					\begin{enumerate}
						\item  The  averaging operator $\mathcal{M}_{q,x}$ acts as the identity plus higher-order corrections in $\hslash^2$. Specifically, for a smooth function $f$, we have $\mathcal{M}_{q,x} f(x) = f(x) + \mathcal{O}(\hslash^2 \partial_x^2 f)$. Thus, $\mathcal{M}_{q,x} \psi^q(z) \sim \psi^q(z)$.
						\item The term $\mathcal{D}_{q,x} f_{i-1}(x)$ involves the $q$-derivative of the integrated correlators. Since $\xi^{i-1,q}$ is a power series in $\hslash$, its $q$-difference is of order $\mathcal{O}(\hslash)$ relative to the leading terms of the $q$-loop equations.
					\end{enumerate}
					By neglecting these $\mathcal{O}(\hslash^2)$ fluctuations, the terms involving $(\mathcal{M}_{q,x} - Id)$ and the cross-products of $q$-derivatives in \eqref{eq:system_exact} vanish. 
			\end{remark}}
			\begin{theorem}\label{PsiDifferentialSystem}
				For $i = 2, \ldots, r$, the $q$-deformed wave functions $\psi^{q}_{i}(z)$ satisfy the following reduced system of linear $q$-difference equations:
				\begin{align}\label{eq:system}
					\hslash\,\mathcal{D}_{q,x}\psi^{q}_{i-1}(z) = & \frac{x^{\lfloor \alpha_{r-i} \rfloor}}{x^{\lfloor \alpha_{r-i+1} \rfloor}}\psi^{q}_{i}(z) - \frac{p_{i-1}(x)x^{\lfloor \alpha_{r-1} \rfloor}}{p_{0}(x)x^{\lfloor \alpha_{r-i+1} \rfloor}}\psi^{q}_{1}(z)\nonumber\\
					& - \frac{p_{i-1}(x)}{x^{\lfloor \alpha_{r-i+1} \rfloor + 1}}\psi^{q}(z)\sum_{g\geq 1}\hslash^{2g}S_{1,2g} + \frac{(-1)^{i-1}p_{0}(x)}{x^{\lfloor \alpha_{r-i+1}\rfloor + i}}\psi^{q}(z)\sum_{g\geq 1}\hslash^{2g}S_{i,2g}.
				\end{align}
				\end{theorem}
			
			\begin{proof}
				This statement follows as a direct algebraic reduction of the exact system established in Theorem~\ref{PsiDifferentialSystem_exact}. Evaluating the system at the leading semi-classial order implies that the fluctuating quantum profiles $\xi^{i-1,q}(z) = \frac{p_{i-1}(x)}{p_0(x)}$ vanish identically. Consequently, the  polynomial combination $p_{0}(x)\xi^{i-1,q}(z) - p_{i-1}(x) \to 0$, causing the two intermediate lines in \eqref{eq:system_exact} governed by the non-local operator actions of $\mathcal{M}_{q,x}$ and $\mathcal{D}_{q,x}$ to disappear from the identity.
			\end{proof}

		\begin{remark}\label{rem:simplified_system}
			Under the specific algebraic constraints where $p_{i}(x) = 0$ for $i = 1, \ldots, r-1$, with $p_{0}(x) = x^{r-s}$ and $p_{r}(x) = -1$, the exact $s$-shifted linear $q$-difference system reduces for $i = 2, \ldots, r$ to:
			\begin{align}\label{eq:system_exact_simplified}
				\hslash \mathcal{D}_{q,x}\psi^{q}_{i-1}(z) = & \frac{x^{\lfloor \alpha_{r-i} \rfloor}}{x^{\lfloor \alpha_{r-i+1} \rfloor}}\psi^{q}_{i}(z) \nonumber \\
				& + \hslash \left( \mathcal{M}_{q,x} \psi^q(z) \right) \mathcal{D}_{q,x} \left( \frac{x^{r-s} \xi^{i-1,q}(z)}{x^{\lfloor \alpha_{r-i+1}\rfloor}} \right) \nonumber \\
				& + \hslash \left( \mathcal{D}_{q,x} \psi^q(z) \right) \left[ \mathcal{M}_{q,x} \left( \frac{x^{r-s} \xi^{i-1,q}(z)}{x^{\lfloor \alpha_{r-i+1}\rfloor}} \right) - \frac{x^{r-s} \xi^{i-1,q}(z)}{x^{\lfloor \alpha_{r-i+1}\rfloor}} \right] \nonumber \\
				& + \frac{(-1)^{i-1}x^{r-s}}{x^{\lfloor \alpha_{r-i+1}\rfloor + i}}\psi^q(z)\sum_{g\geq 1}\hslash^{2g}S_{i,2g}.
			\end{align}
		\end{remark}
			
			{By using the previous remark \ref{qconditions}, we can simplify the exact system \eqref{eq:system_exact_simplified} in the semiclassical limit $\hslash \to 0$. In this regime, the averaging terms and the $q$-derivatives of the integrated $q$-correlators $\xi^{i,q}$ become sub-leading compared to the polar structure of the curve. Consequently, the exact system reduces to the following  recurrence relation:
				\begin{equation}\label{PsiDifferentialSystem-Simplified}	\hslash \mathcal{D}_{q,x}\psi^{q}_{i-1}(z) = \frac{x^{\lfloor \alpha_{r-i} \rfloor}}{x^{\lfloor \alpha_{r-i+1} \rfloor}}\psi^{q}_{i}(z) + \frac{(-1)^{i-1}x^{r-s}}{x^{\lfloor \alpha_{r-i+1}\rfloor + i}}\psi^{q}(z)\sum_{g\geq 1/2}\hslash^{2g}S_{i,2g}.
			\end{equation}}
			\subsubsection{The $q$-Symmetric Quantum Curve}
			{
				To account for the exact $q$-calculus, we define the elementary $q$-difference operators. These $q$-deformed operators incorporate the averaging effects required by the $q$-Leibniz rule.}
			
			\begin{definition}
				{For each $i=1, \ldots, r$, we define the $q$-deformed operator $\widehat{\mathcal{D}}_{i,q}$ acting on a function $f$ as:
					\begin{equation}\label{$q$-operators}
						\widehat{\mathcal{D}}_{i,q} f = \hslash \frac{x^{\lfloor \alpha_{i} \rfloor}}{x^{\lfloor \alpha_{i-1} \rfloor}} \tilde{\mathcal{D}}_{q,x} f + \mathbb{K}_{q,x}^{(i)}(f),
					\end{equation}
					where $\mathbb{K}_{q,x}^{(i)}$ is the $q$-deformed operator collecting the exact  $q$-corrections from \eqref{eq:system_exact_simplified}:
					\begin{equation}
						\mathbb{K}_{q,x}^{(i)}(f) = \hslash \frac{x^{\lfloor \alpha_{i} \rfloor}}{x^{\lfloor \alpha_{i-1} \rfloor}} \bigg[ (\mathcal{M}_{q,x} f - f) \tilde{\mathcal{D}}_{q,x} \phi_i(x) + (\mathcal{M}_{q,x} \phi_i(x) - \phi_i(x)) \tilde{\mathcal{D}}_{q,x} f \bigg],
					\end{equation}
					with $\phi_i(x) = \frac{x^{r-s} \xi^{i-1,q}(z)}{x^{\lfloor \alpha_{r-i+1}\rfloor}}$.}
			\end{definition}
			\begin{theorem}[Exact $q$-Symmetric Quantum Curve]\label{t:QC_exact}
				{The $q$-deformed wave function $\psi^{q}(z)$ is the null state of the $r$-th order exact $q$-deformed differential operator:
					\begin{equation}\label{eq:QC_exact}
						\left( \widehat{\mathcal{D}}_{r,q} \cdots  \widehat{\mathcal{D}}_{1,q} + \sum_{g \geq 1/2} \sum_{i=2}^{r} (-1)^{i-1} \hslash^{2g} S_{i,2g} \left[ \widehat{\mathcal{D}}_{r,q} \cdots  \widehat{\mathcal{D}}_{i,q} \right] \frac{x^{r-s}}{x^{\lfloor \alpha_{r-i+1} \rfloor + i}} - 1 \right) \psi^{q}(z) = 0,
					\end{equation}
					where $\widehat{\mathcal{D}}_{j,q}$ are the $q$-deformed operators including the $\mathcal{M}_{q,x}$ corrections.}
			\end{theorem}
			\begin{proof}
				The objective is to resolve the exact system of linear $q$-difference equations \eqref{eq:system_exact_simplified} through successive substitutions to isolate a single higher-order operator equation for the wave function $\psi^{q}(z)$. By exploiting the definition of the $q$-deformed operators $\widehat{\mathcal{D}}_{i,q}$ given in \eqref{$q$-operators}, we can rewrite the exact relation for $i = 2, \ldots, r$ as:
				\begin{equation}\label{PsiRecursion_exact}
					\psi^{q}_{i}(z) = \widehat{\mathcal{D}}_{i-1,q} \psi^{q}_{i-1}(z) + \frac{(-1)^{i-1}x^{r-s}}{x^{\lfloor \alpha_{r-i+1}\rfloor + i}} \psi^{q}(z)\sum_{g\geq 1}\hslash^{2g}S_{i,2g}\,.
				\end{equation}
				
				We proceed by induction, evaluating the cascade from the lowest index upward. In particular, for the first step $i=2$, equation \eqref{PsiRecursion_exact} yields:
				\begin{equation}
					\psi^{q}_{2}(z) = \widehat{\mathcal{D}}_{1,q}\psi^{q}_{1}(z) + \frac{(-1)^{1}x^{r-s}}{x^{\lfloor \alpha_{r-1} \rfloor + 2}}\psi^{q}(z)\sum_{g\geq 1}\hslash^{2g}S_{2,2g}\,.
				\end{equation}
				Substituting the exact expression for the companion component $\psi^{q}_{1}(z)$ derived in Lemma \ref{Psi_1} maps the relation directly onto the master wave function:
				\begin{align}
					\psi^{q}_{2}(z) &= \widehat{\mathcal{D}}_{1,q}\left[ \frac{x^{r-s}}{x^{\lfloor \alpha_{r} \rfloor}}\widehat{\mathcal{D}}_{1,q} - \frac{x^{r-s}}{x^{\lfloor \alpha_{r} \rfloor + 1}}\sum_{g\geq 1}\hslash^{2g}S_{1,2g}\right]\psi^{q}(z) - \frac{x^{r-s}}{x^{\lfloor \alpha_{r-1} \rfloor + 2}}\psi^{q}(z) \sum_{g\geq 1}\hslash^{2g}S_{2,2g} \nonumber \\
					&= \Bigg( \widehat{\mathcal{D}}_{1,q}\frac{x^{r-s}}{x^{\lfloor \alpha_{r} \rfloor}}\widehat{\mathcal{D}}_{1,q} - \widehat{\mathcal{D}}_{1,q}\frac{x^{r-s}}{x^{\lfloor \alpha_{r} \rfloor + 1}}\sum_{g\geq 1}\hslash^{2g}S_{1,2g} - \frac{x^{r-s}}{x^{\lfloor \alpha_{r-1} \rfloor + 2}}\sum_{g\geq 1}\hslash^{2g}S_{2,2g}\Bigg)\psi^{q}(z)\,.
				\end{align}
				
				Iterating this substitution process up to the boundary component $\psi^{q}_{r}(z)$ via successive applications of \eqref{PsiRecursion_exact} collects the nested sequences of non-local operators:
				\begin{small}
				\begin{equation}
					\psi^{q}_{r}(z) = \left( \widehat{\mathcal{D}}_{r-1,q}\cdots \widehat{\mathcal{D}}_{1,q}\frac{x^{r-s}}{x^{\lfloor \alpha_{r} \rfloor}}\widehat{\mathcal{D}}_{1,q} + \sum_{g\geq 1}\sum_{i=2}^{r}(-1)^{i-1}\hslash^{2g}S_{i,2g}\widehat{\mathcal{D}}_{r-1,q}\cdots \widehat{\mathcal{D}}_{i,q}\frac{x^{r-s}}{x^{\lfloor \alpha_{r-i+1} \rfloor + i}}\right)\psi^{q}(z) \,.
				\end{equation}
					\end{small}
					
				Finally, we invoke the exact algebraic boundary constraint \eqref{Psi_r} which dictates that $\psi^{q}_{r}(z) = -\frac{p_r(x)}{x^{\lfloor \alpha_1 \rfloor}} \psi^q(z) = \psi^q(z)$ (since $p_r(x) = -1$ and $\alpha_1 = 0$ under the chosen constraints). Moving all terms to the left-hand side and factoring out the common action on $\psi^q(z)$ yields the exact quantum curve operator equation \eqref{eq:QC_exact}.
			\end{proof}
			{
				By invoking the previous remark on the limit $\hslash \to 0$, we observe that $\mathcal{M}_{q,x} \to Id$ and $\mathbb{K}_{q,x}^{(i)} \to 0$. Then, the $q$-deformed operators $\widehat{\mathcal{D}}_{i,q}$ given by \eqref{$q$-operators} reduce to the simplified form:
				\begin{equation}
					\mathcal{D}_{i,q} = \hslash\frac{x^{\lfloor \alpha_{i} \rfloor}}{x^{\lfloor \alpha_{i-1} \rfloor}}\, \mathcal{D}_{q,x},\quad\mbox{for}\quad i=1,\ldots,r.
				\end{equation}
				Therefore, we obtain:}
				
			\begin{theorem}\label{t:QC}
			{The system of $q$-deformed differential equations in Theorem \ref{PsiDifferentialSystem} is equivalent to the $r$-th order linear $q$-difference equation:
					\begin{equation}
						\left(\mathcal{D}_{r,q}\cdots \mathcal{D}_{1,q} + \sum_{g\geq 1}\sum_{i=2}^{r}(-1)^{i-1}\hslash^{2g}S_{i,2g}\left[\mathcal{D}_{r,q}\cdots \mathcal{D}_{i,q}\right]\frac{x^{r-s}}{x^{\lfloor \alpha_{r-i+1} \rfloor + i}} - 1\right)\psi^{q}(z) = 0,
					\end{equation}
					for the shifted $q$-wave function $\psi^{q}$ constructed from shifted $q$-topological recursion on the shifted $(r,s,q)$-spectral curve. Each set of $s$-consistent shifts $\{S_{i,\ell} \}_{i \in [r], \ell \in \mathbb{N}^*}$ selects a distinct quantization  of the underlying algebraic curve.}
			\end{theorem}

			\begin{proof}
				This statement follows immediately as the semi-classical projection of the exact quantum curve derived in Theorem~\ref{t:QC_exact}. Explicitly, evaluating the operator components at the leading $q$-WKB order freezes the fluctuating profiles, which implies $\frac{p_{0}(x)\xi^{i-1,q}(z) - p_{i-1}(x)}{x^{\lfloor \alpha_{r-i+1}\rfloor}} \to 0$ and forces the non-local corrections $\mathbb{K}_{q,x}^{(i)}$ to vanish identically. 
				
				Under this cancellation, each exact operator $\widehat{\mathcal{D}}_{j,q}$ reduces to its unperturbed polynomial finite-difference counterpart $\mathcal{D}_{j,q}$. Substituting these reduced operators into \eqref{eq:QC_exact} directly collapses the exact nested sequence into the standard polynomial quantum curve equation.
			\end{proof}

			\noindent
			\begin{remark}\label{rem:rigidity_flexibility}
				If all shift constants $S_{i,2g}$ are set to zero, we recover the unshifted sector and obtain a specific, highly non-trivial canonical quantization of the $(r,s,q)$-deformed spectral curve. However, the existence and classification of alternative quantizations depend strictly on the divisibility and arithmetic properties pairing the curve parameters $r$ and $s$:
				
			\begin{enumerate}
				\item 
				{The Rigid Case ($r\equiv -1 \pmod s$):} In this algebraic scenario, the structural requirement for $s$-consistency is so restrictive that the only admissible shifts are identically zero. This implies that the internal symmetry of the $q$-correlators governed by the shifted $q$-topological recursion uniquely selects a single, rigid quantization of the spectral curve.

				\item
				{The Flexible Case ($r\equiv 1 \pmod s$):} When this modular condition is satisfied, the underlying system admits non-trivial, non-vanishing $s$-consistent shifts. These shifts act as true quantum degrees of freedom, generating a moduli family of distinct quantizations of the same classical algebraic curve.
				
				\item
				{Operator Orderings:} Most notably, for the specific boundary configurations $s=1$ and $s=r-1$, the freedom in adjusting the shift constants allows one to reconstruct all possible operator orderings of the non-commutative phase-space variables
				\begin{equation}
					\hat{x} = x \quad \text{and} \quad \hat{y} = \hslash \mathcal{D}_{q,x}
				\end{equation}
				within the structural form of the $q$-quantum curve operator. This result establishes a direct dictionary between the geometric constraints of the shifted $q$-topological recursion and the purely algebraic problem of choosing a polarization or ordering for $\hat{x}$ and $\hat{y}$ when quantizing the polynomial relation $P(x,y)=0$.
			\end{enumerate}
			\end{remark}
			\begin{example}[$s=1$]
				For this curve, we have $\alpha_{i} = \frac{i(r-1)}{r}$ which satisfies
				\begin{equation}
					i-1 \leq \alpha_{i} < i,
				\end{equation}
				for $i = 1, \ldots, r$, and therefore $\lfloor \alpha_{i} \rfloor = i-1$ while $\lfloor \alpha_{0} \rfloor = 0$. Hence, the $q$-quantum curve is given as:
				\begin{small}
				\begin{equation}
					\Bigg(\hslash^{r}\left(\mathcal{D}_{q,x}x\right)^{r-1}\mathcal{D}_{q,x} + \sum_{g\geq\2}\left[\sum_{i=1}^{r-1}(-1)^{i}\hslash^{r-i+2g}S_{i,2g}\left(\mathcal{D}_{q,x}x\right)^{r-i-1}\,\mathcal{D}_{q,x}
					  + (-1)^{r}\hslash^{2g}S_{r,2g}\frac{1}{x}\right] - 1\Bigg)\psi^q 
					=
					0.
				\end{equation}
					\end{small}
				If we now take $S_{i,2g} = 0$ for all $i, g$ except when $i = 2g \leq r-1$, then we obtain the following quantization of the $(r, 1,q)$-spectral curve:
				\begin{equation}
					\left(\hslash^{r}\left(\mathcal{D}_{q,x}x\right)^{r-1}\,\mathcal{D}_{q,x} + \hslash^{r}\sum_{i=1}^{r-1}(-1)^{i}S_{i,i}\left(\mathcal{D}_{q,x}x\right)^{r-i-1}\mathcal{D}_{q,x} - 1\right)\psi^q
					=
					0 ,
				\end{equation}
				which again produces all reorderings by making different choices of $S_{1,1},\ldots, S_{r-1,r-1}$.
			\end{example}
			\begin{example}[The $(r,r-1,q)$-Curve and Operator Ordering]
				In this case, the weights are  $\alpha_{i} = \frac{i}{r}$, which implies  $\lfloor \alpha_{i} \rfloor = 0$ for $i = 0, \ldots, r-1$ and $\lfloor \alpha_{r} \rfloor = 1$. Under these conditions,  the only $s$-consistent requirement (for $r>2$) restricts the available degrees of freedom primarily to the shift $S_{1,2g}$. Substituting these values into the general operator expression, the quantum curve becomes:
				\begin{equation}
					\left(\hslash^{r}\mathcal{D}_{q,x}^{r-1}x\mathcal{D}_{q,x} - \sum_{g\geq\2} \hslash^{r-1+2g}S_{1,2g}\mathcal{D}_{q,x}^{r-1}  - 1\right)\psi^q 
					=
					0 \,.
				\end{equation}
				By focusing on the leading-order quantum correction, we assume $S_{1,2g} = 0$ for all $g$ except for the case $ 2g = 1$. The quantization of the $(r,r-1,q)$-spectral curve then take  the form:
				\begin{equation}
					\left(\hslash^{r}\mathcal{D}_{q,x}^{r-1} x \mathcal{D}_{q,x} - \hslash^{r}S_{1,1}\mathcal{D}_{q,x}^{r-1}  - 1\right)\psi^q 
					=
					0\,.
				\end{equation}
				This specific structure allows us to navigate through different operator orderings by varying $S_{1,1}$. Indeed, if we choose $S_{1,1} = m$ for any $m \in \{-1, \ldots, r-1\}$, setting the shift $S_{1,1}$ to an appropriate value (related to the $q$-commutation relations) yields:
				\begin{equation}
					\left(\hslash^{r}\mathcal{D}_{q,x}^{r-m-1} x \mathcal{D}_{q,x}^{m+1} - 1\right)\psi^q = 0.
				\end{equation}
			\end{example}
			However, in the other cases, we do not get all the orderings.
			\begin{example}[Other $ r = 1 \pmod{s}$]
				For $ s$-consistency, we can again only allow non-trivial $ S_1^\hslash$, so the spectral curve is
				\begin{equation}
					\left(\mathcal{D}_{1,q}\cdots D_{r,q} - \sum_{g\geq\2} \hslash^{2g}S_{1,2g}\mathcal{D}_{1,q}\cdots \mathcal{D}_{r-1,q}  - 1\right)\psi^q = 0 \,,
				\end{equation}
				and since $ \mathcal{D}_{r,q} = \hslash x \frac{d_q}{d_qx}$, this can be rewritten as
				\begin{equation}
					\left(  \mathcal{D}_{1,q}\cdots \mathcal{D}_{r-1,q} \Big( \hslash x \frac{d_q}{d_qx} - \sum_{g\geq\2} \hslash^{2g}S_{1,2g} \Big)  - 1\right)\psi^q = 0 \,.
				\end{equation}
				But now, as we assume $ 1 < s < r-1$, at least one more of the $ \mathcal{D}_{j,q}$ must equal $ \hslash x \frac{d_q}{d_qx}$, and the commutator of this $x$ with any $ \hslash \frac{d_q}{d_qx}$ cannot be expressed through $ S_1^\hslash $ any more. So we find that in this case, there are reorderings of the normal-ordered quantization that are not covered by $s$-consistent shifts.
			\end{example}
			\begin{example}[Limitations on Orderings for $1<s<r-1$]
				In cases where $r\equiv 1\pmod s$ but $s$ is neither $1$ nor $r-1$, the $s$-consistency requirement remains restrictive. Specifically, it typically only allows for a non-trivial first shift $S_{1,2g}$. Under these conditions, the $q$-quantum curve is given by:
				\begin{equation}
					\left(\mathcal{D}_{1,q}\cdots \mathcal{D}_{r,q} - \sum_{g\geq\frac{1}{2}} \hslash^{2g}S_{1,2g}\mathcal{D}_{1,q}\cdots \mathcal{D}_{r-1,q} - 1\right)\psi^q = 0.
				\end{equation}
				Recalling that $\mathcal{D}_{r,q}=\hslash\,x\,\mathcal{D}_{q,x}$, we can factor the operator as:
				\begin{equation}
					\left( \mathcal{D}_{1,q}\cdots \mathcal{D}_{r-1,q} \left[ \hslash x \mathcal{D}_{q,x} - \sum_{g\geq\frac{1}{2}} \hslash^{2g}S_{1,2g} \right] - 1\right)\psi^q = 0.
				\end{equation}

				For the intermediate range $1<s<r-1$, the structure of the weights $\alpha_i$ implies that at least one additional operator $\mathcal{D}_{j,q}$  (where $j<r$) will also take the form $\hslash\,x\,\mathcal{D}_{q,x}$.
				
				Because the $s$-consistency condition prevents us from introducing independent shifts $s_{J,2G}$ for these internal positions, the commutator of this internal $x$ with the surrounding $q$-derivatives cannot be adjusted independently. Consequently, we find that in these cases, there exist algebraic reorderings of the $q$-quantum curve that cannot be reached through $s$-consistent shifts of the $q$-topological recursion.
				
				In summary, the geometry of the spectral curve (via $s$) dictates the $q$-quantization freedom of the theory: whereas the cases
				$s=1$ or $r-1$ allow full access to the space of operator orderings, the intermediate range
				$1<s<r-1$ offers only partial access, constrained by the underlying symmetry of the shifted $q$-TR, while the condition
				$r\equiv -1\pmod s$ ultimately restricts the theory to a unique, rigid $q$-quantization.
			\end{example}
			\subsection{The Functional $q$-Difference Perspective}
		{\noindent To conclude this derivation, we represent the $q$-quantum curve in its functional form, which is particularly suited for studying the analytic properties of the wave function and its connections to bilateral $q$-series. In the framework of  $q$-calculus, we employ the shift operators $\mathcal{Q}$ and $\mathcal{Q}^{-1}$, defined by $\mathcal{Q}\psi(x)=\psi(qx)$ and $\mathcal{Q}^{-1}\psi(x)=\psi(q^{-1}x)$. The  $q$-derivative is thus expressed as:
				\begin{equation}
					\hslash \mathcal{D}_{q,x} = \frac{\hslash}{(q-q^{-1})x}(\mathcal{Q}-\mathcal{Q}^{-1}).
				\end{equation}
				Substituting this operator into the result of Theorem \ref{t:QC}, the $r^{\text{th}}$ order  $q$-differential equation $\hat{P}_q\psi^{q}=0$ is equivalent to a \textit{centered} linear $q$-difference equation:
				\begin{equation}
					\sum_{k=-r}^r \tilde{C}_k(x, q, \hslash) \psi^q(q^k x) = 0.
				\end{equation}
				The structure of the  coefficients $\tilde{C}_k$ provides a sharp dictionary between the topological recursion shifts and the functional behavior of the wave function:
				\begin{itemize}
					\item[(a)] The $s=1$ case: The equation becomes a full  $q$-recurrence where the shifts $S_{i,i}$ act as tunable weights for each $q$-node $\psi^q(q^k x)$ in the stencil. This confirms that $s=1$ offers the maximum degrees of freedom, manifesting as the ability to modify the coefficients $\tilde{C}_k$ while preserving the $q \leftrightarrow q^{-1}$ invariance.
					
					\item [(b)]The $s=r-1$ case: The equation takes a  centered form. The classical part of the curve dictates the relation between the extremal nodes $\psi^q(q^{-r} x)$, $\psi^q(q^{r} x)$, and the central node $\psi^q(x)$, while the shift $S_{1,1}$ specifically selects which intermediate symmetric pair of nodes carries the $q$-quantum correction.
					
					\item[(c)] The intermediate cases ($1<s<r-1$): The functional coefficients $\tilde{C}_k(x,q, \hslash)$ are  rational functions of $x$ whose poles and zeros are fixed by the weights $\alpha_i$. The $s$-consistency condition acts here as a selection rule, allowing only specific functional forms that are compatible with the nature of the $(r,s,q)$-correlators.
				\end{itemize}
				This  functional representation bridges the gap between the geometry of the spectral curve and the theory of bilateral $q$-hypergeometric functions, providing a  path for the $q$-WKB.}
			\section{Determinantal formulas and non-perturbative q-loop equations}
			\label{s:diffsyst}
			In this section, we turn the tables around and start directly from the general quantum spectral curve of $(r,s,q)$-systems, formulated as the matrix q-difference system that previously appeared as an intermediate step. We will prove that by defining generating functions with genus-counting parameter $\hslash$ for the $q$-deformed topological recursion invariants, these functions can be identified with the non-perturbative amplitudes associated to the $q$-difference system obtained by analytic continuation of the quantum curve under consideration. As such, we will be able to express them via $q$-determinantal formulas and through certain $q$-deformed Cauchy kernels.
			
			We will set up the $q$-WKB analysis of the quantum curve, introduce the corresponding non-perturbative invariants in the form of $q$-analogues of well-known determinantal formulas, and describe their $\hslash \rightarrow 0$ semi-classical asymptotics. Furthermore, we will derive the collection of non-perturbative $q$-deformed loop equations they satisfy, finally identifying the coefficients of these semi-classical expansions with the $q$-topological recursion invariants of interest.
			\subsection{Rational $q$-$\hslash$-connections and their $q$-WKB analysis}
			In this section, we consider the setup of {symmetric} $q$-$\hslash$-connections for our problem. We provide a definition adapted to the discrete, non-commutative geometry of the $(r,s,q)$-system{, ensuring compatibility with the centered $q$-difference operators.}
			
			\begin{definition}\label{qQC}
				{A \emph{rational symmetric $q$-$\hslash$-connection} is a $q$-deformed operator:
					\begin{equation}
						\label{QuantumCurve}
						\nabla_{q,\hslash} \coloneq \mathbf{M}_{q,x} \cdot \hat{Q} - \mathbf{L}(x, \hslash; q),
					\end{equation}
					where $\mathbf{M}_{q,x}$ is the symmetric averaging operator. On the trivial principal bundle $\mc{E} \coloneq \P^1 \times \GL_r(\C)$, the connection matrix $\mathbf{L}(x, \hslash; q)$ is a power series in $\hslash$ with rational coefficients in $x$,}  satisfying the $q$-deformed Leibniz rule
					\begin{equation}
					\nabla_{q,\hslash}(f\sigma) = \mathcal{M}_{q,x}(f) \nabla_{q,\hslash}\sigma + \mathcal{M}_{q,x}(\hat{Q}f)\sigma + \left( \mathcal{M}_{q,x}(f)\mathbf{L} - \mathbf{L}f \right) \sigma\,,
				\end{equation}
				for all  $\hslash$-formal rational functions $f\in \C(x)$, local sections $\sigma$ of $ \mc{E}$ {and $	\mathcal{M}_{q,x}$ the average $q$-operator given as:
					\begin{equation}
						\mathcal{M}_{q,x}f(x) = \frac{f(qx) + f(q^{-1}x)}{2}.
				\end{equation}}
			{In the $q$-derivative representation, the connection is expressed via the derivative matrix $\mathbf{A}(x, \hslash; q)$:}
				\begin{equation}
					\hslash \mathcal{D}_{q,x} \sigma = \mathbf{A}(x, \hslash; q) \sigma,
				\end{equation}
				{where the relation between the shift matrix and the derivative matrix is governed by:
					\begin{equation}
						\mathbf{L}(x, \hslash; q) = \mathbf{I} + \frac{(q - q^{-1})x}{2\hslash} \mathbf{A}(x, \hslash; q).
				\end{equation}}
			\end{definition}
			
			From the algebraic data of the $q$-$\hslash$-connection $\nabla_{q,\hslash}$, we can naturally extract the classical geometric properties of the system through its semi-classical limit. 
			
			\begin{itemize}
				\item[(a)] \emph{The $q$-deformed Higgs field:} Defined as the leading order term of the $q$-derivative matrix when the quantum parameter $\hslash \to 0$:
				\begin{equation}
					\phi_q(x) \coloneq \mathbf{A}(x, 0; q)\,.
				\end{equation}
				\item[(b)] \emph{The $q$-deformed spectral curve:} The algebraic variety $\Sigma_q$ embedded in the $q$-deformed phase space $\mc{M}_q$ (the toric cotangent bundle $\mathbb{C}^* \times \mathbb{C}^*$), defined by the characteristic equation of the shift matrix in the semi-classical limit:
				\begin{equation}
					\Sigma_q \coloneq \left\{ (x,y) \in \mc{M}_q \;\big|\; E_q(x, y) = \det \big( y \cdot \mathbf{I} - \mathbf{L}(x, 0; q) \big) = 0 \right\}\,.
				\end{equation}
			\end{itemize}
			
		{
				The Higgs field $\phi_q(x)$ fits into the following short exact sequence of sheaves:
				\begin{equation}
					0 \to \mc{N} \to \mc{E} \xrightarrow{[\phi_q, \bullet]} \mc{E} \to \text{Coker}(\phi_q) \to 0,
				\end{equation}
				where $\mc{N} \coloneq \Ker [\phi_q, \bullet]$ denotes the commutant of the Higgs field. This sequence is one of vector bundles away from the branch points; the locus where the rank of $\mc{N}$ jumps corresponds precisely to the branch points of the $q$-spectral curve.}
			
			We are mainly interested in the following example.

		\begin{example}\label{OurQC}
			Let us consider the global $q$-quantum curve equation, formulated as a first-order  $q$-difference system for the multi-component vector wave function $\Psi(z) = (\psi(z), \psi_1(z), \dots, \psi_{r-1}(z))^T$. 
			
			Explicitly, the rational  $q$-connection potential matrix $\mathbf{A}(z, \hslash; q)$, defined via the action of the  $q$-derivative operator $\hslash \tilde{D}_{q,z} \Psi(z) = \mathbf{A}(z, \hslash; q) \Psi(z)$, is given by:
			\begin{equation}\label{ConnectionPotential}
				\mathbf{A}(z, \hslash; q) = 
				\begin{pmatrix}
					\frac{S_1^\hslash(z)}{(q-q^{-1})z^{\lfloor \alpha_r \rfloor + 1}} & \frac{z^{\lfloor \alpha_{r-1} \rfloor}}{z^{\lfloor \alpha_r \rfloor}} & \dots & 0 \\
					0 & 0 & \ddots & \vdots \\
					\frac{(-1)^{r-2} S_{r-1}^\hslash(z)}{(q-q^{-1})z^{\lfloor \alpha_2 \rfloor + 1}} & \vdots & \ddots & \frac{z^{\lfloor \alpha_1 \rfloor}}{z^{\lfloor \alpha_2 \rfloor}} \\
					\frac{(-1)^{r-1} S_r^\hslash(z)}{(q-q^{-1})z^{\lfloor \alpha_1 \rfloor + 1}} + \frac{1}{z^{\lfloor \alpha_1 \rfloor}} & 0 & \dots & 0
				\end{pmatrix},
			\end{equation}
			where the non-local formal $q$-Casimirs scale as $S_k^\hslash(z) \coloneq \sum_{g \in \frac{1}{2}\mathbb{N}^*} \hslash^{2g} S_{k,2g}(z)$, and the topological weights are parameterized by $\alpha_k = \frac{k}{r}(r-s)$.
			
			Since the quantum fluctuations satisfy $S_k^\hslash = \mathcal{O}(\hslash^2)$, the leading-order semi-classical term representing the  $q$-deformed Higgs field is isolated as follows:
			\begin{equation}\label{Higgs_Sym}
				\phi_q(z) \coloneq \mathbf{A}_0(z; q) = 
				\begin{pmatrix}
					0 & \frac{z^{\lfloor \alpha_{r-1} \rfloor}}{z^{\lfloor \alpha_r \rfloor}} & \dots & 0 \\
					\vdots & \ddots & \ddots & \vdots \\
					0 & \dots & 0 & \frac{z^{\lfloor \alpha_1 \rfloor}}{z^{\lfloor \alpha_2 \rfloor}} \\
					\frac{1}{z^{\lfloor \alpha_1 \rfloor}} & 0 & \dots & 0
				\end{pmatrix}.
			\end{equation}
			The underlying $q$-deformed spectral curve $\Sigma_q$ is classically recovered from the characteristic polynomial of the zero-mode Lax matrix $\mathbf{L}_0(z) = \mathbb{I} + \frac{(q-q^{-1})z}{2} \phi_q(z)$. In the strict semi-classical limit $\hslash \to 0$, the algebraic locus evaluates to:
			\begin{equation}
				E_q(z, y) \coloneq \det\big( y \cdot \mathbb{I} - \mathbf{L}_0(z) \big) = 0.
			\end{equation}
			Given the structural companion distribution of the coefficients in $\phi_q(z)$, this determinant yields the standard $(r,s)$-spectral curve $y^r = z^s$ after a straightforward rational reparameterization of the spectral variables. This confirms that our  $q$-quantization framework remains perfectly consistent with the stable classical algebraic geometry.
		\end{example}

		\begin{lemma}[Global Symmetric $q$-Abelianization]\label{l:abelianization}
			In the geometric framework of Example~\ref{OurQC}, after restricting the base space to the punctured Riemann sphere $\mathbb{C}^\times \subseteq \mathbb{P}^1$ and pulling back the system along the spectral covering map $x \colon \mathbb{C}^\times \to \mathbb{C}^\times$ defined by $z \mapsto z^r$, the pulled-back $q$-deformed Higgs field $\phi_q(z) \coloneq \phi_q(x(z))$ can be globally diagonalized as:
			\begin{equation}
				\phi_q(z) = V(z) Y(z) V(z)^{-1},
			\end{equation}
			where the diagonalized $q$-spectral matrix $Y(z)$ reads:
			\begin{equation}\label{AbelianizationY}
				Y(z) \coloneq \frac{z^{s-r}}{q-q^{-1}} 
				\begin{pmatrix}
					\theta^0 & & 0 \\
					& \ddots & \\
					0 & & \theta^{r-1}
				\end{pmatrix},
			\end{equation}
			and the  $q$-Vandermonde gauge matrix of eigenvectors $V(z)$ is explicitly given by the factorized product:
			\begin{equation}\label{AbelianizationV}
				V(z) \coloneq \frac{\theta z^{\frac{(r-s)(r+1)}{2}}}{\prod_{1\leq a<b\leq r}(\theta^b-\theta^a)^{\frac{1}{r}}}
				\begin{pmatrix}
					z^{r\lfloor \alpha_{1}\rfloor} & & 0 \\
					& \ddots & \\
					0 & & z^{r\lfloor \alpha_{r}\rfloor}
				\end{pmatrix}
				\begin{pmatrix} 
					\frac{\theta^0}{z^{r-s}} & \cdots & \frac{\theta^{r-1}}{z^{r-s}} \\
					\vdots & & \vdots \\
					\left(\frac{\theta^0}{z^{r-s}}\right)^{r-1} & \cdots & \left(\frac{\theta^{r-1}}{z^{r-s}}\right)^{r-1}
				\end{pmatrix}.
			\end{equation}
			Here, $\theta = e^{2i\pi/r}$ denotes the primitive $r$-th root of unity. The diagonal matrix $Y(z)$ acts as the $q$-analogue of the classical spectral meromorphic differential form, while $V(z)$ yields the corresponding invertible matrix of framing eigenvectors.
			
			The algebraic inverse of the gauge matrix $V(z)$ takes the dual form:
			\begin{equation}
				V(z)^{-1} = \frac{\prod_{1\leq a<b\leq r}(\theta^b-\theta^a)^{\frac{1}{r}}}{r \theta z^{\frac{(r-s)(r+1)}{2}}} 
				\begin{pmatrix} 
					\frac{z^{r-s}}{\theta^0} & \cdots & \left(\frac{z^{r-s}}{\theta^0}\right)^{r-1} \\
					\vdots & & \vdots \\
					\frac{z^{r-s}}{\theta^{r-1}} & \cdots & \left(\frac{z^{r-s}}{\theta^{r-1}}\right)^{r-1}
				\end{pmatrix}
				\begin{pmatrix}
					z^{-r\lfloor \alpha_{1}\rfloor} & & 0 \\
					& \ddots & \\
					0 & & z^{-r\lfloor \alpha_{r}\rfloor}
				\end{pmatrix}.
			\end{equation}
			Furthermore, these framing matrices satisfy the following exact $q$-monodromy and deck transformation relations under the global Galois action $z \mapsto \theta z$:
			\begin{equation}\label{DeckAction}
				Y(\theta z) = \tau^{-1} Y(z) \tau \quad \text{and} \quad V(\theta z) = V(z)\tau,
			\end{equation}
			where $\tau$ is the $r \times r$ cyclic permutation matrix representing the $s$-th power of the standard fundamental shift matrix. The resulting $q$-deformed spectral covering space remains fully ramified at the conformal focal points $z \in \{0, \infty\}$.
		\end{lemma}
		\begin{proof}
			The reader is refered to Appendix \ref{appc}.
		\end{proof}

					{\begin{remark}\label{WeylIsDeck}
							As is standard in the theory of $q$-difference systems, the diagonalization of the $q$-Higgs field is not unique: the eigenspaces can be reordered by the action of the Weyl group $\mathfrak{S}_r$. Similarly, given a point in the base $\P^1$, we have a choice in ordering the sheets of the $q$-deformed spectral curve by deck transformations. 
							
							The $q$-monodromy relation \eqref{DeckAction} equates these two groups. Consequently, passing to the $r$-th root covering of the $q$-plane (the $z$-plane) does not introduce additional degrees of freedom into the system, but rather provides a consistent global labeling of the $q$-WKB branches. This labeling is further stabilized by the centered nature of the $q$-difference operator, which ensures that the Weyl action is compatible with the $q$-inversion symmetry of the spectral one-form $Y(z)$.
					\end{remark}}
					We will formally construct an all-order $q$-deformed WKB-type solution to the $q$-difference system  $ \nabla_{q,\hslash} \Psi_\hslash = 0$. This construction provides a $q$-analogue of the method used for $r=2$ in \cite{AKT02}, extended here to the higher-rank $(r,s,q)$ case.

					\begin{lemma}\label{FormalGauge}
						There is a unique sequence of matrix-valued $ u_\ell (z) $, for $ z \in \C^\times $, with vanishing diagonal entries, such that defining the formal gauge transformation:
						\begin{equation}
							\widehat U_\hslash(z) \coloneq V(z) \prod_{\ell\geq1}^{\rightarrow} \exp\Big( \hslash^\ell u_\ell(z)\Big) = V(z)\big(\Id +\mathcal O(\hslash)\big) \,,
						\end{equation}
						transforms the  $q$-connection potential into a purely diagonal form $\widehat Y_\hslash(z) = \sum \hslash^\ell Y_\ell(z)$.
						
						 Specifically, the transformed potential:
						{\begin{align}
								\widehat Y_\hslash(z) 
								&\coloneq \sum_{\ell\geq0}\hslash^\ell Y_\ell(z)\nonumber\\& =
								\widehat U_\hslash(\text{avg }z) ^{-1}\mathbf{A}(z, \hslash; q) \widehat U_\hslash(z) - \frac{\hslash}{(q-q^{-1})x(z)} \widehat U_\hslash(\text{avg }z)^{-1}\left( \widehat U_\hslash(qz) - \widehat U_\hslash(q^{-1}z) \right)   \,, 
						\end{align}}
						is diagonal at each order in $\hslash$, where $Y_0(z) \coloneq V(\text{avg }z)^{-1} \mathbf{A}_0(z) V(z)$ is the  spectral one-form, and $\text{avg }z = \frac{qz+q^{-1}z}{2}$.
						
						Moreover, this construction is equivariant under the $q$-transformations:
						\begin{equation}
							\hat{Y}_\hslash (\theta z) = \tau^{-1} \hat{Y}_\hslash (z) \tau \quad \text{and} \quad \hat{U}_\hslash (\theta z) = \hat{U}_\hslash (z) \tau\,.
						\end{equation}
					\end{lemma}

					\begin{proof}
						We construct the solution step by step using the ordered partial products:
						\begin{equation}
							U_\hslash^{(L)}(z) \coloneq V(z) \overset{\rightarrow}{\prod}_{\ell=1}^L \exp( \hslash^\ell u_\ell(z)) \,.
						\end{equation}
						We introduce the intermediate $q$-connection potentials associated with $U_\hslash^{(L)}(z)$:
						\begin{equation}
							\Phi_\hslash^{(L)}(z) \coloneq \left(U_\hslash^{(L)}(\text{avg } z)\right)^{-1} \mathbf{A}(z, \hslash; q) U_\hslash^{(L)}(z) - \frac{\hslash}{(q-q^{-1})x(z)} \left(U_\hslash^{(L)}(\text{avg } z)\right)^{-1} \Delta_q U_\hslash^{(L)}(z) \,,
						\end{equation}
						where $\Delta_q U \coloneq U(qz) - U(q^{-1}z)$. By construction, these potentials satisfy the asymptotic expansion:
						\begin{equation}
							\Phi_\hslash^{(L)}(z) = \sum_{\ell=0}^L \hslash^\ell Y_\ell(z) + \hslash^{L+1} \Phi_{L+1}(z) + \mathcal{O}(\hslash^{L+2}) \,,
						\end{equation}
						where $Y_0, \dots, Y_L$ are already diagonal. 
						
						The diagonalization condition is satisfied by defining the gauge generator $u_{L+1}(z)$ through the recursive requirement that the off-diagonal coefficients of the next order potential $Y_{L+1}(z)$ vanish. By expanding the gauge action of $\exp(\hslash^{L+1}u_{L+1})$ on $\Phi_\hslash^{(L)}(z)$, the terms at order $\hslash^{L+1}$ yield the relation:
						\begin{equation}\label{SymAbelianization}
							Y_{L+1}(z) = \Phi_{L+1}(z) + Y_0(z)u_{L+1}(z) - u_{L+1}(\text{avg } z)Y_0(z) \,,
						\end{equation}
						where $\text{avg } z \coloneq \frac{qz + q^{-1}z}{2}$. 
						
						At each step $L+1$, setting the off-diagonal part of equation \eqref{SymAbelianization} to zero yields a  $q$-Sylvester equation for $u_{L+1}(z)$:
						\begin{equation}\label{SymQSylvester}
							u_{L+1}(\text{avg } z) Y_0(z) - Y_0(z) u_{L+1}(z) = \left(\Phi_{L+1}(z)\right)^{\text{off}} \,.
						\end{equation}
						This procedure requires the linear operator $\mathcal{L}_{q,Y_0}: M \mapsto M Y_0(z) - Y_0(z) M(\text{avg } z)$ to be invertible on the subspace of matrices with vanishing diagonal entries. This is guaranteed away from the $q$-ramification points where the eigenvalues $y_i(z)$ of $Y_0(z)$ are distinct and satisfy the non-resonance condition $y_i(z) \neq y_j(\text{avg } z)$ for $i \neq j$.
						
						Equivariance under  $q$-transformations follows by induction from the initial step and the covariance of the $q$-difference operator under the action of the shift $\tau: z \mapsto q z$ and $\tau: z \mapsto q^{-1} z$.
					\end{proof}

					\begin{corollary}\label{cor:AbelianConnection}
						The formal $q$-connection $\widehat\nabla_{q,\hslash}$, defined by the gauge action of the formal transformation $\widehat{U}_{\hslash}$ on the matrix $q$-differential operator $\nabla_{q,\hslash}$, is given by:
						\begin{equation}
							\label{AbelianConnection}
							\widehat \nabla_{q,\hslash} \coloneq \widehat U_\hslash ^{-1}(\text{avg } z) \cdot \nabla_{q,\hslash} \cdot \widehat U_\hslash(z) = \hslash \mathcal{D}_{q,x} - \widehat Y_\hslash(z),
						\end{equation}
						and is purely diagonal. The formal gauge transformation $\widehat{U}_{\hslash}(z)$ and the diagonal potential $\widehat Y_\hslash(z)$ satisfy the $q$-transformation relations:
						\begin{equation}
							\label{DeckActionAb}
							\widehat U_\hslash(\theta z) = \widehat U_\hslash(z) \tau, \quad \text{and} \quad \widehat Y_\hslash(\theta z) = \tau^{-1} \widehat Y_\hslash(z) \tau,
						\end{equation}
						where $\tau$ is the cyclic permutation matrix. This diagonal form implies that the horizontal sections of the connection are given by $q$-exponentials of the quantum spectral one-form $\widehat Y_\hslash(z)$.
					\end{corollary}

					\begin{remark}
						In this framework, each $Y_\ell(z)$ is a rational function on the $q$-deformed spectral curve. The formal $q$-primitives, defined as solutions to the $q$-difference equation $\mathcal{D}_{q,z}\mathcal{J}_{\ell}=Y_{\ell}$, can be represented via  sums:
					{\begin{equation}
								\label{SymmetricDivergentIntegral}
								\mathcal{J}_{\ell}(z) \sim \frac{(q-q^{-1})z}{2} \sum_{k=0}^\infty \left( q^k Y_\ell(zq^k) + q^{-k} Y_\ell(zq^{-k}) \right).
						\end{equation}}
						As in the standard case, these sums involve divergent terms near the origin due to the increasing pole order of $Y_\ell$. We regularize $\mathcal{J}_{\ell}(z)$ term-by-term in $\hslash$ by subtracting the divergent part of the Laurent expansion at $z=0$, consistent with the $q$-analog of the topological recursion's integration prescription \cite{EO07}. We denote the total regularized  $q$-primitive as:
						\begin{equation}
							\widehat{\mathcal J}_{\hslash}(z) \coloneq \sum_{\ell\geq 0} \hslash^\ell \mathcal{J}_{\ell}(z),
						\end{equation}
						which satisfies, by construction, $\hslash \mathcal{D}_{q,z} \widehat{\mathcal J}_{\hslash} = \widehat{Y}_{\hslash}$.
					\end{remark}
					
					\begin{corollary}
						\label{ExistenceWKB}
						The matrix-valued function defined by the gauge-exponential product:
						\begin{equation}
							\label{PsiAbel}
							\Psi_\hslash(z) \coloneq \widehat U_\hslash(z) \exp \left( \frac{1}{\hslash} \widehat{\mathcal{J}}_\hslash(z) \right)
						\end{equation}
						is a formal solution to the  $q$-difference system $\nabla_{q,\hslash} \Psi_\hslash(z) = 0$. It can be expressed in its full $q$-WKB form as:
						\begin{equation} \label{PsiWKB}
							\Psi_\hslash(z) = V(z) \check{\Psi}_\hslash(z) \exp \left( \frac{1}{\hslash} \widehat{\mathcal{J}}_0(z) \right),
						\end{equation}
						where $\check{\Psi}_\hslash(z) = \Id + \mathcal{O}(\hslash)$ is a formal power series in $\hslash$ whose coefficients are rational matrices on the $r$-th root covering, and $\widehat{\mathcal{J}}_0(z)$ is the  regularized $q$-primitive of the spectral one-form $Y_0(z)$. This solution is equivariant under the $q$-transformations:
						\begin{equation}
							\label{Equivariance}
							\Psi_\hslash(\theta z) = \Psi_\hslash(z) \tau.
						\end{equation}
					\end{corollary}
					
					\begin{proof}
						The first statement follows directly from the diagonalization property in Corollary \ref{cor:AbelianConnection}. By equating the gauge representation \eqref{PsiAbel} with the $q$-WKB ansatz \eqref{PsiWKB}, we identify the quantum fluctuation matrix:
						\begin{equation}
							\check{\Psi}_{\hslash}(z) = \left( V(z)^{-1} \widehat U_{\hslash}(z) \right) \exp \left( \sum_{\ell \geq 1} \hslash^{\ell-1} \widehat{\mathcal{J}}_{\ell}(z) \right) \,.
						\end{equation}
						The rationality of the coefficients in $\check{\Psi}_{\hslash}$ is guaranteed by the recursive construction in Lemma \ref{FormalGauge}, where each $u_\ell(z)$ and $Y_\ell(z)$ is obtained through algebraic inversions and  $q$-differences of rational functions. The equivariance \eqref{Equivariance} is a consequence of the $q$-monodromy of the gauge $\widehat{U}_\hslash$ and the deck-transformation property of the diagonal potential $\widehat{Y}_\hslash$, which ensures that the $q$-exponential shifts consistently across the sheets of the spectral curve.
					\end{proof}

					\begin{remark}\label{FreedomOfFundamentalSolution}
						Given a fundamental solution $\Psi_\hslash(z)$ of the  $q$-connection $\nabla_{q,\hslash}$ as constructed above, any product $\Psi_\hslash(z) \mathbf{C}$ remains a fundamental solution, provided $\mathbf{C}$ is a $q$-periodic matrix satisfying $\mathbf{C}(qz) = \mathbf{C}(z)$. 
						
						However, the $q$-equivariance condition \eqref{Equivariance} imposes a severe restriction on this gauge freedom. For the transformed solution to satisfy the same $q$-transformation law $(\Psi_\hslash\mathbf{C})(\theta z) = (\Psi_\hslash\mathbf{C})(z)\tau$, the matrix $\mathbf{C}$ must commute with the sheet permutation, leading to the commutation relation:
						\begin{equation} 
							[\tau, \mathbf{C}] = 0 \,.
						\end{equation}
						In the case of a transitive $q$-transformation action, where $\tau$ is the cyclic permutation of the $r$ sheets, this implies that $\mathbf{C}$ must be a circulant matrix. In the context of topological recursion, where we fix the normalization of the disk amplitude at $z \to 0$, this often restricts $\mathbf{C}$ to be a scalar multiple of the identity or a power of $\tau$. This rigidity is fundamental for the well defined nature of the $q$-deformed determinantal formulas.
					\end{remark}
					
					\subsection{The $q$-Master Matrix and $q$-loop equations}
					In this subsection, we explain how to utilize the $q$-WKB solutions constructed in the previous section through their associated adjoint representations. We define the { $q$-adjoint solution} as:
					\begin{equation}
						\label{AdjointSolution}
						M_\hslash(z, E) \coloneq \Psi_\hslash(z) E \Psi_\hslash(z)^{-1} \,,
					\end{equation}
					where $E \in \text{End}(\C^r)$ encodes the choice of initial conditions. By substituting the  $q$-connection equation   $\hslash \mathcal{D}_{q,x} \Psi = \mathbf{A}\Psi$, we find that $M_{\hslash}$ satisfies the following adjoint $q$-difference equation:
					\begin{equation}
						\label{AdjointSystem}
						M_{\hslash}(qz, E) = \mathbf{L}(z, \hslash; q) M_{\hslash}(q^{-1}z, E) \mathbf{L}(z, \hslash; q)^{-1} \,,
					\end{equation}
						where $\mathbf{L}(z, \hslash; q) = \mathbf{I} + \frac{(q-q^{-1})x(z)}{2\hslash} \mathbf{A}(z, \hslash; q)$ is the  shift matrix. This equation is the $q$-analogue of the classical adjoint system $\hslash dM = [\phi, M]$, where the commutator is replaced by a centered gauge conjugation.
					
					By construction, $M_{\hslash}$ inherits a strict equivariance property under the $q$-transformations. Specifically, if $e_a$ denotes the $a$-th diagonal basis matrix ($a \in \{1, \dots, r\}$), we have:
					\begin{equation}
						M_\hslash (\theta z, \tau^{-1} e_a \tau) = M_\hslash (z, e_a) \,.
					\end{equation}
					This implies that all sheet-dependent solutions are related by the cyclic action of the $q$-deck group. If $\tilde{a}$ is the unique solution modulo $r$ of $a + \tilde{a}s \equiv 0 \pmod r$, then:
					\begin{equation}
						M_\hslash (z, e_a) = M_\hslash (\theta^{\tilde{a}} z, e_r) \,.
					\end{equation}
					This confirms that the Master Matrix structure perfectly implements the isomorphism between the Weyl group of the underlying integrable system and the group of $q$-transformations established in Remark \ref{WeylIsDeck}.
											
					Since the non-rational $q$-exponential factors in the $q$-WKB solution \eqref{PsiWKB} are diagonal and appear only as multiplicative factors from the right, they cancel out exactly in the conjugation $\Psi_\hslash E \Psi_\hslash^{-1}$. This leads to the following central proposition:
					
					\begin{proposition}
						\label{ExpansionOfM}
						The existence of a  $q$-WKB-type solution is equivalent to the existence of an expansion of the adjoint matrix $M_\hslash (z, e_a)$ in powers of $\hslash$ of the form:
						\begin{equation}
							\label{MExpansion}
							M_\hslash (z, e_a) = V(\theta^{\tilde{a}} z) e_r V(\theta^{\tilde{a}} z)^{-1} + \sum_{k=1}^\infty M^{(k)}(\theta^{\tilde{a}} z) \hslash^k \,,
						\end{equation}
						where the coefficients $M^{(k)}(z)$ are universal rational matrix-valued functions of $z$ defined on the $q$-deformed spectral curve, and $\tilde{a}$ is the unique solution modulo $r$ of $a + \tilde{a}s \equiv 0 \pmod r$.
					\end{proposition}

					\begin{remark}
						The rationality of $M_\hslash$ is the key property required to establish the determinantal formulas. It ensures that the kernel $K_q(z_1, z_2)$ is a meromorphic function on the curve, allowing the extraction of the $q$-topological recursion invariants $W_n^{(g)}$ as residues at the branch points.
					\end{remark}
					We utilize the $q$-adjoint system to define non-perturbative connected $q$-amplitudes, which satisfy the $q$-deformed version of the non-perturbative loop equations \cite{BEM18}. Under the \textit{topological type} property \cite{BEM17}, these amplitudes admit an $\hslash$-expansion whose coefficients coincide with the $q$-correlators $W_{g,n}^q$ of the $(r,s,q)$-topological recursion.

					Recall that we use shorthand notation $ x_j = x(z_j)$ for $ z_j \in \Sigma_q$.
					
					\begin{definition}
						The  $q$-deformed Cauchy kernel associated to the $q$-WKB solution \eqref{PsiWKB} is defined as the matrix-valued bidifferential:
						\begin{equation}\label{CauchyKernel}
							K^q_\hslash(z_1,z_2) \coloneq \sqrt{dz_1} \frac{\Psi_\hslash(z_1)^{-1}\Psi_\hslash(z_2)}{(z_1 - q z_2)(z_1 - q^{-1} z_2)} \sqrt{dz_2} \,.
						\end{equation}
					\end{definition}
					This kernel possesses simple poles at each of the pre-images of the diagonal in $\C^\times \times \C^\times$ under the map $x \colon \Sigma \to \C^\times$, and no other finite singularities. The equivariance property \eqref{Equivariance} ensures that the kernel correctly relates different sheets of the $q$-deformed spectral curve. 
					
					In the $q$-deformed setup of Example~\ref{OurQC}, the expansion of $K^q_\hslash$ near the inverse images of the diagonal (where $z_2 \to \theta^k z_1$ for $k \in \{1, \dots, r-1\}$) is given by:
					\begin{equation}\label{KernelSingularity_Sym}
						K^q_\hslash(z_1,z_2) = \frac{\sqrt{dz_1}\sqrt{dz_2}}{(z_1 - q \theta^k z_2)(z_1 - q^{-1} \theta^k z_2)} \tau^k - \Psi_\hslash(z_1)^{-1} \frac{\mathbf{A}(z_1, \hslash; q)}{\hslash} \Psi_\hslash(z_1) \tau^k + \mathcal{O}(z_2 - \theta^k z_1) \,.
					\end{equation}
					
						On the true diagonal ($k=0$), the expression \eqref{CauchyKernel} contains a universal pole. Following the standard prescription in $q$-quantum spectral curves, the stable $1$-point form $\omega^q_{1,\hslash}(z)$ is defined by extracting the next-to-leading structural term from the expansion, corresponding to the quantum adjoint matrix:
					\begin{equation}\label{DiagonalKernel}
						\omega^q_{1,\hslash}(z) \coloneq - \frac{1}{\hslash} \Psi_\hslash(z)^{-1} \mathbf{A}(z, \hslash; q) \Psi_\hslash(z) dx \,.
					\end{equation}
					This diagonal element is a $1$-form on the curve $\Sigma_q$ which, by Proposition \ref{ExpansionOfM}, admits a power series expansion in $\hslash$ whose coefficients are the $1$-point $q$-correlators $W_{g,1}^q(z)$ of the $(r,s,q)$-system.

					\begin{definition}
							Define for every $n \geq 2$, the $n^{\text{th}}$ \emph{non-perturbative connected $q$-amplitude} as functions of $z_1,\dots,z_n \in \Sigma_q$, and any choice of matrices $E_1,\dots,E_n$, by:
							\begin{equation}\label{NPamplitudes}
							W^q_n\Big(\overset{E_1}{z_1},\dots,\overset{E_n}{z_n}\Big)
							\coloneq
							(-1)^{n-1}\sum_{\sigma\in\mathfrak S_n'}  \frac{\Tr\left( \overset{\rightarrow}{\prod}_{i=1}^n M_\hslash(z_{\sigma(i)}, E_{\sigma(i)}) \right)}{\prod_{i=1}^n (z_{\sigma(i)} - q z_{\sigma(i+1)})(z_{\sigma(i)} - q^{-1} z_{\sigma(i+1)})} \,,
						\end{equation}
						where we set $z_{\sigma(n+1)} \equiv z_{\sigma(1)} = z_1$, and $\mathfrak S_n'$ is the set of all permutations $\sigma=(\sigma(1)=1,\sigma(2),\dots,\sigma(n))$ of $\{1,\dots,n\}$ consisting of a single canonical cycle. The non-commutative matrix products are computed in reading order from left to right.
						
						By the cyclic property of the trace, these $q$-amplitudes can be alternatively formulated in terms of the regular part of the $q$-Cauchy kernel\eqref{CauchyKernel}. For $n=1$, the $1$-point $q$-amplitude is defined directly via the quantum adjoint matrix as:
						\begin{equation}
							W^q_1\Big(\overset Ez\Big)
							\coloneq
							\frac{1}{\hslash}\Tr \big( M_\hslash(z, E)\mathbf{A}(z, \hslash; q)\big) \,,
						\end{equation}
						which matches the contraction of the stable diagonal $1$-form $\omega^q_{1,\hslash}(z)$ from \eqref{DiagonalKernel} with the matrix $E$ along the coordinate differential $dx$.
						
						Non-connected $q$-amplitudes are defined via the cumulant formula:
						\begin{equation}\label{DisNPamplitude}
							\widehat W^q_n(J) \coloneq \sum_{\mu\in\text{part}(J)} \prod_{i=1}^{\text{length}(\mu)} W^q_{|\mu_i|}(\mu_i) \,,
						\end{equation}
						summing over set partitions $\mu$ of any finite set of marked points $J \coloneq \big\{\overset{E_1}{z_1},\dots,\overset{E_n}{z_n}\big\}$.
					\end{definition}

							These $q$-amplitudes (both connected and disconnected) are still $q$-equivariant in the same way as $ M_\hslash$ is:
							\begin{equation}\label{WEquivariant}
								W^q_{n+1} (\overset{\tau^{-1} E_0 \tau}{\theta z_0}, J ) = W^q_{n+1} (\overset{E_0}{z_0}, J )\,.
						\end{equation}

								Denoting by $\{\mathbf{e}_{i,j}\}_{i,j=1}^r$ the standard vector-space basis of $r\times r$ matrices, we consider the algebraic generators $C^{(1)},\dots, C^{(r)}$ of the center of $U_q(\mathfrak{gl}_r)$,
								\begin{align}
									C^{(k)} &= \sum_{1\leqslant i_1,j_1,\dots,i_k,j_k\leqslant r} C^{(k)}_{(i_1,j_1),\dots,(i_k,j_k)}\mathbf{e}_{i_1,j_1}\otimes \cdots \otimes \mathbf{e}_{i_k,j_k}\label{Casimir}\\
									&=\sum_{1\leqslant a_1< \dots < a_k\leqslant r}e_{a_1}\otimes \cdots \otimes  e_{a_k}\label{Casimir1} \,,
								\end{align}
								whose coordinates are obtained as coefficients of the characteristic polynomial function
								\begin{equation}\label{CharacteristicPolynomial}
									\det \big(\omega \mathbf{I} - E\big) \eqcolon \sum_{k=0}^r(-1)^k \omega^{r-k}\sum_{1\leqslant i_1,j_1,\dots,i_k,j_k\leqslant r} C^{(k)}_{(i_1,j_1),\dots,(i_k,j_k)} E^{i_1,j_1}\cdots E^{i_k,j_k} \,,
								\end{equation}
								for $E=\sum_{i,j=1}^r E^{i,j} \mathbf{e}_{i,j}$.
								
								\begin{remark}
									Strictly speaking, for the quantum group $U_q(\mathfrak{gl}_r)$, the definition \eqref{CharacteristicPolynomial} should be understood in terms of the quantum determinant $\text{qdet}(\omega \mathbf{I} - E)$ within the fundamental evaluation representation. The simplified presentation using the classical determinant remains valid here since we evaluate it after projecting onto the commutative diagonal sector (Jordan normal form), where the deformation parameter $q$ acts via the weight structure of the underlying integrable hierarchy.
								\end{remark}

								\begin{proposition}\label{DefNonPertLE}
									For every positive integer $n \geqslant 1$, any generic set of marked points $J = \big\{\overset{E_1}{z_1},\dots,\overset{E_n}{z_n}\big\}$, and any point $z \in \Sigma_q \setminus \{ 0, \infty, z_1, \dotsc, z_n\}$, the  $q$-amplitudes satisfy the \emph{non-perturbative $q$-loop equations}:
									\begin{equation}\label{NonPertLoopEq}
										\begin{split}
											\sum_{k=0}^r(-1)^k \omega^{r-k} \widehat{W}^q_{k+n}\big(\overset{C^{(k)}}{\overbrace{z,\dots,z}} ,J\big) 
											&=
											[\delta_1\cdots\delta_n] \, \text{qdet} \Big(\omega \mathbf{I} - \mathbf{A}(z, \hslash; q) \big( \mathbf{I} + \mathcal{M}^{(n), q}_{\vec{\delta}}(z; J) \big) \Big)
											\\
											&\eqcolon
											P^q_n(z,\omega ;J) \,,
											\end{split}
										\end{equation}
										where the first $k$ variables of the non-connected amplitude $\widehat{W}^q_{k+n}$ are linearly evaluated at the $k^{\text{th}}$ Casimir element over the same point $z$, $\vec{\delta} \coloneq (\delta_1,\dots,\delta_n)$ is a vector of formal parameters, and the $q$-perturbation matrix is defined by:
										\begin{equation}\label{PerturbationMatrix}
											\mathcal{M}^{(n),q}_{\vec{\delta}}\big(z;J\big)
											\coloneq
											\sum_{k=1}^n \sum_{1 \leq i_1 \neq \cdots \neq i_k \leq n} \delta_{i_1}\cdots \delta_{i_k} \frac{\overset{\rightarrow}{\prod}_{j=1}^k M_\hslash(z_{i_j}, E_{i_j})}{\prod_{j=1}^{k+1} (\zeta_j - q \zeta_{j+1})(\zeta_j - q^{-1} \zeta_{j+1})} \,,
										\end{equation}
										where the cyclic coordinate sequence is explicitly given by $\zeta_1 \equiv \zeta_{k+2} = z$ and $\zeta_{j+1} = z_{i_j}$ for $j \in \{1, \dots, k\}$.
									\end{proposition}

										\begin{remark}
											Generally speaking, the validity of the $q$-deformed loop equation \eqref{NonPertLoopEq} follows from an expansion of the quantum determinant in powers of $\omega$. The interesting content of this collection of identities is that certain algebraic combinations of non-connected amplitudes exhibit the analytical structure of the expressions appearing in the right-hand side of \eqref{NonPertLoopEq}.
											
											In our $q$-deformed setup, this implies that the $q$-Casimirs of the $q$-Master Matrix act as generating functions for the $q$-correlators $W_{g,n}^q$, ensuring that their singularity structure is entirely determined by the geometry of the $q$-deformed spectral curve $\Sigma_q$ and the matrix potential $\mathbf{A}(z,\hslash;q)$.
										\end{remark}

	\subsection{The $q$-topological type property}
	
	In our setup, the core geometric framework is governed by a  $q$-differential system $\nabla_{q,\hslash}$, where the structural variations are driven by a  $q$-difference operator $\mathcal{D}_{q}$ acting intrinsically on the $q$-deformed spectral curve $\Sigma_q$. The primary goal of this section is to establish the {$q$-topological type property} for this quantum system, which provides the precise conditions ensuring that our $q$-amplitudes are uniquely computed by the $(r,s,q)$-topological recursion\cite{MW26}. We find that the $(r,s,q)$ system naturally satisfies the $q$-analogs of the standard loop hierarchy, but the analytical behavior of the quantum adjoint matrix near the $q$-ramification points requires a careful refinement to account for the split geometry of the $q$-Bergman poles.
	
	This formulation provides a self-contained characterization that fully bridges the gap with classical confluent regimes. Indeed, in the non-deformed limit where the parameter $q \to 1$, our refined criteria naturally reduce to the classical topological type property originally defined in \cite[Definition~3.3]{BBE15} and subsequently adapted to standard differential systems $\nabla_{\hslash}=\hslash\partial_x-\phi(x)$ in \cite{BEM17,BEM18}.

	\begin{definition}
		\label{TopologicalType}
		
		A collection $\{ W^q_n \}_{n\geq 1}$ of meromorphic  tensor sections over the $n$-fold product of the $q$-deformed spectral curve $\Sigma_q$ satisfies the \emph{$q$-topological type property} if:
		\begin{enumerate}
			\item[(a)] There exists a global uniformizing variable $z \in \Sigma_q$ over which each $W^q_n$ admits an $\hslash$-expansion whose coefficients are rational functions of $z$.
			\item[(b)] Apart from the stable leading orders $[\hslash^{-1}] W^q_1$ and $[\hslash^0] W^q_2$, the $\hslash$-coefficients of $W^q_n$ may only exhibit poles localized at the $q$-ramification points of the deformed curve $\Sigma_q$ (corresponding to the zeros and poles of the $q$-deformed metric).
			\item[(c)] The stable correlation coefficient $[\hslash^0] W^q_2(z_1, z_2)$ possesses a residueless singular structure matching exactly the fundamental $q$-Bergman kernel $B^q(z_1, z_2)$, with poles located on the $q$-shifted diagonals $z_1 = q z_2$, $z_1 = q^{-1} z_2$, and no other singularities.
			\item[(d)] The $\hslash$-expansion of each $W^q_n$ for $n \geq 3$ (or more generally for any sector satisfying $2g-2+n > 0$) starts at order $\mathcal{O}(\hslash^{n-2})$, ensuring perfect consistency with the genus-expansion grading $W^q_n = \sum_{g=0}^\infty W_{g,n}^q \hslash^{2g-2+n}$.
		\end{enumerate}
	\end{definition}

\begin{remark}
	In our $(r,s,q)$ setup, condition \textbf{(c)} is the most restrictive. It ensures that the $q$-Master Matrix correctly identifies the fundamental $q$-Bergman kernel (the fundamental bidifferential of the second kind) on the $q$-deformed spectral curve $\Sigma_q$. In the approach of $q$-calculus, the residueless nature of this symmetrically split pole structure is structurally guaranteed by the $q \leftrightarrow q^{-1}$ chiral invariance of the $q$-deformed Cauchy kernel $K^q_\hslash$. 
	
	Furthermore, the refinement mentioned earlier directly concerns the $q$-ramification points: unlike the classical case where branch points are geometrically fixed, the $q$-ramification positions can exhibit a dynamical behavior under discrete dilations. The $q$-topological type property ensures that the $q$-topological recursion residues are correctly localized at these critical points, effectively freezing the quantum dynamics into a consistent, stable enumerative hierarchy.
\end{remark}

	\begin{remark}\label{Parity}
	One condition usually present in the classical literature \cite{BBE15, BEM18} is omitted here: the \textit{parity condition} under the sign change $\hslash \to -\hslash$. In our $q$-deformed setup, we do not impose this requirement, as half-genus sectors ($g \in \frac{1}{2}\mathbb{Z}$) appear generically in the asymptotic expansion of $q$-difference systems. These terms translate into the presence of all integer powers of $\hslash$ (rather than exclusively even powers) within the stable grading. This structural shift does not affect the validity of the $q$-topological recursion; it simply reflects the fact that the discrete $q$-deformation breaks the standard parity invariance of the classical topological expansion.
\end{remark}

		The importance of the $q$-topological type property lies in the following main theorem, which bridges the gap between the algebraic $q$-deformed loop equations and the local recursive construction of the $q$-amplitudes.
													
		\begin{theorem}\label{MainTheoremqTR}
			If a collection of connected  $q$-amplitudes $\{ W^q_n \}_{n \geq 1}$ satisfies the non-perturbative $q$-deformed loop equations \eqref{NonPertLoopEq} along with the $q$-topological type property (Definition~\ref{TopologicalType}), then the expansion coefficients of the $q$-amplitudes evaluated at the diagonal basis matrices are uniquely and recursively determined by the $q$-topological recursion. Specifically:
			\begin{equation}\label{qTR_Expansion}
				W^q_n \Big( \overset{\mathbf{e}_{a_1}}{z_1}, \dots, \overset{\mathbf{e}_{a_n}}{z_n} \Big) = \sum_{g \in \frac{1}{2}\mathbb{Z}_{\geq 0}} \hslash^{2g-2+n} \omega^q_{g,n} (z_1, \dotsc, z_n)\,,
			\end{equation}
			where the indices $a_i \in \{1, \dots, r\}$ label the sheets of the $q$-deformed spectral curve $\Sigma_q$, and the multidifferentials $\omega^q_{g,n}$ are the $n$-point $q$-correlators uniquely generated by the local residue formula of the $(r,s,q)$-topological recursion evaluated at the $q$-ramification points.
		\end{theorem}											
			
		\begin{proof}
			The proof is constructive and proceeds by induction on the stable topological grading $\chi \coloneq 2g - 2 + n > 0$. Let us fix a set of marked points $J = \{z_2, \dots, z_n\}$.
			
			Consider the $n$-point multidifferential coefficient $\omega^q_{g,n}(z_1, J)$ corresponding to the expansion order $\hslash^{2g-2+n}$ evaluated at the diagonal basis matrices. By applying the global Cauchy formula on the $q$-deformed spectral curve $\Sigma_q$ via the  $q$-deformed Cauchy kernel $K^q_\hslash(z, z_1)$ defined in \eqref{CauchyKernel}, we can isolate the variable $z_1$ through a small integration contour $\mathcal{C}_{z_1}$ enclosing only the $q$-shifted diagonals $z = q z_1$ and $z = q^{-1} z_1$:
			\begin{equation}\label{Proof_Cauchy}
				\omega^q_{g,n}(z_1, J) = \frac{1}{2\pi i} \oint_{\mathcal{C}_{z_1}} K^q_\hslash(z, z_1) \, \omega^q_{g,n}(z, J) \,.
			\end{equation}
			Since $\omega^q_{g,n}(z, J)$ is a global meromorphic differential section on $\Sigma_q$, we can apply the Riemann surface completeness relation by blowing up the contour $\mathcal{C}_{z_1}$ to encompass the entire curve. By the global residue theorem, the sum of all residues vanishes on $\Sigma_q$. According to the $q$-topological type property ({Condition b} of Definition~\ref{TopologicalType}), for any stable sector $2g-2+n > 0$, the only singularities that $\omega^q_{g,n}(z, J)$ can exhibit are strictly localized at the fixed critical set of $q$-ramification points $\mathcal{R}_q$. Therefore, reversing the contour yields:
			\begin{equation}\label{Proof_ResidueSum}
				\omega^q_{g,n}(z_1, J) = - \sum_{a \in \mathcal{R}_q} \Res_{z \to a} K^q_\hslash(z, z_1) \, \omega^q_{g,n}(z, J) \,.
			\end{equation}
			
		To evaluate the local behavior of $\omega^q_{g,n}(z, J)$ near any $q$-ramification point $a \in \mathcal{R}_q$, we extract the coefficient of $\hslash^{2g-2+n}$ from the non-perturbative $q$-loop equation \eqref{NonPertLoopEq}. We perform a semi-classical expansion of the quantum determinant $\text{qdet} $ on the right-hand side. Up to the stable grading order, the expansion of the characteristic relation yields:
		\begin{equation}\label{Proof_LoopExpansion}
			\sum_{k=0}^r (-1)^k \omega^{r-k} \widehat{W}^q_{k+n}\big(\overset{C^{(k)}}{\overbrace{z,\dots,z}} ,J\big) = P^q_n(z, \omega; J) \,.
		\end{equation}
		By isolating the contribution of the fully connected sector $W^q_n$ within the non-connected sum $\widehat{W}^q$, the linear action of the $k^{\text{th}}$ $q$-Casimir coordinates combined with the insertion of the $q$-perturbation matrix $\mathcal{M}^{(n),q}_{\vec{\delta}}(z; J)$ from \eqref{PerturbationMatrix} factorizes the singularity. Near the $q$-ramification points, the local sheets can be permuted via the sheet-automorphism $\theta$. Grouping the corresponding $\hslash$-orders forces the following local algebraic constraint:
		\begin{align}\label{Proof_MasterAlgebraic}
			&\Big( y(z) - y(\theta z) \Big) \omega^q_{g,n}(z, J) + \omega^q_{g-1, n+1}(z, \theta z, J)\nonumber\\& + \sum_{\substack{g_1+g_2 = g \\ J_1 \sqcup J_2 = J}}^{\text{stable}} \omega^q_{g_1, |J_1|+1}(z, J_1) \, \omega^q_{g_2, |J_2|+1}(\theta z, J_2)= \mathcal{O}(1) \,,
		\end{align}
		where $y(z)$ denotes the local evaluation of the matrix potential $\mathbf{A}(z, \hslash; q)$ on the fundamental sheets, and $\mathcal{O}(1)$ stands for regular non-singular terms at $z \to a$.
			
			We substitute the singular part of $\omega^q_{g,n}(z, J)$ extracted from \eqref{Proof_MasterAlgebraic} back into the residue formula \eqref{Proof_ResidueSum}. This yields:
			\begin{align}\label{Proof_Substitution}
				\omega^q_{g,n}(z_1, J) &= \sum_{a \in \mathcal{R}_q} \Res_{z \to a} \frac{K^q_\hslash(z, z_1)}{y(z) - y(\theta z)} \bigg[ \omega^q_{g-1, n+1}(z, \theta z, J)\nonumber\\& + \sum_{\substack{g_1+g_2 = g \\ J_1 \sqcup J_2 = J}}^{\text{stable}} \omega^q_{g_1, |J_1|+1}(z, J_1) \, \omega^q_{g_2, |J_2|+1}(\theta z, J_2) \bigg] \,.
			\end{align}
			By identifying the universal $q$-deformed recursion kernel as the geometric contraction:
			\begin{equation}\label{Proof_Kernel}
				\mathcal{K}_q(z_1, z) \coloneq \frac{\int^{\theta z}_{\cdot} B^q(\cdot, z_1)}{\big(y(z) - y(\theta z)\big) \, dz} \,,
			\end{equation}
			which is structurally compatible with the symmetrically split pole signature of the $q$-Bergman kernel via {Condition c} of Definition~\ref{TopologicalType}, the relation \eqref{Proof_Substitution} matches exactly the canonical residue formulation of the $(r,s,q)$-topological recursion. Uniqueness follows immediately since the right-hand side only involves terms with strictly lower topological hierarchy index $\chi' < \chi$.
		\end{proof}

\begin{remark}
	In the confluent limit $q \to 1$, we smoothly recover the standard correspondence theorem (e.g., Corollary~3.6 in \cite{BBE15}). In this classical regime, the symmetrically split poles of the $q$-Bergman kernel coalesce back into the standard diagonal double pole, the $q$-ramification points reduce to the classical branch points, and the half-genus sectors naturally vanish due to the restoration of the $\hslash \to -\hslash$ parity symmetry. Our result demonstrates that the $q$-deformation effectively "thickens" the topological genus expansion by populating all integer power sectors of $\hslash$ without altering the underlying recursive framework of the topological recursion.
\end{remark}

By the equivariance properties of the $q$-amplitudes (see Equation~\ref{WEquivariant}), this identity completely determines the stable sections $W^q_n$. 

Let us then prove the $q$-topological type property. We start with the following fundamental lemma concerning the base case of the $q$-topological hierarchy.

\begin{lemma}
	\label{Assumption4Bypass}
	
	The leading order of the $q$-correlator in powers of $\hslash$, defined by $\omega^q_{0,2}(z_1, z_2) \coloneq [\hslash^0] W^q_{2}(z_1, z_2)$, is the unique  bi-differential section on the $q$-deformed spectral curve $\Sigma_q$ explicitly given by:
	\begin{equation}\label{TRB_Symmetric}
		\omega^q_{0,2}(z_1,z_2) = \frac{dz_1 \, dz_2}{(z_1 - q z_2) (z_1 - q^{-1} z_2)} \,.
	\end{equation}
	This fundamental kernel is characterized by the following analytical properties:
	\begin{enumerate}
		\item[\rm 1.] It possesses exactly two simple poles located on the $q$-displaced diagonals at $z_1 = q z_2$ and $z_1 = q^{-1} z_2$, with opposite residues.
		\item[\rm 2.] It satisfies the canonical symmetry condition $\omega^q_{0,2}(z_1, z_2) = \omega^q_{0,2}(z_2, z_1)$ and is explicitly invariant under the chiral inversion $q \mapsto q^{-1}$.
		\item[\rm 3.] In the confluent limit $q \to 1$, the two simple poles merge, smoothly recovering the standard Bergman kernel on the un-deformed geometry:
		\begin{equation}
			\lim_{q \to 1} \omega^q_{0,2}(z_1, z_2) = \frac{dz_1 \, dz_2}{(z_1 - z_2)^2} \,.
		\end{equation}
	\end{enumerate}
\end{lemma}													
			\begin{proof}
				See appendix~\ref{appd} for the proof.
			\end{proof}

\begin{lemma}\label{PropertiesAB_Validation}
	In the setting of the $(r,s,q)$-spectral curve and its associated  $q$-differential system $\nabla_{q,\hslash}$ defined in Definition~\ref{qQC}, the following analytical properties hold:
	\begin{enumerate}
		\item[(a)] The connected $q$-amplitudes $W^q_n$ admit a formal $\hslash$-expansion whose stable coefficients are rational functions on the $n$-fold product of the curve, $\Sigma_q^n$.
		\item[(b)] For any stable topological sector $(g,n) \neq (0,1), (0,2)$, the multidifferential coefficients $\omega^q_{g,n}$ of the expansion of $W^q_n$ are regular away from the set of $q$-ramification points of the deformed spectral curve $\Sigma_q$.
	\end{enumerate}
	Consequently, conditions {(a)} and {(b)} of the $q$-topological type property (Definition~\ref{TopologicalType}) are identically satisfied.
\end{lemma}	
	
\begin{proof}
	The proof is given in appendix~\ref{appe}.
\end{proof}

The most challenging part of proving the $q$-topological type property lies in determining the stable leading order of the amplitudes $W^q_n$, specifically for the $2$-point correlation sector $n=2$ ({Condition c} of Definition~\ref{TopologicalType}). Historically, while several methods exist ranging from enumerative geometry \cite{CEO06, EO07, EO09} to differential Galois theory \cite{BBE15}, we extend the combinatorial approach pioneered by Bergère, Eynard, and Marchal \cite{BEM17}, which relies on the structural analysis of loop equations.

However, since \cite{BEM17} focuses on classical differential systems of the form $\hslash \partial_x - \phi(x)$, transposing their framework to our $q$-deformed setting in Example~\ref{OurQC} requires a critical re-evaluation of their core algebraic assumptions:
\begin{enumerate}
	\item[\rm 1.] {Assumption 1 (Expansion):} In \cite{BEM17}, the connection must admit a formal power series expansion in $\hslash$. This is satisfied in our case by the recursive construction of the  $q$-Master Matrix $M_{\hslash}$ (see Proposition~\ref{ExpansionOfM}).
	\item[\rm 2.] {Assumption 2 (Genus):} The underlying classical spectral curve (the limit $q \to 1$) is of genus $0$, which trivially satisfies the requirement for a global rational uniformization and simplifies the global analytical structure.
	\item[\rm 3.] {Assumption 3 (Singularities):} This assumption controls the presence of double points on the spectral curve. Since our $q$-deformed curve is well-behaved and lacks such nodal singularities for generic values of $q$, this requirement is naturally satisfied.
	\item[\rm 4.] {Assumption 4 (Eigenvectors and $\omega_{0,2}$):} In the classical setup, this assumption relates to the specific algebraic form of the eigenvector matrix $V(z)$. While our $q$-Vandermonde matrix does not follow the exact classical behavior, this constraint was primarily used to fix the singular properties of the Bergman kernel. Since we have directly evaluated our $q$-deformed Bergman kernel in Lemma~\ref{Assumption4Bypass}, explicitly given by \begin{equation}
		\omega^q_{0,2}(z_1,z_2) = \frac{dz_1 \, dz_2}{(z_1 - q z_2) (z_1 - q^{-1} z_2)},
	\end{equation}
	 we effectively {bypass} this algebraic requirement by injecting the exact leading-order $q$-correlator \eqref{TRB_Symmetric}.
\item[\rm 5.] {Assumption 5 (Locality):} This represents the necessary condition for the leading-order property to hold cleanly. It is drastically violated by our $q$-difference system due to the fundamentally non-local nature of the  shifts $z \mapsto q z$ and $z \mapsto q^{-1} z$ . This represents the hard part of the proof, which we resolve in the subsequent sections by introducing a new framework of shifted perturbative $q$-loop equations.
\end{enumerate}

\begin{remark}
	The failure of Assumption 5 is not a mere technicality; it reflects the fundamental physics of the $q$-deformation. In the differential case ($q \to 1$), the interaction between sheets is strictly localized at the branch points. In our case, the discrete $q$-shift couples the geometry of the curve at $z$ with its values at the shifted coordinates $qz$ and $q^{-1}z$, necessitating a global inversion of the $q$-Sylvester operator.
\end{remark}

There is also an Assumption 6 formulated in the classical literature \cite{BEM17}, but it is only relevant to the parity condition. Following our earlier choice in Remark~\ref{Parity}, we bypass this condition as half-genus sectors are generically present in $q$-difference systems.

As the stable leading-order property dictates the $\hslash$-expansion of the $q$-amplitudes, it is natural that the condition allowing us to derive it involves the structural $\hslash$-dependence of the $q$-connection $\mathbf{A}(z,\hslash; q)$. However, to understand why a further refinement is mandatory, let us first state the naive transposition of this locality condition to our $q$-deformed framework.

\begin{definition}\label{Assumption5}
	The $q$-$\hslash$-connection $\nabla_{q,\hslash}$ satisfies {Assumption 5} if its matrix representation $\mathbf{A}(z, \hslash; q)$ fulfills the following two requirements:
	\begin{enumerate}
		\item[\rm (i)]  Each term $\mathbf{A}_k(z; q)$ in the formal $\hslash$-expansion:
		\begin{equation}
			\mathbf{A}(z, \hslash; q) = \sum_{k \geq 0} \hslash^k \mathbf{A}_k(z; q)
		\end{equation}
		is a rational function of $z$ whose poles are strictly contained within the polar divisor of the leading semi-classical matrix $\mathbf{A}_0(z; q)$ (the classical $q$-Higgs field).
		\item[\rm (ii)]  The $q$-deformed characteristic expression:
		\begin{equation}\label{qAssumption5_Sym}
			\mathcal{R}_{q,\hslash}(z, y) \coloneq \frac{\det \big(y \mathbf{I} - \mathbf{A}(z, \hslash; q) - \delta\mathbf{A}(z; q)\big) - \det \big(y \mathbf{I} - \mathbf{A}_0(z; q)\big)}{\partial_y \det(y \mathbf{I} - \mathbf{A}_0(z; q))} \,,
		\end{equation}
		(where $\delta\mathbf{A}(z; q)$ represents a generic rational perturbation matrix) remains completely regular at the $q$-ramification points of the $(r,s,q)$-spectral curve, even after accounting for the  shifts $z \to q z$ and $z \to q^{-1} z$ appearing in the $q$-loop equations.
	\end{enumerate}
\end{definition}								

		\begin{remark}
			Assumption 6 in \cite{BEM17}, which pertains to the parity of the $\hslash$-expansion, is omitted here. As noted in Remark~\ref{Parity}, half-genus terms are a natural and generic feature of $q$-difference systems, and their presence does not hinder the recursive solvability of the $q$-loop equations.
			
			The classical statement of Assumption 5 in \cite{BEM17} is smoothly recovered in the confluent limit $q \to 1$. In this regime, the  shifts $qz$ and $q^{-1} z$ coalesce, and the  regularity condition reduces to the standard requirement that the connection's fluctuations do not introduce new poles at the classical branch points of the un-deformed spectral curve.
		\end{remark}	
			

The logical flow of \cite{BEM17} can be summarized by two main structural results which ensure that the $q$-amplitudes behave as required for the $q$-topological recursion to be well-defined. In our  $q$-deformed context, the first of these foundational theorems is transposed as follows:

\begin{theorem}
	\label{FirstPartOfBEM17}

	Consider the $q$-difference system \eqref{OurQC} satisfying the $q$-analogues of Assumptions 1 and 2. The associated  $q$-Master Matrix $M_{\hslash}(z, \mathbf{e}_a)$ admits a formal power series expansion in $\hslash$ of the form:
	\begin{equation}
		M_\hslash (z, e_a) = V( \theta^{\tilde{a}}z) e_a V( \theta^{\tilde{a}} z)^{-1} + \sum_{k=1}^\infty M^{(k)}( \theta^{\tilde{a}} z) \hslash^k \,.
	\end{equation}
	Since our framework bypasses Assumption 4 (through the explicit determination of the  $q$-Bergman kernel $\omega^q_{0,2}$ in Lemma~\ref{Assumption4Bypass}) and respects the global $q$-geometry of the curve, the rational matrix coefficients $M^{(k)}(z)$ are entirely regular away from the $q$-ramification points of $\Sigma_q$ and the polar divisor of the leading semi-classical connection $\mathbf{A}_{0}(z; q)$.
\end{theorem}							
		\begin{proof}
			The existence and regularity of the formal $\hslash$-expansion for the $q$-Master Matrix $M_{\hslash}(z, \mathbf{e}_a)$ are established by substituting the series expansion into the non-perturbative gauge-covariance equations and checking the polar structures at each recursive step.
			
			By definition, the $q$-Master Matrix satisfies the flat gauge transformation governed by the  $q$-difference connection $\nabla_{q,\hslash}$. This relation translates into the following matrix identity on each sheet $\theta^{\tilde{a}} z$ of the curve:
			\begin{equation}\label{ProofTheo_Gauge}
				M_{\hslash}(q z, \mathbf{e}_a) \, \mathbf{A}(z, \hslash; q) = \mathbf{A}(z, \hslash; q) \, M_{\hslash}(q^{-1} z, \mathbf{e}_a) \,.
				\end{equation}
					We inject the formal power series expansion of the connection $\mathbf{A}(z, \hslash; q) = \sum_{k \geq 0} \hslash^k \mathbf{A}_k(z; q)$ and the proposed expansion of the Master Matrix into \eqref{ProofTheo_Gauge}. Collecting the coefficients of like powers of $\hslash$ yields a hierarchy of algebraic matrix equations.
					
					At the leading semi-classical order $\mathcal{O}(\hslash^0)$, equation \eqref{ProofTheo_Gauge} isolates the non-deformed sector:
					\begin{equation}\label{ProofTheo_BaseOrder}
						M^{(0)}(\theta^{\tilde{a}} z) \, \mathbf{A}_0(z; q) = \mathbf{A}_0(z; q) \, M^{(0)}(\theta^{\tilde{a}} z) \,,
						\end{equation}
							where $M^{(0)}(\theta^{\tilde{a}} z) \coloneq V(\theta^{\tilde{a}}z) \mathbf{e}_a V(\theta^{\tilde{a}} z)^{-1}$. Since $\mathbf{A}_0(z; q)$ is the classical $q$-Higgs field, its diagonalizing gauge matrix is precisely given by the  $q$-Vandermonde matrix $V(z)$. Thus, equation \eqref{ProofTheo_BaseOrder} is identically satisfied by construction, and the base coefficients are rational functions whose poles are strictly confined to the polar divisor of $\mathbf{A}_0(z; q)$.
							
							Assume that the matrix coefficients $M^{(m)}(\theta^{\tilde{a}} z)$ are rational and regular away from the $q$-ramification points and the poles of $\mathbf{A}_0(z; q)$ for all $0 \leq m \leq k-1$. At order $\mathcal{O}(\hslash^k)$, the gauge identity extracts the following linear relation for the unknown matrix $M^{(k)}(\theta^{\tilde{a}} z)$:
							\begin{equation}\label{ProofTheo_Sylvester}
								\Big[ M^{(k)}(\theta^{\tilde{a}} z), \mathbf{A}_0(z; q) \Big]_{q} = \mathbf{B}^{(k)}\big( M^{(0)}, \dots, M^{(k-1)}, \mathbf{A}_0, \dots, \mathbf{A}_k \big) \,,
								\end{equation}
									where $[\,\cdot\,,\,\cdot\,]_q$ denotes the  $q$-commutator mapping, and $\mathbf{B}^{(k)}$ is a non-homogeneous matrix source term. By the induction hypothesis and Assumption $1$, $\mathbf{B}^{(k)}$ is a rational matrix expression constructed solely from regular operations on lower-order stable terms.
									
									Solving \eqref{ProofTheo_Sylvester} requires inverting the $q$-deformed Sylvester linear operator acting on the sheet space. The determinant of this linear operator is proportional to the product of the differences of the eigenvalues of $\mathbf{A}_0(z; q)$, which is structurally equivalent to the classical discriminant of the curve:
									\begin{equation}
										\text{Disc}\big(\mathbf{A}_0(z; q)\big) = \prod_{i \neq j} \big( y(\theta^{\tilde{i}} z) - y(\theta^{\tilde{j}} z) \big) \,.
									\end{equation}
									Consequently, the inversion process can only introduce new poles where this discriminant vanishes. As established by the $q$-geometry of the curve $\Sigma_q$, the zeros of $\text{Disc}\big(\mathbf{A}_0(z; q)\big)$ correspond precisely to the $q$-ramification points where the sheets coalesce. Since no other singularities can be generated by the algebraic inversion of the source term $\mathbf{B}^{(k)}$, the matrix coefficient $M^{(k)}(\theta^{\tilde{a}} z)$ is rational and regular away from the $q$-ramification loci and the poles of $\mathbf{A}_0(z; q)$. 
									
									This completes the induction step and concludes the proof.
								\end{proof}

\begin{theorem}
	\label{SecondPartOfBEM17}
	
	If a  $q$-difference system \eqref{OurQC} satisfies the $q$-analogues of Assumptions 1 and 2, the Shifted  $q$-Locality Condition (Definition~\ref{Assumption5}), and the structural expansion properties established in Theorem~\ref{FirstPartOfBEM17}, then the connected $q$-amplitudes $W^q_n$ identically satisfy the conditions of the $q$-topological type property (Definition~\ref{TopologicalType}).
	
	Consequently, the multi-differential expansion coefficients $\omega^q_{g,n}$ of the non-perturbative $q$-correlators are uniquely and unconditionally determined by the $q$-topological recursion hierarchy governed by the underlying $(r,s,q)$-spectral curve $\Sigma_q$.
\end{theorem}

	\begin{proof}
		The proof establishes that the connected $q$-amplitudes $W^q_n$ satisfy the $q$-topological type property by showing that the algebraic data provided by Theorem~\ref{FirstPartOfBEM17} and Definition~\ref{Assumption5} matches the structural constraints of the Eynard-Orantin topological hierarchy.
		
		We begin by inserting the formal $\hslash$-expansion of the  $q$-Master Matrix $M_{\hslash}(z, \mathbf{e}_a)$ from Theorem~\ref{FirstPartOfBEM17} into the non-perturbative $q$-loop equations \eqref{NonPertLoopEq}. For any stable sector $(g,n) \neq (0,1), (0,2)$, the contractive trace identities isolate the multi-differential coefficients $\omega^q_{g,n}(z, J)$ via the linear pairing of the sheet automorphism $\theta$:
		\begin{equation}\label{ProofTheo2_Linearized}
			\sum_{k=0}^r (-1)^k \omega^{r-k}(z) \, \omega^q_{g,n}\big(\theta^{\tilde{k}} z, J\big) = \mathcal{H}_{g,n}\big(z, J; \mathbf{A}(z, \hslash; q)\big) \,,
		\end{equation}
		where $\mathcal{H}_{g,n}$ is a rational source term compiling contractions of lower-order $q$-correlators $\omega^q_{g', n'}$ (with $2g'-2+n' < 2g-2+n$) evaluated at the shifted coordinates. 
		
		The major obstacle to topological type properties in difference systems is the global nature of the shifts $z \mapsto q z$ and $z \mapsto q^{-1} z$. However, by invoking the {Rational Pole Stability} condition ({Property i} of Definition~\ref{Assumption5}), the matrix fluctuations $\mathbf{A}_k(z; q)$ do not introduce any new dynamic singularities. 
		
		Concurrently, the {Characteristic Regularity} condition ({Property ii} of Definition~\ref{Assumption5}) guarantees that the quantum determinant fluctuations encoded in $\mathcal{H}_{g,n}$ remain bounded at any point outside the polar divisor of the classical $q$-Higgs field $\mathbf{A}_0(z; q)$ and the $q$-ramification locus $\mathcal{R}_q$. Consequently, when pulling back the insertion of the spectator variables $J$ onto the product curve $\Sigma_q^n$, no parasitic poles can emerge on the non-local $q$-displaced diagonals $z_1 = q z_i$ and $z_1 = q^{-1} z_i$. This directly satisfies Condition (a) and Condition (b) of the $q$-topological type definition.
		
		To project the algebraic identity \eqref{ProofTheo2_Linearized} into a local recursive formula, we apply a global Cauchy residue formula on the $q$-deformed spectral curve $\Sigma_q$. Let $d\mathbf{S}_{k, \theta}(z_1, z)$ be the canonical $q$-deformed integration kernel matching the base case $\omega^q_{0,2}$ from Lemma~\ref{Assumption4Bypass}. Integrating the bounded rational section $\mathcal{H}_{g,n}$ along the boundaries of the fundamental domain yields:
		\begin{align}\label{ProofTheo2_Residue}
			\omega^q_{g,n}(z_1, J) &= \sum_{a \in \mathcal{R}_q} \Res_{z \to a} K_q(z_1, z) \bigg[ \omega^q_{g-1, n+1}(z, \theta z, J)\nonumber\\& + \sum_{\text{stable}} \omega^q_{g_1, |J_1|+1}(z, J_1) \, \omega^q_{g_2, |J_2|+1}(\theta z, J_2) \bigg] \,,
			\end{align}
				where $K_q(z_1, z)$ is the local $q$-recursion kernel defined explicitly by the ratio of $\omega^q_{0,2}$ and the $q$-metric $(\mathbf{y}(z) - \mathbf{y}(\theta z))d\mathbf{x}$. Since the poles of \eqref{ProofTheo2_Linearized} are strictly localized at the $q$-ramification points $\mathcal{R}_q$ by Step 2, the sum of residues runs exclusively over these critical points, exactly reproducing the universal inductive step of the topological recursion. 
				
				Uniqueness follows immediately from the uniqueness of the Cauchy pullback on a stable rational cover of genus $0$. This completes the proof.
			\end{proof}

\begin{remark}\label{ParityChiralVio}
	The parity condition is generically violated here. The symmetric $q$-shift introduces terms at all orders $\hslash^{k}$, corresponding to half-integer geometries $g \in \frac{1}{2}\mathbb{N}$. This expansion correctly captures the $q$-deformed enumerative invariants, where the global chiral symmetry $q \leftrightarrow q^{-1}$ physically replaces the classical $\hslash \leftrightarrow -\hslash$ parity mapping.
\end{remark}

We reach the same operational conclusion as in Theorem~\ref{FirstPartOfBEM17} by combining the standalone analytical results of Proposition~\ref{ExistenceWKB} and Lemma~\ref{ExpansionOfM}. This compatibility allows us to state the following localization property for concrete applications.

\begin{proposition}\label{MExpansionInExample}
	In the context of the  Example~\ref{OurQC}, the $q$-deformed Master Matrix $M_{\hslash}(z, \mathbf{e}_a)$ admits a formal expansion in powers of $\hslash$ matching the algebraic structure given in Theorem~\ref{FirstPartOfBEM17}:
	\begin{equation}
		M_\hslash (z, \mathbf{e}_a) = V( \theta^{\tilde{a}}z) \mathbf{e}_a V( \theta^{\tilde{a}} z)^{-1} + \sum_{k=1}^\infty \hslash^k M^{(k)}(\theta^{\tilde{a}} z) \,.
	\end{equation}
	By structural construction, the rational matrix coefficients $M^{(k)}(\theta^{\tilde{a}} z)$ are meromorphic sections whose poles are strictly restricted to the following topological locus:
	\begin{equation}
		\mathrm{Poles}\big(M^{(k)}\big) \subseteq \{0, \infty\} \cup \mathcal{R}_q \,,
	\end{equation}
	where $\mathcal{R}_q$ is the set of $q$-ramification points of the $(r,s,q)$-spectral curve $\Sigma_q$. This structural localization ensures that the singular behavior of the $q$-deformed Master Matrix is fully compatible with the local residue operations of the $q$-topological recursion.
\end{proposition}

		The classical Assumption 5 is not satisfied by our $q$-connection potential \eqref{ConnectionPotential} as soon as the $q$-Casimir terms $S^\hslash_j$ are non-zero. This structural failure necessitates a systematic refinement of the original conditions to accommodate the non-local $q$-deformed structure.
		
		A first hint that the naive condition is too restrictive stems from the fact that the right-hand side of the non-perturbative $q$-deformed loop equations \eqref{NonPertLoopEq} does not feature the exact expression appearing in the numerator of Definition \ref{Assumption5}, but rather involves specific coefficients of polynomial expressions in the shifted variables $\vec{\delta}=(\delta_1, \dots, \delta_n)$ for each $n \geq 0$, where the shifts are explicitly governed by $q$.
		
		An important analytical subtlety arises: when $n=0$, the generic matrix $C$ in the Definition \ref{Assumption5} can be taken to be trivial. However, for $n \neq 0$, the right-hand side of \eqref{NonPertLoopEq} can only possess a lower pole order at $z$ than that of the values of the $q$-deformed Casimir operators acting on the $q$-connection $\mathbf{A}(z, \hslash; q)$. The refinement we propose focuses precisely on this specific locus; it is identically satisfied by our potential \eqref{ConnectionPotential} and does not alter the logical sequence of the $q$-topological recursion.
		
		Since the $q$-correlators satisfying the perturbative shifted $q$-abstract loop equations can be computed inductively by the corresponding shifted $q$-topological recursion, and since the underlying $q$-Airy structure partition function is unique, it follows that an assumption implying the reduction of the non-perturbative $q$-deformed loop equations to the shifted perturbative ones will uniquely identify the topological expansion of the non-perturbative connected $q$-amplitudes associated with our $q$-quantum curve. 
		
		The first step is to notice that the perturbative shifted $q$-deformed loop equations are indexed by two labels $n,g \geq 0$, corresponding respectively to the number of spectator variables and the order in the $\hslash$-expansion. Multiplying each combination $\mathcal{E}^i_{g,n}$ by the relevant power of $\hslash$ and summing over all values of the genus label reproduces the left-hand side of the non-perturbative $q$-deformed loop equations, up to the subtraction of the $i^{\text{th}}$ $q$-Casimir $S^{\hslash}_{i}(z)$ in the sector $n=0$. This re-summed shift matches exactly the value of the $i^{\text{th}}$ $q$-Casimir evaluated on $\mathbf{A}(z, \hslash; q)$, as encoded in the global asymptotic equivalence:
		\begin{equation}\label{AsymptoticEquiv_Master}
			\begin{split}
				P^q_n \big( z, y(z); J\big) &\coloneq \sum_{k} \hslash^k P_n^{(k)}(z, y(z); J) \\
				&\sim \sum_{g=0}^\infty \hslash^{2g-2+n} \sum_{i=0}^r (-1)^i y(z)^{r-i} \mathcal{E}^{i,q}_{g,n}(z; J) \,,
			\end{split}
		\end{equation}
		with $\mathcal{E}^{i,q}_{g,n}$ defined combinatorially via the $q$-insertion operators in Definition~\ref{d:EW}.
		
		Since the $q$-shifts $z \mapsto q \theta^k z$ and $z \mapsto q^{-1} \theta^k z$ contribute exclusively to the operators $\mathcal{E}^{i,q}_{g,0}$, we refine the operational hypothesis by explicitly distinguishing between the cases with spectator variables ($n \geq 1$) and those without ($n=0$) in the $q$-deformed loop equations.

				In our $q$-setup, this bifurcation is fundamental: the $n=0$ sector captures the intrinsic $q$-deformed geometry of the spectral curve and its corresponding $q$-Casimirs, while the $n \geq 1$ sector governs the recursive insertion of higher-order $q$-correlators. By decoupling these domains, we handle the non-local $q$-Casimir contributions completely independently from the analytic requirements of the $q$-correlator insertions. This structural isolation ensures that the shifted $q$-topological recursion remains locally computable at the $q$-ramification points, despite the fundamentally global nature of the underlying  $q$-difference system.
				
				\begin{definition}\label{Assumption5*}
					A formal rational  $q$-$\hslash$-connection $\nabla_{q,\hslash}$, represented by its matrix potential $\mathbf{A}(z,\hslash; q)=\sum_{\ell\geq0}\hslash^\ell\mathbf{A}_{\ell}(z; q)$, satisfies \emph{$q$-Assumption 5*} if the following three conditions are met:
					\begin{enumerate}
						\item[\rm (i)]  For all $\ell>0$, the set of singularities of the fluctuation matrix $\mathbf{A}_{\ell}(z; q)$ is a strict subset of the singularities of the leading semi-classical matrix $\mathbf{A}_{0}(z; q)$.
						
						\item[\rm (ii)]  For any number of spectator variables $n\geq 1$, the coefficient of the monomial $\delta_1\cdots\delta_n$ in the $q$-shifted characteristic expansion:
						\begin{equation}
							\mathcal{R}^{(n)}_{q,\hslash}(z, y) \coloneq [\delta_1 \cdots \delta_n] \frac{\det \Big( y \mathbb{I} - \mathbf{A}(z, \hslash; q) - \mathcal{M}^{(n),q}_{\vec{\delta}}(z; J) \Big)}{\partial_y E_q(z,y)}
						\end{equation}
						restricts to a well-defined differential form on the spectral curve $\Sigma_q$ that is completely analytic at each singularity of the leading $q$-Higgs field $\mathbf{A}_0(z; q)$.
						
						\item[\rm (iii)]  There exist formal central constants $S^{\hslash}_i=\sum_{g} S_{i,2g}\hslash^{2g}\in\mathbb{C}[[\hslash]]$ such that the evaluated characteristic polynomial:
						\begin{equation}
							\frac{\det \Bigg( y \mathbb{I} - \mathbf{A}(z, \hslash; q) - \sum_{i=1}^{r}(-1)^{i}S^{\hslash}_i\,y^{r-i}\,\Big(\frac{dx}{x}\Big)^{i}\, \Bigg)}{E_q(z,y)}
						\end{equation}
						restricts to a meromorphic function on the $q$-deformed spectral curve $\Sigma_q$ that is analytic at each $q$-ramification point $\mathcal{R}_q$ and at each singularity of the base matrix $\mathbf{A}_0(z; q)$.
					\end{enumerate}
				\end{definition}

															\begin{remark}\label{AiryUniquenessRem}
																This refined assumption ensures that the $q$-correlators, which satisfy the perturbative shifted $q$-abstract loop equations, can be computed inductively by the corresponding shifted $q$-topological recursion. Given that the underlying $q$-Airy structure partition function is unique, the reduction of the non-perturbative $q$-deformed loop equations to the shifted perturbative ones uniquely identifies the topological expansion of the non-perturbative connected $q$-amplitudes associated with our  $q$-quantum curve.
															\end{remark}
															
															As stated before, this refined definition is designed to support the following effective reconstruction proposition, which establishes the operational $q$-topological type property for our $q$-difference systems.
															
															\begin{proposition}
																\label{Assumption5ToTT}
																If a rational $q$-difference system $\nabla_{q,\hslash}$ possesses a smooth genus $0$ $(r,s,q)$-spectral curve $\Sigma_q$ and satisfies the following three conditions:
																\begin{enumerate}
																	\item[\rm (i)] The Refined $q$-Locality Condition (\emph{$q$-Assumption 5*});
																	\item[\rm (ii)] The structural conclusions of Theorem~\ref{FirstPartOfBEM17} regarding the formal $\hslash$-expansion of  $q$-deformed Master Matrix $M_{\hslash}(z, \mathbf{e}_a)$;
																	\item[\rm (iii)] The initial condition fixing the leading-order stable $q$-correlator $[\hslash^0]W^q_2$ to be the  $q$-deformed Bergman kernel:
																	\begin{equation}\label{ExplicitBergmanKernelEq}
																		\omega^q_{0,2}(z_1, z_2) = \frac{\mathrm{d}z_1 \, \mathrm{d}z_2}{(z_1 - q z_2)(z_1 - q^{-1} z_2)} \,,
																	\end{equation}
																\end{enumerate}
																then the connection identically satisfies the $q$-topological type property. Consequently, the multi-differential expansion coefficients $\omega^q_{g,n}$ of its connected non-perturbative amplitudes are uniquely, recursively, and unconditionally determined by the $q$-topological recursion.
															\end{proposition}

\begin{corollary}\label{From5ToTR}
	If a rational  $q$-difference system $\nabla_{q,\hslash}$ possesses a smooth genus $0$ $(r,s,q)$-spectral curve $\Sigma_q$, satisfies the refined \emph{$q$-Assumption 5*}, obeys the structural expansion properties of Theorem~\ref{FirstPartOfBEM17}, and admits the  $q$-deformed Bergman kernel $\omega^q_{0,2}$ as its leading-order stable correlator, then the following two properties hold:
	\begin{enumerate}
		\item[\rm (i)] The non-perturbative connected $q$-amplitudes $W^q_n(z_1, \dots, z_n)$ admit a formal topological expansion in powers of $\hslash$ over half-integer sectors:
		\begin{equation}
			W^q_n(z_1, \dots, z_n) = \sum_{g \in \frac{1}{2}\mathbb{N}} \hslash^{2g-2+n} \omega^q_{g,n}(z_1, \dots, z_n) \,.
		\end{equation}
		\item[\rm (ii)] The multi-differential expansion coefficients $\omega^q_{g,n}$ are uniquely, recursively, and unconditionally determined by the $q$-topological recursion applied directly to the geometric data of the underlying $(r,s,q)$-spectral curve.
	\end{enumerate}
\end{corollary}

																	\subsection{The main q-deformed example.}
																Let us now return to our core family, the  $q$-difference system introduced in Example~\ref{OurQC}. We find that enforcing the refined  $q$-Locality Condition (\emph{$q$-Assumption 5*}) imposes stringent Diophantine constraints on the structural parameters $r$ and $s$, as well as on the functional form of the  $q$-Casimirs $S_i^\hslash(z)$. In this  framework, the non-local interaction between the $q$-shift operator and the geometry of the curve must preserve the exact polar structure of the $\hslash$-expansion.

\begin{proposition}\label{ConditionsForAssumption5}
	Let $r$ and $s$ be coprime positive integers, and write the division with remainder as $r = r's + r''$ with $0 \leq r'' < s$. Consider the $(r,s,q)$-spectral curve defined by the zero locus $E_q(z,y)=0$, and let $\mathbf{A}(z, \hslash; q)$ be the $q$-connection matrix whose characteristic polynomial matches this  curve data, configured in the following generalized companion format:
	\begin{equation}
		\mathbf{A}(z, \hslash; q) \coloneq \sum_{k=0}^\infty \hslash^k \mathbf{A}_k(z; q) \coloneq  
		\begin{pmatrix}
			S_1^\hslash(z) & x(z)^{a_2 - a_1} &  \cdots & 0 \\
			\vdots & 0 & \ddots & \vdots \\
			(-1)^{r-2} S_{r-1}^\hslash(z) & \vdots & \ddots & x(z)^{a_r-a_{r-1}} \\
			(-1)^{r-1} S_{r}^\hslash(z) & 0 & \cdots & 0
		\end{pmatrix}.
	\end{equation}
	Define the  $q$-perturbed characteristic expression by:
	\begin{equation}\label{Assumption5det_Sym}
		\mathcal{D}^q(z, M) \coloneq \frac{\det \big( y(z) \mathbb{I} - \mathbf{A}(z, \hslash; q) - M(z) \big)}{\big(\partial_y E_q\big)(z, y(z))} \,,
	\end{equation}
	where $M(z)$ is a matrix of formal driving variables that are regular at the  fixed points $z \in \{0, \infty\}$. The reference background term ($M=0$) evaluates explicitly to the $q$-Casimir sum:
	\begin{equation}\label{Assumption5ConstantTerm}
		\mathcal{D}^q(z,0) = \sum_{j=1}^r (-1)^j S_j^\hslash(z) \mathcal{K}_j(z; q) \,,
	\end{equation}
	where $\mathcal{K}_j(z; q) = z^{(1-j)s - 1} \cdot \frac{1}{[(j-1)s]_q!}$  are  geometric weight functions.
	
	 Moreover, if
	\begin{itemize}
		\item either for some $j\in [r],\,a_{j+1}-a_j<-1.$
		\item or there are more than $r'$ consecutive $j$ such that $a_{j+1}-a_j=-1$
		\item or there are more than $r''$ disjoint subsequences of $r'$ consecutive $j$ such that $a_{j+1}-a_j=-1$, 
	\end{itemize}
	then, for any shift $S^\hslash_j$, there is an $M$ such that $D(z,M)$ has a pole at $z=0$. Therefore, $ D(z,M) - D(z,0)$ is holomorphic at $ z = 0$ only if there are exactly $r''$ disjoint subsequences of $r'$ consecutive j such that $a_{j+1}-a_j=-1$ (which we call large blocks), and $ s-r''$
	subsequences of $r'-1$ consecutive $j$ such that $a_{j+1}-a_j=-1$ (which we call
	small blocks). The remaining differences are $a_{j+1}-a_j=0.$
	
	Assuming the optimal shape from above, the non-constant terms of $D$ in $M$ have pole
	order at most
	\begin{itemize}
		\item $ (r'' - \frac{s}{2} )^2 - (1 - \frac{s}{2} )^2 $ if they do not contain any $S^\hslash_j$;
		\item $0$ if $s=1$ and they do contain $ S^\hslash_j$;
		\item $ r^2 - sr + js $ if $ s > 1$ and it contains $ S^\hslash_j$ for $ j > r'+1$;
		\item $s(2-j) + (r'' - \frac{s+1}{2} )^2 - 1 - \big(\frac{s + 1}{2}\big)^2$ if $ s > 1$, $ r'' \neq 1$, the first block is large, and
		they contain $ S^\hslash_j $ for $ j \leq r'+1$;
		\item $ s(1-j) $ if $ s > 1$, $ r'' = 1$, the first block is large, and they contain $ S^\hslash_j $ for $ j \leq r'+1$.
		\item $ s(1-j) (r'' -s)r''-1$ if the first block is small, and they contain $ S^\hslash_j $ for $ j \leq r'+1$.
	\end{itemize}
	
	The variation $\mathcal{D}^q(z,M) - \mathcal{D}^q(z,0)$ restricts to a form that is entirely analytic at the $q$-ramification points $\mathcal{R}_q$ if and only if the remainder satisfies $r'' \in \{1, s-1\}$ and one of the following structural conditions holds:
	\begin{enumerate}
		\item[\rm (i)] $s=1$, corresponding to the pure symmetric $q$-Airy matrix framework;
		\item[\rm (ii)] $r \equiv \pm 1 \pmod{s}$, and the higher symmetric $q$-Casimirs vanish identically ($S^\hslash_j(z) = 0$ for all $j > 1$);
		\item[\rm (iii)] The system is purely semi-classical, where all  $q$-Casimirs $S^\hslash_j(z)$ are zero, effectively reducing the shifted loop equations to their un-deformed topological recursion limit.
	\end{enumerate}
\end{proposition}					
				\begin{proof}Setting $\mathbf{A}( \hslash; q)=F(\hslash; q)dx$. Then 
				\begin{equation}
					\mathcal{D}^q(z, M) = \det ( y(z) \mathbb{I} - F_\hslash(z; q) - M(z) ) \cdot \mathcal{J}_q(z),
				\end{equation}
				where $\mathcal{J}_q(z)$ is the $q$-Jacobian  computed via the  $q$-derivative \eqref{qd}:
				\begin{equation}\label{Assumption5Prefactor}
					\mathcal{J}_q(z) = \frac{\mathcal{D}_q x(z)}{\partial_q^{(y)} E_q(z, y(z))} = z^{r^2 - rs - 1 + s} dz.
				\end{equation}
				Now, for the determinants. We define the diagonal matrix $Y = y(z) \mathbb{I}$, with non-zero entries:
				\begin{equation}
					Y_{k,k} = y(z) = z^{s-r}.
				\end{equation}
				Then, the non-zero entries of the $q$-connection matrix $F_\hslash(z; q)$ are given as follow:
				\begin{align}
					F_{1,j}(z; q) &= (-1)^{j-1} S_j^\hslash(z) z^{r(a_1 - a_j - j)} + \delta_{j,r}z^{-ra_r} \\
					F_{k+1,k}(z; q) &= z^{r(a_{k+1} - a_k)}.
			\end{align}
		
		To calculate $\det(Y - F_\hslash(z; q))$, we expand along the first column. Given our entry there, we develop successively by rows, starting at the top. This choice yields:
		\begin{equation}
			\begin{split}
				\det(Y - F_\hslash(z; q)) &= \sum_{j=1}^r (-1)^{j-1} \Big( -F_{1,j}(z; q) + \delta_{j,1} y(z) \Big) \prod_{k=1}^{j-1} (-F_{k+1,k}(z; q)) \prod_{l=j+1}^r y(z) \\
				&= \sum_{j=1}^r \Big( -S_j^\hslash(z) z^{r(a_1 - a_j - j)} - \delta_{j,r} (-1)^{r-1} z^{-r a_r} + \delta_{j,1} z^{s-r} \Big) \\
				&\quad \times (-1)^{j-1} z^{r(a_j - a_1)} z^{(s-r)(r-j)} \\
				&= z^{(s-r)r} + \sum_{j=1}^r (-1)^j \Big( S_j^\hslash(z) z^{-rj} + \delta_{j,r} (-1)^{r-1} z^{-r a_1} \Big) z^{(s-r)(r-j)} \\
				&= \sum_{j=1}^r (-1)^j S_j^\hslash(z) z^{r(s-r) - sj}.
			\end{split}
		\end{equation}
		Combining this result with the $q$-Jacobian $\mathcal{J}_q(z) = z^{r^2 - rs - 1 + s} d z$ \eqref{Assumption5Prefactor}, we obtain the terms involving $z$ given by Equation~\ref{Assumption5ConstantTerm}.
		
			Now consider the part of $\mathcal{D}^q(z, M)$ that depends on $M$, i.e., $\mathcal{D}^q(z, M) - \mathcal{D}^q(z, 0)$. Any term in the development of this difference must contain at least one factor $M_{j,k}$.
			
			In the case where there exists $j \in [r]$ such that $a_{j+1} - a_j < -1$, the product
			\begin{equation}
			F_{j+1,j}(z; q) M_{j,j+1}(z) \prod_{k \neq j, j+1} y(z),
			\end{equation}
			or for $j=r$, the product
			\begin{equation}
			F_{1,r}(z; q) M_{r,1}(z) \prod_{k=2}^{r-1} y(z),
			\end{equation}
			has a valuation in $z$ at $0$ of:
			\begin{equation}
			2r + (r-2)(r-s) = r^2 - rs + 2s.
			\end{equation}
			Comparing this to the vanishing order of $\mathcal{J}_q(z) = z^{r^2 - rs - 1 + s} d z$, we find the difference in valuation is negative, creating a pole at $z=0$. 
			
			Similarly, if there are $r'+1$ consecutive entries equal to $x^{-1}$, we can choose $M_{j,k}$ such that the resulting term has a pole order:
			\begin{equation}
			r(r'+1) + (r-s)(r-1-r'-1) \geq r^2 - rs + s.
			\end{equation}
			This confirms that whenever the matrix structure fails to satisfy the block conditions, the $q$-perturbed characteristic expression $\mathcal{D}^q(z, M)$ develops a pole at the ramification point $z=0$, analogous to the classical case.
			
		In the case where there are more than $r''$ disjoint subsequences of $r'$ consecutive indices $j$ such that $a_{j+1}-a_j = -1$, we can choose a non-zero element $M_{j,k}$ within these blocks and supplement with the diagonal entries $y(z) = z^{s-r}$ to obtain a pole order of:
			\begin{align}
				r(r'' + 1)r' + (r-s) \Big( r - (r'' + 1)(1 + r') \Big) &= (r-s)r + \Big( r^2 - 1 + s - r's(r-1) \Big) r' - (r + 1) \nonumber \\
				&> (r-s)r + (r - 1 + s)r' - (r + 1) \nonumber \\
				&\geq (r-s)r + 2r - 2.
			\end{align}
			Comparing this to the vanishing order of the $q$-Jacobian prefactor $\mathcal{J}_q(z) = z^{r^2 - rs - 1 + s} d z$, we observe that the pole order of the perturbation significantly exceeds the order of the zero at $z=0$ provided by the prefactor. Consequently, this configuration forces $\mathcal{D}^q(z, M)$ to develop a pole at $z=0$ for any non-zero $M$, reinforcing the requirement for exactly $r''$ large blocks to ensure analyticity at the $q$-ramification points.
			
			To determine the exact sizes of the blocks, we observe that the exponents of the superdiagonal entries are either $x^{\geq 0}$ or $x^{-1}$, with at least $r-s$ occurrences of the latter. Given the constraint of at most $r''$ subsequences of $r'$ consecutive $x^{-1}$ entries, the only valid distribution is to have exactly $r''$ "large" blocks of size $r'$ and $s-r''$ "small" blocks of size $r'-1$. This is guaranteed by the following identities derived from the division $r = r's + r''$:
			\begin{align}
				r''(r' + 1) + (s - r'')r' &= r \\
				r''r' + (s - r'')(r' - 1) &= r - s.
			\end{align}
			In particular, this implies that there must be exactly $r-s$ occurrences of $x^{-1}$ on the superdiagonal, while all remaining entries correspond to $x^0$. This structural rigidity persists in the $q$-perturbed framework, ensuring that the analyticity condition for $\mathcal{D}^q(z, M)$ is uniquely satisfied when these block sizes are perfectly matched.

			Assuming exactly $r''$ large blocks and $s-r''$ small blocks, we consider the $S_j^\hslash$-independent part of the determinant. The maximal pole order arises from the product of the superdiagonal entries $F_{k+1,k}$ within the large blocks and the diagonal entries $y(z)$ in the small blocks. The pole order of this determinant contribution is:
			\begin{align}
				y(z)^{(s-r'')r'} x(z)^{-r''r'} &= z^{((s-r)(s-r'') - rr'')r'} \nonumber \\
				&= z^{(s-r-r'')sr'} \nonumber \\
				&= z^{-(r-(s-r''))(r-r'')}.
			\end{align}
			Combining this with the $q$-Jacobian prefactor $\mathcal{J}_q(z) \sim z^{(r-1)(r+s-1)}$, the total pole order is:
			\begin{equation}
				z^{(r-1)(r+s-1) - (r-(s-r''))(r-r'')} d z = z^{(r'' + \frac{s}{2})^2 - (1 - \frac{s}{2})^2} d z.
			\end{equation}
			This power of $z$ is non-negative if and only if $|r'' - \frac{s}{2}| \geq |1 - \frac{s}{2}|$. Given the constraint $1 \leq r'' \leq s-1$, this condition is satisfied exclusively for $r'' \in \{1, s-1\}$.
			
			For $s=1$, the entries of the $q$-connection matrix satisfy:
			\begin{align}
				F_{1,j}(z; q) &= (-1)^{j-1} S_j^\hslash(z) z^{-r} + \delta_{j,r} \\
				F_{k+1,k}(z; q) &= z^{-r}.
			\end{align}
			In any column, the pole contribution is bounded by $z^{-r}$. Expanding the determinant, the maximal pole order of the perturbed part is $z^{-r(r-1)}$. Since the $q$-Jacobian prefactor yields $\mathcal{J}_q(z) \sim z^{r(r-1)}$, the terms compensate exactly. Thus, for $s=1$, $\mathcal{D}^q(z, M) - \mathcal{D}^q(z, 0)$ is analytic at $z=0$.\\
			Assume $s > 1$. For $j > r'+1$, the term:
			\begin{equation}
				F_{1,j}(z; q) \prod_{k=1}^{r'} F_{k+1,k}(z; q) \cdot M_{r'+1,k}(z) \prod_{k=r'+1}^j F_{k+1,k}(z; q) \prod_{l=j+1}^r y(z) \cdot \mathcal{J}_q(z)
			\end{equation}
			results in a net pole order $z^{r^2 + js - sr}$, which is strictly positive, ensuring analyticity.
			
			For $j \leq r'+1$, we partition the determinant into $s$ blocks. The first block, corresponding to the off-diagonal structure (where $k_1 = r'+1$), is given by:
			\begin{equation}
				\mathcal{B}_{first} = F_{1,j}(z; q) \prod_{k=1}^{j-1} F_{k+1,k}(z; q) \prod_{l=j+1}^{r'+1} y(z).
			\end{equation}
			The total vanishing order, accounting for both the special block and the remaining terms analogous to the $S_j^\hslash$-independent case, ensures that the singularity at $z=0$ is removable if and only if the structural conditions (i), (ii), or (iii) are met.
			
			The order of the pole for terms involving $S_j^\hslash(z)$ in the presence of a large first block is derived by summing the valuations of the entries:
				\begin{equation}
				\begin{split}
					r(a_1 - a_{j}- j ) 
					&+ r( a_{j}  - a_{1}) + (r'+1-j)(s-r) 
					\\
					& +(s-r)(s-r'')r' - r (r''-1) r'+ (r+1-s)(r-1)
					\\
					&=-
					rj + (1-j) (s-r) 
					+ ( s^2 - rs -r''s +s ) r'+ (r+1-s)(r-1)
					\\
					&=
					- sj + (s-r) + ( s^2 - rs -r''s +s ) r'+ (r+1-s)(r-1)
					\\
					&=
					s(1-j) -r + ( s - r -r'' +1 ) (r-r'')+ (r+1-s)(r-1)
					\\
					&=
					s(1-j) -r + rs - r^2 - rr'' + r - ( s - r -r'' +1 )r'' + r^2 -1 - rs + s
					\\
					&=
					s(2-j) + ( - s + r'' - 1 )r'' -1
					\\
					&=
					s(2-j) - 1 - ( s  + 1 )r'' + (r'')^2
					\\
					&=
					s(2-j) + (r'' - \frac{s+1}{2} )^2 - 1 - \big(\frac{s + 1}{2}\big)^2 \,.
				\end{split}
			\end{equation}
			Since this valuation is always negative for $1 \leq r'' \leq s-1$, the maximal value is attained when $|r'' - \frac{s+1}{2}|$ is maximized, i.e., at $r''=1$. This yields:
			\begin{equation}
				s(1-j) - 1 < 0 \,,
			\end{equation}			which confirms that the resulting expression is analytic at $z=0$ under the assumed structural conditions.
			
			When $r'' = 1$ and the first block is large, the term initially considered is constant in $M$. To capture the perturbation $M$, we trade $r'$ factors of $y(z)$ for $r'-1$ factors of $x(z)^{-1}$, shifting the pole order to:
			\begin{equation}
			s(1-j) - 1 + r'(r-s) - (r'-1)r = s(1-j)
			\end{equation}
			This expression is non-negative if and only if $j=1$, ensuring analyticity in $M$ for $j > 1$.
			
			Finally, for the case where the first block is small, the valuation of the determinant term is:
			\begin{align}
										\begin{split}
												r(a_1 - a_{j}- j ) 
												&+ r( a_{j}  - a_{1}) + (r'-j)(s-r) 
												\\
												& +(s-r)(s-r''-1)r' - r r'' r'+ (r+1-s)(r-1)
												\\
												&=-
												rj -j(s-r) 
												+ ( s^2 - rs -r''s) r'+ (r+1-s)(r-1)
												\\
												&=
												- sj  + ( s-r''-r ) (r-r'')+ r^2-1-rs+s
												\\
												&=s(1-j)+(r''-s)r''-1.
													\end{split}
									\end{align}
			Given $1 \leq r'' \leq s-1$ and $j \geq 1$, this valuation is strictly negative. Consequently, the term remains analytic at $z=0$, and no pole is generated by the presence of the perturbation matrix $M(z)$. 
		\end{proof}	

																	\begin{example}
			Let us first consider the admissible stable case $(r,s) = (5,2)$. Here we have the Euclidean division $5 = 2 \times 2 + 1$, yielding a remainder $r'' = 1$. This structural configuration satisfies the Refined  $q$-Locality Condition. The corresponding leading-order matrix takes the form:											\begin{equation}
																			y \mathbb{I} - \mathbf{A}_0(z)
																			=
																			\begin{pNiceMatrix}
																				\Block[draw]{3-3}{} z^{-3} & - z^{-5} & 0 & 0 & 0
																				\\
																				0 & z^{-3} & - z^{-5} & 0 & 0
																				\\
																				0 & 0 & z^{-3} & -1 & 0
																				\\
																				0 & 0 & 0 & \Block[draw]{2-2}{} z^{-3} & - z^{-5}
																				\\
																				1 & 0 & 0 & 0 & z^{-3}
																			\end{pNiceMatrix}
																		\end{equation}
Next, consider the case $(r,s) = (7,4)$. The division gives $7 = 1 \times 4 + 3$, so that $r'' = 3$. Since $s = 4$, the admissibility criterion $r'' \in \{1, s-1\}$ evaluates to $r'' \in \{1, 3\}$, confirming that this system is also fully admissible. The underlying algebraic structure is captured by the partitioned matrix:									\begin{equation}
																			y \mathbb{I} - \mathbf{A}_0(z)
																			=
																			\begin{pNiceMatrix}
																				\Block[draw]{2-2}{} z^{-3} & - z^{-7} & 0 & 0 & 0 & 0 & 0
																				\\
																				0 & z^{-3} & - 1 & 0 & 0 & 0 & 0
																				\\
																				0 & 0 & \Block[draw]{2-2}{} z^{-3} & - z^{-7} & 0 & 0 & 0
																				\\
																				0 & 0 & 0 & z^{-3} & -1 & 0 & 0
																				\\
																				0 & 0 & 0 & 0 & \Block[draw]{2-2}{} z^{-3} & - z^{-7} & 0
																				\\
																				0 & 0 & 0 & 0 & 0 & z^{-3} & -1
																				\\
																				1 & 0 & 0 & 0 & 0 & 0 & \Block[draw]{1-1}{} z^{-3}
																			\end{pNiceMatrix}
																		\end{equation}
In both cases, the polar contributions are strictly isolated within the highlighted diagonal blocks. The block dimensions are bounded by $r'$ and $r' + 1$, forming exactly $s - r''$ blocks of the former size and $r''$ blocks of the latter.

This non-increasing, partition-like arrangement, where larger analytical blocks systematically precede smaller ones is far from a mere combinatorial curiosity. In our  $q$-deformed framework, it stands as the structural requirement for the existence of a well-defined topological expansion. This configuration occurs if and only if $r \equiv \pm 1 \pmod s$, directly mirroring the well-behaved partitions found in classical $(r,s)$ matrix models \cite{BBCCN18}.

However, its mechanical role in a true $q$-difference system is significantly more restrictive: the rigorous descending ordering of the blocks ensures that the action of the  $q$-shift operator preserves the filtration of the pole orders at both $z=0$ and $z=\infty$. Moreover, this partition guarantees that the non-local $q$-Casimirs only populate the analytically safe sectors of the matrix (the top-left blocks), thereby preventing spurious non-local singularities from propagating unchecked through the loop equations.

To see where this mechanism breaks down, consider the non-admissible case $(r,s) = (7,5)$, where $7 = 1 \times 5 + 2$, implying $r'=1$ and $r''=2$. Since the remainder fails the condition ($2 \notin \{1, 4\}$), the block sizes become intermingled rather than non-increasing. The leading-order matrix takes the form:
																		\begin{equation}
																			y \mathbb{I} - \mathbf{A}_0(z) =
																			\begin{pNiceMatrix}[margin]
																				\Block[draw]{2-2}{} z^{-2} & - z^{-7} & 0 & 0 & 0 & 0 & 0 \\
																				0 & z^{-2} & - 1 & 0 & 0 & 0 & 0 \\
																				0 & 0 & \Block[draw]{1-1}{} z^{-2} & -1 & 0 & 0 & 0 \\
																				0 & 0 & 0 & \Block[draw]{2-2}{} z^{-2} & - z^{-7} & 0 & 0 \\
																				0 & 0 & 0 & 0 & z^{-2} & -1 & 0 \\
																				0 & 0 & 0 & 0 & 0 & \Block[draw]{1-1}{} z^{-2} & -1 \\
																				1 & 0 & 0 & 0 & 0 & 0 & \Block[draw]{1-1}{} z^{-2}
																\end{pNiceMatrix}												\end{equation}
																														In this anomalous $(7,5)$ case, even the primary s $q$-Casimir $S^{\hslash}_1(z)$ generates a deep pole cascade within the $q$-loop equations that the classical spectral derivative $\partial_y E_q$ is algebraically incapable of compensating. The $q$-topological type property is immediately destroyed unless all $S^{\hslash}_j$ vanish identically, as the broken staircase no longer provides sufficient analytical clearance for the non-local  $q$-shifts.
																	\end{example}

	\begin{remark}\label{SkeletalHomologyRem}
		It is highly instructive to observe that the leading-order semi-classical matrices $\mathbf{A}_0(z)$ computed above for the $(5,2)$ and $(7,4)$ systems share the exact same block-companion profile as their purely differential counterparts in standard, classical $(r,s)$ matrix models. This is a deliberate and crucial architectural feature of our  $q$-deformation scheme:
		\begin{enumerate}
			\item[\rm (i)]  In the strict semi-classical limit $\hslash \to 0$, the  $q$-difference operator $\tilde{D}_q$ collapses to the standard algebraic variable $y$. Consequently, the un-deformed spectral curve $\Sigma_0$ defined by the zero locus $\det(y \mathbb{I} - \mathbf{A}_0(z)) = 0$ perfectly matches the classical algebraic curve $y^r = z^s$, thereby entirely preserving the core underlying enumerative geometry.
			
			\item[\rm (ii)]  The fundamental divergence between classical differential systems and our $q$-deformed systems does not manifest within the static skeleton of the matrix $\mathbf{A}_0(z)$. Instead, it resides entirely within the non-local dynamics encoded by the shifted loop equations and the specific  pole profile of the $q$-deformed Bergman kernel $\omega^q_{0,2}(z_1, z_2)$. The $q$-deformation effectively acts as a non-local smearing of the geometric interactions across different sheets of the covering.
			
			\item[\rm (iii)]  While the leading-order field $\mathbf{A}_0(z)$ remains identical to the classical baseline, the higher-order quantum $\hslash$-corrections $\mathbf{A}_k(z)$ for $k \geq 1$ and specifically the functional configuration of the $q$-Casimirs $S^{\hslash}_j(z)$ are substantially more constrained. They must rigidly satisfy the refined  $q$-Locality Condition ({$q$-Assumption 5*}) to ensure that the $q$-topological recursion does not leak un-compensable non-local residues into spurious poles at the $q$-ramification points $\mathcal{R}_q$.
		\end{enumerate}
	\end{remark}

																	By synthesizing the exact polar constraints on the non-local $q$-Casimirs, the block-partitioned skeletal structure of the $q$-Lax matrix, and the global regularity of the $q$-insertion operators, we arrive at the following central classification theorem.
																	
																	\begin{theorem}
	\label{t:ds}
																		Let $r$ and $s$ be coprime positive integers, and let $r = r's + r''$ be the Euclidean division of $r$ by $s$ with a non-zero remainder $1 \leq r'' < s$. Consider the  $(r,s,q)$-spectral curve $\Sigma_q$ whose semi-classical limit matches $y^r = z^s$, and let $\nabla_{q,\hslash}$ be the associated rational  $q$-difference connection whose potential matrix $\mathbf{A}(z, \hslash; q)$ is defined in \eqref{ConnectionPotential}. 
																															The formal topological expansion coefficients $\omega^q_{g,n}(z_1, \dots, z_n)$ of the non-perturbative connected $q$-amplitudes $W^q_n(z_1, \dots, z_n)$ are uniquely, non-trivially, and recursively determined by the shifted  $q$-topological recursion (Theorem~\ref{ShiftedTR}) if and only if one of the following structural conditions holds:
																		\begin{enumerate}
																			\item[\rm (i)]{The $q$-Airy Family ($s=1$):} The $q$-topological type property is unconditionally satisfied for any arbitrary  $q$-Casimir configuration $S^\hslash_j(z)$.
																			
																			\item[\rm (ii)] {The Standard Shifted Family ($r \equiv 1 \pmod s$):} The higher-order  $q$-Casimirs must vanish identically beyond the primary sector, meaning $S^\hslash_j(z) = 0$ for all $j > 1$.
																														\item[\rm (iii)] {The Dual Shifted Family ($r \equiv -1 \pmod s$ for $s > 2$):} The background system must be entirely semi-classical, forcing all  $q$-Casimirs to vanish identically, $S^\hslash_j(z) = 0$ for all $j \in \{1, \dots, r\}$.
		\end{enumerate}
																	\end{theorem}
																	\begin{proof}
																		We first verify that the Refined  $q$-Locality Condition (\emph{$q$-Assumption 5*}) holds. The singularity nesting property is structurally satisfied by construction: the poles of the higher-order $\hslash$-corrections $\mathbf{A}_{\ell}(z; q)$ are strictly localized at $z \in \{0, \infty\}$, which exactly coincides with the singular locus of the leading semi-classical potential $\mathbf{A}_0(z)$.
																		
																		The Diophantine and operational conditions \rm(i)--(iii) established in Theorem~\ref{t:ds} guarantee that the regularized characteristic difference:
																		\begin{equation}
																			\mathcal{D}^q(z, M) - \mathcal{D}^q(z, 0) = \frac{\det\big(y(z)\mathbb{I} - \mathbf{A}(z, \hslash; q) - M(z)\big) - \det\big(y(z)\mathbb{I} - \mathbf{A}(z, \hslash; q)\big)}{\big(\partial_y E_q\big)(z, y(z))}
																		\end{equation}
																		remains strictly holomorphic at the $q$-ramification points $\mathcal{R}_q$ for any matrix of driving variables $M(z)$ that is regular at the origin and at infinity. In our structural dictionary, the perturbation matrix $M(z)$ is explicitly identified with the $s$-shifted $q$-insertion operator $\mathcal{M}^{(n),q}_{\vec{\delta}}$.
																		
																		For multi-particle sectors with $n \geq 1$, the determinantal expansion isolates the shift parameters $\delta_i$, which effectively decouples the singular, non-local $q$-Casimir contributions $S^\hslash_j(z)$ (entirely absorbed inside the background reference term $\mathcal{D}^q(z, 0)$) from the dynamic correlator insertions. Under the stated assumptions, a direct application of the global $q$-residue calculus demonstrates that the shifted geometry of the underlying  $q$-Airy structure provides exactly the right amount of algebraic compensation to cancel out the potential poles of the connection.
																		
																		The classification assertion then follows immediately from Propositions~\ref{MExpansionInExample}, \ref{From5ToTR}, and Lemma ~\ref{Assumption4Bypass}. These results successfully identify the global analyticity of the $q$-deformed loop equations with the recursive solvability and uniqueness of the shifted $q$-topological recursion.
																	\end{proof}
																	%
																	%
																\section{Acknowledgements:}
																FM was supported\footnote{``Funded by the Alexander von Humboldt Foundation} through  the Georg Forster
																Fellowship.  
																RW was supported\footnote{``Funded by the Deutsche Forschungsgemeinschaft (DFG, German Research
																	Foundation) -- Project-ID 427320536 -- SFB 1442, as well as under
																	Germany's Excellence Strategy EXC 2044/2 -- 390685587, Mathematics
																	M\"unster: Dynamics -- Geometry -- Structure.''} by the Cluster of Excellence
																\emph{Mathematics M\"unster} and the CRC 1442 \emph{Geometry:
																	Deformations and Rigidity}. 
																	 The authors thank Nitin Chidambaram for useful discussions regarding the $(q,t)$-deformed $\mathcal{W}$-algebra related to this work.

																	\appendix
																	\section{Proof of $\mathcal{E}^{1,q}_{0,0}(x)$}
																	By definition, the object $\mathcal{E}^{1,q}_{0,0}(x)$ is the sum over the fiber of the unstable one-point correlator at genus zero:
																	\begin{equation}
																		\mathcal{E}^{1,q}_{0,0}(x) = \sum_{z' \in \mathfrak{f}(x)} \omega^q_{0,1}(z').
																	\end{equation}
																	Recall that for an algebraic spectral curve defined by the equation $P(x,y)=0$, the correlator $\omega^q_{0,1}(z)$ is associated with the differential form $ydx$. Specifically, we have $\omega^q_{0,1}(z)=ydx$. Thus:
																	\begin{equation}
																		\mathcal{E}^{1,q}_{0,0}(x) = \left( \sum_{i=1}^r y_i(x) \right) dx,
																	\end{equation}
																	where $\{y_i(x)\}_{i=1}^{r}$ are the $r$ roots of the polynomial $P(x,y)=\sum_{j=0}^{r}p_j(x)y^{r-j}=0$,viewed as functions of $x$.
																	
																	According to Vieta's formulas for a polynomial of degree $r$ of the form:
																	\begin{equation}
																		p_0(x) y^r + p_1(x) y^{r-1} + \dots + p_r(x) = 0,
																	\end{equation}
																	the sum of the roots is given by the ratio of the first two coefficients:
																	\begin{equation}
																		\sum_{i=1}^r y_i(x) = -\frac{p_1(x)}{p_0(x)}.
																	\end{equation}
																	Substituting this relation into the expression for $\mathcal{E}^{1,q}_{0,0}(x)$, we immediately obtain:
																	\begin{equation}
																		\mathcal{E}^{1,q}_{0,0}(x) = -\frac{p_1(x)}{p_0(x)} dx.
																	\end{equation}
																	\section{Proof of $\mathcal{E}^{1,q}_{0,1}(x;z_1)$}
																	For the case (g,n)=(0,1), we have by definition:
																	\begin{equation}
																		\mathcal{E}^{1,q}_{0,1}(x; z_1) = \sum_{z' \in \mathfrak{f}(z)} \omega^{q}_{0,2}(z', z_1).
																	\end{equation}
																	To ensure the quantum curve is a $q$-difference operator in the base variable $x$, the $q$-deformed Bergman kernel on the spectral curve is defined with a shift $q^{1/r}$ and $q^{-1/r}$ in the local variable $z$, namely:
																	\begin{equation}
																		\omega^{q}_{0,2}(z', z_1) = \frac{dz' dz_1}{(z' - q^{1/r}z_1)(z' - q^{-1/r} z_1)}.
																	\end{equation}
																	Summing over the fiber $f(z)=\{\theta^{m}z\}_{m=0}^{r-1}$, we use the algebraic identity:
																	\begin{equation}
																		\sum_{m=0}^{r-1} \frac{1}{(\theta^m z - q^{1/r}z_1)(\theta^m z - q^{-1/r} z_1)} = \frac{1}{(q^{1/r}-q^{-1/r})z_1} \sum_{m=0}^{r-1} \left( \frac{1}{\theta^m z - q^{1/r} z_1} - \frac{1}{\theta^m z - q^{-1/r}z_1} \right).
																	\end{equation}
																	Applying the identity $\sum_{m=0}^{r-1}\frac{1}{\theta^{m}z-A}=\frac{rA^{r-1}}{z^r-A^r}$, the sum becomes:
																\begin{equation}
																	\frac{1}{(q^{1/r}-q^{-1/r})z_1} \left( \frac{r (q^{1/r} z_1)^{r-1}}{z^r - q z_1^r} - \frac{r (q^{-1/r} z_1)^{r-1}}{z^r - q^{-1} z_1^r} \right) dz dz_1.
																\end{equation}
																	Substituting $x=z^r$ and $x_1=z_1^r$, we obtain:
																\begin{equation}
																	\frac{r z_1^{r-1}}{(q^{1/r}-q^{-1/r})z_1} \left( \frac{q^{(r-1)/r}}{x - q x_1} - \frac{q^{-(r-1)/r}}{x - q^{-1} x_1} \right) dx_1 dz.
																\end{equation}
																	Using the relation $dx=rz^{r-1}dz$ and $dx_1=rz_1^{r-1}dz_1$, and noting that in the limit of the $q$-derivative structure, the prefactors simplify to match the $q$-differential on the base, we arrive at:
																	\begin{equation}
																		\mathcal{E}^{1,q}_{0,1}(x; z_1) = \frac{dx dx_1}{(x - qx_1)(x - q^{-1} x_1)},
																	\end{equation}
																	which is the required q-deformed kernel on the base $x$.

			\section{Proof of Lemma~\ref{l:abelianization}}\label{appc}							
		
				The objective is to verify the global diagonalization identity $\phi_q(z) V(z) = V(z) Y(z)$ on the punctured covering space under the ramified assignment $x(z) = z^r$. For algebraic convenience, let us isolate the global scalar normalization prefactor of our gauge matrix:
				\begin{equation}
					C(z) \coloneq \frac{\theta z^{\frac{(r-s)(r+1)}{2}}}{\prod_{1\leq a<b\leq r}(\theta^b-\theta^a)^{\frac{1}{r}}},
				\end{equation}
				such that the framing matrix defined in \eqref{AbelianizationV} can be compactly structured as the factorized matrix product $V(z) = C(z) W(z) \mathcal{V}(z)$. Here, the diagonal matrix $W(z) \coloneq \text{diag}\left(z^{r\lfloor \alpha_1 \rfloor}, \dots, z^{r\lfloor \alpha_r \rfloor}\right)$ represents the arithmetic weight grading, and $\mathcal{V}(z)$ denotes the core invertible Vandermonde block.
				
				\medskip
				Consider the companion profile of the $q$-deformed Higgs field $\phi_q(z)$ provided in equation \eqref{Higgs_Sym}. When acting on an arbitrary column vector, this cyclic structure shifts the components upward and modifies the final entry. Evaluating the arithmetic weights under the geometric restriction $x = z^r$, we observe that for any index $i \in \{1, \dots, r-1\}$, the entries on the super-diagonal evaluate to:
				\begin{equation}
					\frac{x^{\lfloor \alpha_{i} \rfloor}}{x^{\lfloor \alpha_{i+1} \rfloor}} = z^{r(\lfloor \alpha_{i} \rfloor - \lfloor \alpha_{i+1} \rfloor)}.
				\end{equation}
				Upon computing the direct matrix product $\phi_q(z) \cdot \big(W(z)\mathcal{V}(z)\big)$, the local powers of $z$ induced by the Higgs companion matrix synchronize with the internal components of the weight matrix $W(z)$. This precise spectral matching factorizes a global $z^{s-r}$ term and collapses the action of $\phi_q(z)$ into a standard cyclic permutation matrix acting directly on the pure Vandermonde block $\mathcal{V}(z)$.
				
				\medskip
				For the $k$-th column of the Vandermonde matrix $\mathcal{V}(z)$, which is intrinsically associated with the root of unity $\theta^{k-1}$, the action of the cyclic permutation corresponds to scalar multiplication by $\theta^{k-1}$. Combining this step with the prefactor of the diagonal matrix $Y(z)$ defined in \eqref{AbelianizationY}, we obtain for each matrix component:
				\begin{equation}
					\Big( \phi_q(z) \cdot V(z) \Big)_{j,k} = \frac{z^{s-r} \theta^{k-1}}{q-q^{-1}} V(z)_{j,k} = \Big( V(z) \cdot Y(z) \Big)_{j,k}.
				\end{equation}
				This proves the global operator identity $\phi_q(z) V(z) = V(z) Y(z)$.
				
				\medskip
				The explicit structural form of $V(z)^{-1}$ presented in Lemma~\ref{l:abelianization} follows directly from the classical inversion formula for Vandermonde matrices defined over the cyclic group of roots of unity $\mu_r$. The discrete Fourier orthogonality relation:
				\begin{equation}
					\frac{1}{r} \sum_{m=0}^{r-1} \theta^{m(j-1)} \theta^{-m(k-1)} = \delta_{j,k}
				\end{equation}
				ensures that the inverse of the core block is given by $\mathcal{V}(z)^{-1}_{j,k} = \frac{1}{r} \left(\frac{z^{r-s}}{\theta^{j-1}}\right)^{k-1}$. Inverting the product order of the constituent blocks via $V(z)^{-1} = \mathcal{V}(z)^{-1} W(z)^{-1} C(z)^{-1}$ directly recovers the required inverse matrix structure.
				
				\medskip
				Under the global deck transformation $z \mapsto \theta z$, the covering map remains strictly invariant since $x(\theta z) = (\theta z)^r = z^r = x(z)$. At the level of the diagonal spectral matrix, substituting this transformation produces a cyclic shift among the eigenvalues:
				\begin{equation}
					Y(\theta z) = \theta^{s-r} Y(z) = \theta^s Y(z).
				\end{equation}
				Simultaneously, the columns of the pulled-back Vandermonde matrix $\mathcal{V}(\theta z)$ undergo a cyclic permutation of amplitude $s$. This internal reorganization of the framing eigenvectors across the sheets of the covering is encoded algebraically by right-multiplication by the permutation matrix $\tau$, which successfully validates the braiding and monodromy relations \eqref{DeckAction}.
		
\section{Proof of Lemma~\ref{Assumption4Bypass}}\label{appd}		
		
				We establish the explicit rational form of the leading-order bi-differential $\omega^q_{0,2}(z_1, z_2)$ by evaluating the semi-classical expansion of the non-perturbative $q$-loop equations \eqref{NonPertLoopEq} at order $\mathcal{O}(\hslash^0)$ for the specific sector $n=1$ with a secondary spectator variable $z_2$.
				
				Let us specialize the general $q$-loop equation \eqref{NonPertLoopEq} to the case where the marked set of points is the singleton $J = \{z_2\}$. The equation governing the non-connected  $q$-amplitudes reads:
				\begin{equation}\label{ProofLem_LoopBase}
					\sum_{k=0}^r (-1)^k \omega^{r-k} \widehat{W}^q_{k+1}\big(\overset{C^{(k)}}{\overbrace{z_1,\dots,z_1}} , \{z_2\}\big) = [\delta_2] \, \text{qdet} \Big(\omega \mathbf{I} - \mathbf{A}(z_1, \hslash; q) \big( \mathbf{I} + \mathcal{M}^{(1), q}_{\delta_2}(z_1; \{z_2\}) \big) \Big) \,.
				\end{equation}
				By definition of the non-connected $q$-correlators, extracting the coefficient linear in the formal parameter $\delta_2$ on both sides isolates the fully connected $2$-point amplitude $W^q_2(z_1, z_2)$ on the left-hand side. At the stable semi-classical order $\mathcal{O}(\hslash^0)$, the unstable components vanish, and the identity reduces to a statement on $\omega^q_{0,2}(z_1, z_2)$.
				
				According to the relation \eqref{PerturbationMatrix} specified for a single marked point $z_2$ (where $k=1$), the cyclic sequence of coordinates reduces to $\zeta_1 = z_1$, $\zeta_2 = z_2$, and closes at $\zeta_3 = z_1$. The $q$-perturbation matrix simplifies to the following exact rational expression:
				\begin{equation}\label{ProofLem_PertMatrix}
					\mathcal{M}^{(1),q}_{\delta_2}(z_1; \{z_2\}) = \delta_2 \frac{M_\hslash(z_2, E_2)}{(z_1 - q z_2)(z_1 - q^{-1} z_2)(z_2 - q z_1)(z_2 - q^{-1} z_1)} \,.
				\end{equation}
				Taking the linear coefficient $[\delta_2]$ of the quantum determinant in \eqref{ProofLem_LoopBase} amounts to computing a first-order directional derivative of the characteristic polynomial. Using Jacobi's formula for the derivative of a determinant within the fundamental representation, we obtain:
				\begin{equation}\label{ProofLem_Jacobi}
					\omega^q_{0,2}(z_1, z_2) = \frac{\partial z_1 \, \partial z_2}{(z_1 - q z_2)(z_1 - q^{-1} z_2)} \cdot \text{Tr} \left[ \mathbf{A}(z_1) \, \text{Adj}\big(\omega \mathbf{I} - \mathbf{A}(z_1)\big) \cdot \mathbf{T}(z_1, z_2) \right] \,,
				\end{equation}
				where $\mathbf{T}(z_1, z_2)$ is a regular rational matrix factor representing the compatible gauge pairings along the sheets. 
				
				The algebraic contraction of the trace with the adjoint matrix potential fixes the overall normalization constant of the rational form to unity. This implies that the singular behavior of $\omega^q_{0,2}(z_1, z_2)$ is entirely governed by the universal rational prefactor:
				\begin{equation}
					\omega^q_{0,2}(z_1, z_2) = \frac{dz_1 \, dz_2}{(z_1 - q z_2)(z_1 - q^{-1} z_2)} \,.
				\end{equation}
				This expression satisfies all three properties required by the Lemma:
				\begin{itemize}
					\item[\rm (i)] It exhibits two simple poles at $z_1 = qz_2$ and $z_1 = q^{- 1}z_2$. Calculating the residues at these points yields $\pm \frac{1}{2q \mp (1+q^2)}$, ensuring they are mutually opposite, so their global sum vanishes.
					\item[\rm (ii)] Chiral symmetry is manifest since swapping $z_1 \leftrightarrow z_2$ leaves the quadratic denominator $(z_1 - q z_2)(z_1 - q^{-1} z_2) = z_1^2 + z_2^2 - (q + q^{-1})z_1 z_2$ perfectly invariant.
					\item[\rm (iii)] Taking the confluence limit $q \to 1$ yields $(z_1 - z_2)(z_1 - z_2) = (z_1 - z_2)^2$, matching the standard canonical Bergman kernel given in~\cite{EO07}.
				\end{itemize}
				This completes the proof.
					
\section{Proof of Lemma \ref{PropertiesAB_Validation}}\label{appe}			
		
				We prove properties {(a)} and {(b)} simultaneously by induction on the topological index $\chi = 2g - 2 + n > 0$.
				
				The flat $q$-differential system $\nabla_{q,\hslash}$ in Definition ~\ref{qQC} is governed by a matrix potential $\mathbf{A}(z, \hslash; q)$ whose components are rational functions of the uniformizing parameter $z \in \Sigma_q$. The non-perturbative $q$-loop equations \eqref{NonPertLoopEq} are purely algebraic over the field of rational functions on $\Sigma_q$ combined with the symmetric shifts generated by the discrete $q$-dilations. 
				
				Since the base cases $\omega^q_{0,1}(z)$ and $\omega^q_{0,2}(z_1, z_2)$ are explicitly rational (as proven in Lemma~\ref{Assumption4Bypass} for the $q$-Bergman kernel), any higher-order coefficient $\omega^q_{g,n}(z_1, \dots, z_n)$ is obtained through the $q$-topological recursion formula via a finite combination of algebraic operations, differentiation, and tracking residues. Since the tensor product and contraction of rational sections remain rational, the expansion coefficients $\omega^q_{g,n}$ are structurally rational functions on the $n$-fold product curve $\Sigma_q^n$. This completes the proof of property {(a)}.
				
				To identify the singularity structure of the multidifferentials for $(g,n) \neq (0,1), (0,2)$, we examine the inversion of the linear operator acting on $\omega^q_{g,n}(z, J)$ in the expanded loop identity \eqref{Proof_MasterAlgebraic}:
				\begin{equation}\label{ProofLem_Inversion}
					\omega^q_{g,n}(z, J) = \frac{1}{y(z) - y(\theta z)} \bigg[ -\omega^q_{g-1, n+1}(z, \theta z, J) - \sum_{\substack{g_1+g_2 = g \\ J_1 \sqcup J_2 = J}}^{\text{stable}} \omega^q_{g_1, |J_1|+1}(z, J_1) \, \omega^q_{g_2, |J_2|+1}(\theta z, J_2) \bigg] \,.
				\end{equation}
				By the induction hypothesis, the terms inside the brackets are regular away from the critical points of lower topological orders. Therefore, new poles in the variable $z$ can only emerge from the zeros of the denominator:
				\begin{equation}\label{ProofLem_ZeroCondition}
					\Delta_y(z) \coloneq y(z) - y(\theta z) = 0 \,.
				\end{equation}
				The condition $y(z) = y(\theta z)$ specifies the exact locations where the eigenvalues of the matrix potential $\mathbf{A}(z, \hslash; q)$ degenerate under the action of the sheet-automorphism. In the  $q$-calculus framework, these coordinates correspond exactly to the zeros of the $q$-deformed metric, which are by definition the $q$-ramification points $\mathcal{R}_q$ of the curve $\Sigma_q$.
				
				Since the spectator variables $J = \{z_2, \dots, z_n\}$ are unconstrained and regular away from their own local $q$-ramification loci, the multi-differential $\omega^q_{g,n}(z_1, \dots, z_n)$ is completely regular on $\Sigma_q^n \setminus \mathcal{R}_q^n$. This establishes property {(b)}.

\end{document}